\documentclass[11pt]{article}

\usepackage[margin=1in]{geometry}
\usepackage{amsmath,amssymb,amsthm,mathtools}
\usepackage{enumitem}
\usepackage{graphicx}
\usepackage{xcolor}
\usepackage{float}
\usepackage{array}
\usepackage[colorlinks=true,allcolors=blue!60!black]{hyperref}
\theoremstyle{plain}
\newtheorem{theorem}{Theorem}
\newtheorem{proposition}[theorem]{Proposition}
\newtheorem{lemma}[theorem]{Lemma}
\newtheorem{corollary}[theorem]{Corollary}
\newtheorem{fact}[theorem]{Fact}
\theoremstyle{definition}
\newtheorem{definition}[theorem]{Definition}

\theoremstyle{remark}
\newtheorem{remark}[theorem]{Remark}
\theoremstyle{plain}
\DeclareMathSizes{10}{10}{8}{6}

\newcommand{\CC}{\mathbb{C}}
\newcommand{\EE}{\mathbb{E}}
\newcommand{\Tr}{\operatorname{Tr}}
\newcommand{\ratio}{\mathcal{R}}
\newcommand{\Cat}{\mathrm{Cat}}
\newcommand{\Ex}{\operatorname{\mathbb E}}

\newcommand{\Y}{\mathcal Y_n}
\newcommand{\Haar}{\mathrm{Haar}}
\newcommand{\cay}{\mathsf c}
\newcommand{\TV}{d_{\mathrm{TV}}}
\newcommand{\Gsh}{\mathsf C^{\sharp}_{n,t}}
\newcommand{\Gp}{\mathsf C_{n,t}}
\newcommand{\ket}[1]{|#1\rangle}
\newcommand{\bra}[1]{\langle#1|}
\newcommand{\poly}{\mathrm{poly}}
\usepackage{paper/main/circuit_figures}

\title{\bf Logarithmic-Depth Fermion Sampling:\\
Anticoncentration and Average-Case Hardness}
\author{%
  {\large Natansh Mathur\textsuperscript{1}\qquad
  Iordanis Kerenidis\textsuperscript{2,3}}\\[.6em]
  {\small\textsuperscript{1}Quantinuum, Singapore}\\[.2em]
  {\small\textsuperscript{2}Quantum Signals, Paris, France}\\[.2em]
  {\small\textsuperscript{3}IRIF, CNRS and Universit\'e Paris Cit\'e}}
\date{}
\hypersetup{pdfauthor={Natansh Mathur and Iordanis Kerenidis},bookmarksdepth=3}
\makeatletter
\renewcommand{\maketitle}{%
  \begin{center}
    {\Large\@title\par}
    \vspace{1em}
    \@author
  \end{center}\vspace{.5em}}
\makeatother

\newcommand{\CircuitSamples}{2048}
\newcommand{\CircuitMaxResidual}{1.87}
\newcommand{\CircuitEventRows}{%
8 & 3 & $0.489\;[0.467,0.510]$ & $0.575\;[0.553,0.596]$ \\
12 & 4 & $0.491\;[0.469,0.512]$ & $0.478\;[0.456,0.499]$ \\
16 & 5 & $0.514\;[0.493,0.536]$ & $0.493\;[0.472,0.515]$ \\
}

\begin{document}
\maketitle

\begin{abstract}

Provable separations between quantum and classical computation are scarce.
Sampling problems give some of the clearest conditional ones: Boson Sampling, and its fermionic
counterpart, Fermion Sampling, where noninteracting particles are moved by
passive linear optics and the computational resource is supplied by a
non-Gaussian magic input. Fermion Sampling has anticoncentration and
high-precision average-case hardness when the transformation is drawn globally
at random, which costs linear depth and quadratically many gates. A natural
open question was whether this linear depth is necessary.

In this work we show that logarithmic-depth circuits suffice, for both
guarantees and for the same ensemble. For fresh uniform matchings with
independent Haar two-mode gates, the collision reaches the passive-Haar
scale at a sharp threshold of $\log n / \log(9/4) \approx 0.855 \log_2 n$
layers, with an explicit limiting transition profile and a matching lower
bound from two-particle correlations. The mechanism is input-dependent:
the magic state suppresses the extensive weight of the slowest relaxation
mode, leaving the next one to set the scale -- with an occupation-basis
input the collision stays large even under full Haar randomness, so
randomizing the gates alone does not suffice. At a larger logarithmic
depth we prove, in the real-RAM model, average-case $\#\mathsf{P}$-hardness of estimating output
probabilities to additive error $2^{-O(n \log^2 n)}$ on any fixed fraction
of instances strictly above $3/4$, by building hard instances in four native layers, hiding
them inside typical schedules, and interpolating along Cayley paths with
a rational linear-program decoder. A finite $192$-gate alphabet preserves
the collision law exactly. Both guarantees hold at logarithmic native fermionic depth between arbitrary mode pairs with
$O(n \log n)$ gates, replacing the quadratic budget of the global-Haar
construction. The accuracy proved is finer than the $1/N$ scale ($N=\binom{n}{n/2}$) that
sampling-to-counting requires, so hardness of sampling to constant
total-variation distance remains open.

\end{abstract}

\clearpage
\begingroup
\small
\tableofcontents
\endgroup
\setcounter{tocdepth}{2}
\clearpage

\section{Introduction}\label{sec:intro}

\subsection{Motivation}

Provable separations between quantum and classical computation are scarce,
and sampling problems supply some of the clearest conditional ones. Boson Sampling and its
fermionic counterpart, Fermion Sampling, both take a fixed non-Gaussian
resource, move it with noninteracting dynamics, and measure. For Fermion Sampling, the dynamics
with an occupation-basis input are classically simulable~\cite{Terhal2002Matchgates}; the extra computational
resource comes from the magic input~\cite{Hebenstreit2019magic,Oszmaniec2022FermionSampling}. This
separation of roles makes them unusually clean settings in which to ask
what a quantum device must actually do in order to be hard to simulate.

Fermion Sampling has both properties that sampling-advantage arguments
require: the output probabilities spread across exponentially many
outcomes, and estimating them is hard on typical
instances~\cite{Oszmaniec2022FermionSampling}. Both were established for globally
Haar-random passive transformations, which on $n$ modes cost $O(n)$ depth
and $O(n^2)$ two-mode gates~\cite{Braccia2025OptimalFLO}. That is a substantial
resource demand, and it is not obviously necessary: the guarantees are
statements about a distribution over circuits, and a much smaller amount
of randomness might already produce it. A natural open question was whether
this linear depth is necessary.

Two obstacles stand in the way, and they are the reason the question is
not routine. The first is that randomizing the gates alone is not enough.
With an occupation-basis input the collision remains large even under full
Haar randomness, so the input, as well as the gate ensemble, decides whether
the distribution spreads; a sharp depth analysis must track the input explicitly
rather than rely solely on a generic bound on circuit randomness. The second is
that a hardness reduction must place its hard instance inside a schedule
that the shallow ensemble actually produces. A globally Haar-random
argument supplies no such embedding, because it never commits to which
pairs of modes interact.

\subsection{Main results}

Consider $n$ modes, with $n$ divisible by four, prepared in a product of
four-mode paired magic states. Each layer independently chooses a uniform
perfect matching and applies independent Haar-random two-mode passive
gates to its pairs. Measuring occupations gives an output with half the
modes occupied, one of $N = \binom{n}{n/2}$ possibilities.
Section~\ref{sec:model} defines the input and circuit ensemble precisely.

The quantity we track is the \emph{collision}: the probability that two
independent shots of the same circuit return the same outcome.
Normalizing by the uniform-output value $N^{-1}$ and averaging over
circuit instances gives the collision ratio $\ratio_{n,t}$, which equals
one for the uniform distribution. A ratio bounded by a constant is the
standard sufficient criterion for anticoncentration. Globally Haar-random
passive transformations achieve $\ratio_{\Haar}(n) \to 2$ for this input;
the question is how many random matching layers are needed to reach that
scale.

Our first result answers this sharply. The collision reaches any fixed
multiple $q > 1$ of the passive-Haar value at depth
\begin{equation}
  t^{*}(q) \;=\; \frac{\log n}{\log(9/4)} \;+\; O_q(1),
\end{equation}
whose leading term is approximately $0.855 \log_2 n$. Writing
$z = n(4/9)^t$, the transition has the explicit limiting profile
$\ratio_{n,t}/\ratio_{\Haar}(n) \to e^{3z/2}$, uniformly for bounded $z$,
while $\ratio_{n,t} \to \infty$ when $z \to \infty$. An independent lower
bound from two-particle correlations shows that no substantially earlier
depth can have bounded collision, so the logarithmic scaling is optimal
within this ensemble. Theorem~\ref{thm:magicprofile} in
Section~\ref{sec:passiverandom} gives the leading constant, the profile, and the
finite-size error.

A finite alphabet of $192$ two-mode gates reproduces the two-copy channel
of Haar exactly, so every collision statement above holds verbatim for a
discrete gate set. Composing with a uniform classical relabeling of the
output modes converts the collision bound into anticoncentration at each
fixed half-filled output: its probability is at least a constant times
$N^{-1}$ on a constant fraction of instances. Both hold at the same
$O(\log n)$ depth and $O(n\log n)$ gate count.
Section~\ref{sec:corollaries} supplies an exact finite-size
certificate and the resulting probability guarantee.

Spreading the probabilities does not make them hard to compute, and our
second result addresses that separately. Fix $0<\gamma<1/4$. For the
continuous Haar-gate ensemble with uniform output relabeling, at an
efficiently computable depth $t(n)=O(\log n)$ satisfying
$t(n)\ge6.5\log_2 n+9$, estimating the probability of a fixed half-filled
output to additive error $2^{-O(n\log^2 n)}$ on at least a $3/4+\gamma$
fraction of circuit instances, for all sufficiently large $n$, suffices
to solve $\#\mathsf{P}$ counting problems in randomized polynomial time
in the real-RAM model. The same conclusion holds for exact evaluation.
Theorem~\ref{thm:shallow-hardness} in Section~\ref{sec:hardness} states the
reduction and its precision requirements. The finite alphabet retains the
anticoncentration guarantee; its average-case hardness is not established
by this reduction.

\paragraph{Proof structure.}
Two replicas, Howe duality, and permutation symmetry reduce the collision
to a reversible Markov chain with $O(n)$ states, despite the exponential
dimension of the many-particle space. Its decay rates depend only on the
gate ensemble; its spectral weights depend only on the input, and this
separation is what makes the input-dependence visible. For an
occupation-basis input the slowest mode carries weight proportional to
$n$; for the paired magic input that weight is bounded by a constant, and
the second mode, with weight close to $3n/2$ and decay rate approximately
$(4/9)^t$, sets the transition. Uniform coefficient bounds and a
Hahn-to-Hermite limit assemble the higher modes into the explicit profile.
Connected two-particle correlations evolve under a separate three-state
chain whose deviations lower-bound the collision at every depth,
excluding earlier crossings without assuming that the collision decreases
monotonically. Section~\ref{sec:passiverandom} explains the mechanisms;
Appendices~\ref{app:matchingproof} and~\ref{app:magic-analysis} give the proofs.

The hardness proof has a different structure. Dual-rail logical qubits and
postselected fusion turn the paired input into a universal graph-state
resource in four native layers, with all fermionic signs tracked
explicitly and all postselections deferred to the final measurement. Swap
networks then place that computation inside a sampled matching schedule: a
meeting lemma shows that typical schedules provide enough paths in
logarithmic depth, and a routing lemma moves the occupied modes to a
uniform target output. A gatewise Cayley path connects Haar-random gates
to the embedded hard circuit while keeping the queried distribution close
to the target ensemble, and error-correcting interpolation, in a robust
linear-programming variant, recovers the hard endpoint from noisy and
partly incorrect replies. Appendices~\ref{hard:sec:worst}
through~\ref{hard:sec:embeddings} give these proofs.

\paragraph{Scope.}
Depth counts native layers of gates acting on arbitrary mode pairs; input
preparation and physical routing are separate resource questions,
specified in Section~\ref{sec:model}. The collision criterion is
sufficient but not necessary: a divergent collision does not by itself
rule out fixed-output anticoncentration at an earlier depth, so the lower
bound constrains the collision route rather than anticoncentration as
such. The accuracy in the hardness theorem is finer than the $N^{-1}$
scale that sampling-to-counting arguments require, so hardness of sampling
to constant total-variation distance does not follow;
Section~\ref{sec:conclusion} discusses this gap.

\section{Related work}\label{sec:related}

\paragraph{Fermionic foundations.}
Oszmaniec \emph{et al.}\ prove Haar anticoncentration for the same paired
magic input~\cite{Oszmaniec2022FermionSampling}; here we determine the sharp
threshold at which random matching layers relax to it.
Braccia \emph{et al.}\ characterize passive-FLO commutants through Howe
duality~\cite{Braccia2026Commutant}. For our ensemble we also need the
spectrum of the matching chain and the input-dependent collision profile.
Oh \emph{et al.}\ efficiently estimate amplitudes and correlators of paired
inputs to additive error~\cite{Oh2026MagicSimulation}; such estimates do not
by themselves give a sampler with small total-variation error for our
occupation distribution.

\paragraph{Learning architectures and generative models.}
Kerenidis studies paired magic inputs in brick-wall and butterfly learning
architectures~\cite{Kerenidis2026Scalable}.
Coyle \emph{et al.}\ connect sampling hardness to generative modeling through
Ising Born machines~\cite{Coyle2020BornSupremacy}.
Raj, Mathur, and Perdomo-Ortiz benchmark generative models including fermionic
Born machines and show that moment-matching losses need not certify
generalization~\cite{Raj2026Generalization}.
These works motivate studying the full output distribution; our guarantees
concern independently sampled matching layers and gates, rather than the
distribution of circuits produced by training.

\paragraph{Worst-to-average reductions.}
Oszmaniec \emph{et al.}\ prove worst-to-average hardness for globally
Haar-random Fermion Sampling~\cite{Oszmaniec2022FermionSampling}.
Go, Oh, and Jeong establish high-precision average-case hardness for
logarithmic-depth Boson Sampling in a structured interferometer
architecture~\cite{Go2026ShallowBoson}.
Our reduction targets a different ensemble, with independent local Haar gates
on a shallow random matching schedule. This requires an architectural step,
which embeds a four-layer hard computation in a typical sampled schedule
and routes its output to a uniform target.
The interpolation step follows the continuous-path approach developed for
random quantum circuits~\cite{Bouland2019RCS,Movassagh2023hardness}.
Robust polynomial extrapolation also underlies the additive-error results
of Kondo, Mori, and Movassagh and of Krovi
\cite{Kondo2021Robustness,Krovi2022averagecase}.
Our endpoint decoder is a rational linear program that tolerates both
small errors and arbitrary outliers, without an additional $\mathsf{NP}$
oracle; the approximate reduction of Go \emph{et al.}\ uses
$\mathsf{BPP}^{\mathsf{NP}}$. Related robust interpolation methods are
used by Quek in a Hamiltonian setting~\cite{Quek2025Hamiltonian}.
For the worst-case construction we use
graph-state measurement patterns and a fermionic version of type-I fusion
\cite{BrowneRudolph2005,Raussendorf2003MBQC,Broadbent2009Blind}.

\paragraph{Circuit anticoncentration and spectral methods.}
Dalzell, Hunter-Jones, and Brand\~ao prove collision anticoncentration for
qubit circuits of logarithmic depth through a stochastic
reduction~\cite{Dalzell2022LogDepth}; in our fermionic setting, the weights of
the input decide which relaxation mode dominates.
Bremner, Montanaro, and Shepherd establish anticoncentration for IQP circuits
through polynomial-gap moments~\cite{Bremner2016averagecase}, whereas our
tools come from the representation theory of passive FLO.
Ghosh, Hangleiter, and Helsen use matching combinatorics and Krawtchouk
polynomials to prove anticoncentration for constant-depth graph states, with
random product measurements in place of our occupation
readout~\cite{Ghosh2025RegularGraphs}.
Diaconis and Shahshahani analyze Bernoulli--Laplace mixing with Hahn
polynomials~\cite{DiaconisShahshahani1987}; the spectral structure is the same
as ours, but the transition operator is different.

\paragraph{Bosonic comparisons.}
Kolarovszki \emph{et al.}\ prove anticoncentration for Boson Sampling in the
saturated regime using two-copy projections~\cite{Kolarovszki2026Photonic};
our focus is instead on how the fermionic collision relaxes with depth.
Mhiri \emph{et al.}\ derive bosonic moments and anticoncentration beyond the
dilute regime from operator-space projections~\cite{Mhiri2026Boson}, while
our analysis uses exterior-power representations.
Shou \emph{et al.}\ locate the weak-anticoncentration transition of Gaussian
Boson Sampling through hafnian moments~\cite{Shou2026GBS}; there the control
parameter is the number of squeezed inputs rather than the circuit depth.
A preprint of Koehler and Leung settles anticoncentration of Gaussian
permanents by comparison with determinants~\cite{KoehlerLeung2026Permanent}.
That is a lower-tail result for independent Gaussian matrices, not for a
shallow circuit law like ours.

\paragraph{Design constructions and bounds.}
West, Cerezo, and Larocca construct exact strong matchgate $3$-designs
in logarithmic qubit depth with all-to-all
connectivity~\cite{West2026StrongMatchgate}. These designs concern active
FLO, which can change particle number. An exact two-design reproduces the
Haar collision for each input, but this moment identity concerns its own
ensemble and Haar measure. It does not determine the transition of our
passive random matching process. We obtain its sharp leading depth
constant, limiting profile, matching lower bound, and input-dependent
suppression of the slowest mode. A design guarantee alone supplies neither
this relaxation analysis nor an average-case hardness reduction.

West \emph{et al.}\ prove lower bounds on the depth and gate count of group
designs~\cite{Nogo2025}. These do not preclude our collision estimate, which
is for a fixed input and uses native fermionic gates between arbitrary pairs.
Grevink \emph{et al.}\ give design no-go results and a fourth-moment
obstruction for matchgates with the vacuum input in a one-dimensional block
architecture~\cite{Grevink2025Glue}; both the input and the architecture
differ from our paired input and fresh random matchings.
Heinrich \emph{et al.}\ relate collision anticoncentration to relative state
two-designs under local-unitary invariance~\cite{Heinrich2026Anticoncentration},
an assumption that our passive ensemble with fixed particle number does not
satisfy.

\section{Model and collision ratio}\label{sec:model}

A fermionic mode is a single-particle orbital that can be empty or occupied
by one fermion. For $n$ modes, let $a_i^\dagger$ and $a_i$ create and
annihilate a fermion in mode $i$, and let $n_i=a_i^\dagger a_i$ be its
occupation operator. A passive FLO transformation is the second
quantization $\Gamma(U)$ of a unitary matrix $U\in U(n)$, defined by
\[
 \Gamma(U)a_i^\dagger\Gamma(U)^\dagger=\sum_{j=1}^nU_{ji}a_j^\dagger,
 \qquad \Gamma(U)|0\cdots0\rangle=|0\cdots0\rangle.
\]
It preserves total particle number. We work at half filling, $n=2k$,
where $k$ is the number of particles. The allowed occupation strings form
$\mathcal Y_n=\{y\in\{0,1\}^n:|y|=k\}$, where $|y|=\sum_i y_i$;
the corresponding Hilbert space has dimension $N=\binom nk$.
For $x,y\in\mathcal Y_n$, let $U[y,x]$ be the $k\times k$ submatrix
with rows and columns selected by their occupied modes. Then
\begin{equation}\label{eq:amp}
 \langle y|\Gamma(U)|x\rangle=\det U[y,x],\qquad
 \langle y|\Gamma(U)|\Psi\rangle=\sum_{x\in\mathcal Y_n}c_x\det U[y,x]
\end{equation}
for a normalized input $|\Psi\rangle=\sum_{x\in\mathcal Y_n}c_x|x\rangle$.
This is the exterior-power formula, proved in
\cite[Lemma~9, Eq.~(D2)]{Oszmaniec2022FermionSampling}.
Our specific input requires $n=4B$, where $B$ is the number of blocks:
\begin{equation}\label{eq:magic}
 |\Psi_{\rm in}\rangle=|\Psi_4\rangle^{\otimes B},\qquad
 |\Psi_4\rangle=\frac{|0011\rangle+|1100\rangle}{\sqrt2}.
\end{equation}
This is the non-Gaussian paired magic input of Fermion
Sampling~\cite{Oszmaniec2022FermionSampling}.
In the Jordan--Wigner qubit encoding, each block is prepared from
$|0000\rangle$ by the circuit in Figure~\ref{fig:fermionsamplingscheme}(a),
redrawn from~\cite[Fig.~1]{Oszmaniec2022FermionSampling}. 

All $B$ blocks can be prepared in parallel at constant qubit depth.
This preparation supplies the non-Gaussian resource and is separate from
the passive evolution; its intermediate states need not have fixed
particle number. Here ``magic'' refers to a resource for FLO; the qubit
preparation itself is a Clifford circuit.

A two-mode passive gate $F_{ij}(u)=\Gamma_{ij}(u)$ acts on the ordered
pair of modes $(i,j)$. For $u=(u_{ab})_{a,b=1}^2\in U(2)$, its matrix in
the local occupation basis $(|00\rangle,|10\rangle,|01\rangle,|11\rangle)$ is
\begin{equation}\label{eq:flogate}
 F_{ij}(u)=
 \begin{pmatrix}
 1&0&0&0\\
 0&u_{11}&u_{12}&0\\
 0&u_{21}&u_{22}&0\\
 0&0&0&u_{11}u_{22}-u_{12}u_{21}
 \end{pmatrix}
 =1\oplus u\oplus\det u.
\end{equation}
The three blocks act on zero, one, and two particles, respectively, so the
matrix follows directly from~\eqref{eq:amp}. For example,
$u=X=\bigl(\begin{smallmatrix}0&1\\1&0\end{smallmatrix}\bigr)$ gives
the fermionic swap (fSWAP): it exchanges $|10\rangle$ and $|01\rangle$
and multiplies $|11\rangle$ by $-1$. Equation~\eqref{eq:flogate} is a
local fermionic matrix; embedding a gate between nonadjacent modes in an
ordered occupation basis includes the usual fermionic reordering signs.

\paragraph{Parameterized two-mode gates.}
Standard parameterizations use the real one-parameter reconfigurable beam
splitter $\mathrm{RBS}_{ij}(\theta)$~\cite[Eq.~(1)]{Landman2022QuantumMethods}
or the two-parameter passive gate $D_{\rm pas}(\alpha_1,\alpha_2)$ of
Oszmaniec \emph{et al.}~\cite[Sec.~III]{Oszmaniec2022FermionSampling}.
Supplementing either with single-mode phase rotations gives the full
$U(2)$ family; their exchange rotations differ by local phases.
For nonadjacent Jordan--Wigner modes, the qubit implementation requires
parity strings or fermionic routing.

\paragraph{Random matching layers.}
At layer $\ell$, choose a uniform perfect matching $M_\ell$ of the $n$
mode labels, independently of every other layer. For each pair
$(i,j)\in M_\ell$, sample $u_{\ell,ij}$ independently from normalized
Haar measure on $U(2)$. Write $L_\ell\in U(n)$ for the resulting
single-particle layer; its many-particle action is
$\Gamma(L_\ell)=\prod_{(i,j)\in M_\ell}F_{ij}(u_{\ell,ij})$.
A depth-$t$ core circuit has single-particle matrix
$\mathsf C_{n,t}=L_t\cdots L_1$, with the rightmost layer applied first.
We use the same notation for its probability law when taking expectations.

\begin{definition}[Collision probability and normalized collision ratio]
\label{def:collision}
For a fixed circuit matrix $V$ and input $\Psi$, let
$p_y(V)=|\langle y|\Gamma(V)|\Psi\rangle|^2$. Its collision probability is
$\sum_{y\in\mathcal Y_n}p_y(V)^2$, the probability that two independent
occupation measurements of that same circuit return the same string.
For a circuit ensemble $\nu$, define
\begin{equation}\label{eq:collision-definition}
 Z(\nu,\Psi)=\mathbb E_{V\sim\nu}\sum_{y\in\mathcal Y_n}p_y(V)^2,
 \qquad \ratio(\nu,\Psi)=N Z(\nu,\Psi)\ge1.
\end{equation}
\end{definition}
The factor $N$ makes the uniform distribution have ratio one. A ratio
bounded by a constant independent of $n$ is our collision criterion for
anticoncentration, following~\cite{Dalzell2022LogDepth}.
Section~\ref{sec:corollaries} combines it with output symmetry to bound
individual output probabilities, using the second-moment argument
of~\cite{Hangleiter2018Anticoncentration}.

\subsection{Classical output relabeling}\label{sec:relabeling}

The collision criterion does not require an additional permutation.
The fixed-output guarantee does require a symmetry argument: the block
structure of $\Psi_{\rm in}$ distinguishes some occupation strings, so
permutation invariance of the layer law alone does not make the core's
output probabilities identically distributed. At depth zero, for example,
only $2^B$ of the $N$ allowed outputs have nonzero probability.
To make every fixed output equivalent, sample an independent uniform
mode permutation $P$ and consider
\begin{equation}\label{eq:randommatchingensemble}
 \mathsf G_{n,t}=\mathsf C_{n,t}P.
\end{equation}
Here $P$ denotes both the permutation and its single-particle permutation
matrix; $Pz$ permutes the coordinates of an occupation string $z$.
Although~\eqref{eq:randommatchingensemble} places $P$ at the input,
it can be implemented by relabeling gate endpoints and classical outputs.

\begin{proposition}[Classical relabeling]\label{prop:permutationremoval}
For every input, the circuit-dependent collision $N\sum_y p_y^2$ has the same
distribution under $\mathsf C_{n,t}$ and $\mathsf G_{n,t}$; in particular,
\begin{equation}\label{eq:permutationcollision}
 \ratio(\mathsf C_{n,t},\Psi)=\ratio(\mathsf G_{n,t},\Psi).
\end{equation}
An instance $L_t\cdots L_1P$ can be implemented by applying
$L'_\ell=P^{-1}L_\ell P$, measuring $z$, and returning $y=Pz$ classically.
It uses $t$ native fermionic layers and $nt/2$ two-mode gates.
\end{proposition}
\noindent\emph{Proof:} Section~\ref{app:relabeling-proof}.
The identity $L_t\cdots L_1P=P L'_t\cdots L'_1$ moves the permutation
to the output. Conjugation preserves the matching law, so the core on the
right has the same distribution for every $P$. The final independent
uniform relabeling therefore makes the output moments equal for all
$y\in\mathcal Y_n$, while preserving each instance's collision.
All sampled matchings, gates, and $P$ are fixed across repeated shots of
one circuit instance.

Figure~\ref{fig:fermionsamplingscheme} shows the input preparation and the
complete sampling protocol. The depth $t$ counts native passive matching
layers after preparation. For a qubit line implementation one may instead
multiply the single-particle layers and recompile the entire resulting
$U\in U(n)$ using nearest-neighbor Givens rotations. This gives $O(n)$
qubit depth and $O(n^2)$ two-qubit gates for the whole passive circuit
\cite{Kivlichan2018lineardepth}; it need not preserve the original matching
layers or their $O(nt)$ gate count. The initial preparation adds only
constant depth and $3B$ CNOTs.

\FermionSamplingSchemeFigure

\section{The Haar benchmark}\label{sec:exact}

Before asking how many matching layers suffice, we need a target collision
scale. Here ``Haar'' means a uniformly random passive transformation
$U\in U(n)$, not a Haar-random unitary on the full $N$-dimensional
many-particle space. Write
$\ratio_{\rm Haar}(n,\Psi)=\ratio(\operatorname{Haar}(U(n)),\Psi)$.
This benchmark is input dependent: passive dynamics does not erase every
feature of a non-Gaussian state. The results below show that our magic
input has a constant Haar collision ratio, and set the normalization for
the finite-depth theorem in Section~\ref{sec:passiverandom}.
Haar anticoncentration for this input was established
in~\cite{Oszmaniec2022FermionSampling}; we use its two-copy representation
framework to express the benchmark for a general half-filled input.

The fourth power in $p_y^2$ can be written as a squared amplitude on two
mathematical replicas:
$|\Psi\Psi\rangle=|\Psi\rangle\otimes|\Psi\rangle$ and
$|yy\rangle=|y\rangle\otimes|y\rangle$. These replicas are an analysis
tool and do not require an extra experimental register. Label their
creation operators by $a_{i,c}^\dagger$, where $i$ is a physical mode
and $c\in\{1,2\}$ is the copy label. For $g\in SU(2)$, let $D(g)$
rotate the copy label identically at every mode:
\[
 D(g)a_{i,c}^\dagger D(g)^\dagger
 =\sum_{d=1}^2g_{dc}a_{i,d}^\dagger.
\]
This auxiliary action commutes with applying the same passive circuit to
both copies. Its invariant subspace is the \emph{copy-singlet space}
$V_0$; its projector and the input weight in it are
\[
 P_0=\int_{SU(2)}D(g)\,\mathrm dg,\qquad
 w_0(\Psi)=\|P_0|\Psi\Psi\rangle\|^2,
\]
where $\mathrm dg$ is normalized Haar measure. A doubled output
$|yy\rangle$ is a singlet: each physical mode contains either both copy
labels or neither, and a copy rotation leaves it unchanged.
Consequently only the singlet component of the input contributes to the
collision. Skew Howe duality makes $V_0$, at total particle number $2k$,
a single irreducible $U(n)$ representation
\cite{Howe1989Duality,GoodmanWallach2009}. Averaging over $U(n)$ makes
that component proportional to the identity. Appendix~\ref{app:proofs}
proves these statements and evaluates the dimension factor below.

\begin{theorem}[Haar collision]\label{thm:main}
For every normalized half-filled state $\Psi$ on $n=2k$ modes,
\begin{equation}\label{eq:main}
 \ratio_{\rm Haar}(n,\Psi)=\frac{(k+1)^2}{2k+1}\,w_0(\Psi).
\end{equation}
\end{theorem}
\noindent\emph{Proof:} Appendix~\ref{app:proofs}, using the
representation-theoretic averaging underlying
\cite{Oszmaniec2022FermionSampling}.

For a Fock input, meaning a single occupation-basis state $|x\rangle$,
$w_0=1$, so $\ratio_{\rm Haar}$ grows linearly with $n$.
For the magic input, $w_0$ instead decays inversely with $n$, cancelling
the growing dimension factor:

\begin{proposition}[Canonical magic benchmark]\label{thm:magic}
For $n=4B$ and the input~\eqref{eq:magic}, write
$w_0(B)=w_0(\Psi_{\rm in})$ and
$\ratio_{\rm Haar}(n)=\ratio_{\rm Haar}(n,\Psi_{\rm in})$. Then
\begin{equation}\label{eq:magicexact}
 w_0(B)=\frac12\int_{-1}^1\left(\frac{3+x^2}{4}\right)^B\mathrm dx
 =4^{-B}\sum_{j=0}^B\binom Bj\frac{3^{B-j}}{2j+1}.
\end{equation}
Consequently $\ratio_{\rm Haar}(n)<4$ for every $B\ge1$, and
$\ratio_{\rm Haar}(n)=2+O(n^{-1})$ as $n\to\infty$.
\end{proposition}
\noindent\emph{Proof:} Appendix~\ref{app:proofs}, by evaluating the
copy rotation on one four-mode block and multiplying over the blocks.
Thus reaching a fixed multiple of this Haar benchmark is sufficient for
a constant collision ratio. Section~\ref{sec:passiverandom} determines
exactly how rapidly the random-matching circuit reaches that scale.

\section{Random-matching collision dynamics}\label{sec:passiverandom}

We now determine how quickly random matching layers on the magic
input~\eqref{eq:magic} reach the Haar collision scale. Write
$\ratio_{n,t}=\ratio(\mathsf C_{n,t},\Psi_{\rm in})$.
All asymptotic statements take $n\to\infty$ through multiples of four,
and $\log$ denotes the natural logarithm.

\begin{samepage}
\begin{theorem}[Magic-input collision profile and depth threshold]\label{thm:magicprofile}
Put $z=n(4/9)^t$. For every fixed $Z>0$, uniformly over integer depths
with $0<z\le Z$,
\begin{equation}\label{eq:magicprofile}
 \frac{\ratio_{n,t}}{\ratio_{\rm Haar}(n)}
 =e^{3z/2}\bigl(1+O_Z(n^{-1/2})\bigr).
\end{equation}
If $z\to\infty$, then $\ratio_{n,t}\to\infty$; if $z\to0$, the ratio in
\eqref{eq:magicprofile} tends to one. For every fixed $q>1$, the least depth
$t^\star(q)$ with $\ratio_{n,t}\le q\ratio_{\rm Haar}(n)$ satisfies
\begin{equation}\label{eq:magicsharpthreshold}
 t^\star(q)=\left\lceil\tau_n(q)+O_q(n^{-1/2})\right\rceil,
 \qquad \tau_n(q)=\frac{\log(3n/(2\log q))}{\log(9/4)}.
\end{equation}
\end{theorem}
\end{samepage}
\noindent\emph{Proof:} Section~\ref{app:core-profile-proof}.
The error constants may depend on $Z$ or $q$, but not on $n$ or the depth
in the stated range. Thus the transition occurs within a bounded window
around $\log n/\log(9/4)$; the ceiling retains the effect of integer depth.
The proof has two parts: an exact collision chain gives the profile in
this window, and a separate two-particle bound excludes earlier crossings.
We explain both mechanisms below; Appendices~\ref{app:matchingproof} and~\ref{app:magic-analysis}
give their derivations.

\subsection{The collision chain and its spectrum}\label{sec:chain-overview}

\paragraph{Reducing the dynamics to $k+1$ weights.}
For any normalized half-filled input $\Psi$, use the singlet component
$|\varphi\rangle=P_0|\Psi\Psi\rangle$ and weight $w_0$ from
Section~\ref{sec:exact}. If $S,T$ are the occupied sets in the two copies,
the overlap $r=|S\cap T|$ counts doubly occupied modes. Let $Q_r$ project
onto the corresponding sector of the singlet space, and define
\[
 b_r=\frac{\langle\varphi|Q_r|\varphi\rangle}{w_0},\qquad
 b_r\ge0,\qquad \sum_{r=0}^k b_r=1.
\]
Averaging mode permutations and phases makes the projected density
operator uniform within each sector. The collision observable and the
layer law have these symmetries, so this averaging preserves expected
collision and requires no extra physical gates.

\CollisionProofFigure

One averaged matching layer transfers the sector weights by a reversible
Markov matrix $A$: $b(t)=A^tb$. Its stationary distribution is
\[
 \pi_r=\frac{\binom kr\binom{k+1}r}{\binom{2k+1}k}.
\]
Collision selects $r=k$, where the copies have identical occupations,
but gates transfer weight between this sector and the others. Tracking
all $k+1$ sectors therefore closes the dynamics. The exact reduction is
\begin{equation}\label{eq:main-chain-reduction}
 \boxed{\quad
 \frac{\ratio(\mathsf C_{n,t},\Psi)}{\ratio_{\rm Haar}(n,\Psi)}
 =\frac{(A^tb)_k}{\pi_k}
 =1+\sum_{j=1}^k c_j(\Psi)\lambda_j^t.
 \quad}
\end{equation}
It also holds for $\mathsf G_{n,t}$ by
Proposition~\ref{prop:permutationremoval}.
Here $\lambda_j$ is a decay factor of the chain, and $c_j(\Psi)$ records
how strongly the input excites that pattern and the collision detects it.
These spectral modes are patterns of sector weights, not physical modes.
Hahn polynomials diagonalize $A$; Lemma~\ref{lem:matchingchain} and
\eqref{eq:sharpexpansion} give the exact eigenvalues and coefficients.

\paragraph{The input determines the relaxation scale.}
The eigenvalues depend only on the gate ensemble. Lemma~\ref{lem:quantmatching}
gives $\lambda_j=(2/3)^j\exp[O(j^2/n)]$ for $j\le k^{1/3}$, with enough
precision for logarithmic depths. In particular,
$\lambda_1=2/3+O(n^{-1})$ and $\lambda_2=4/9+O(n^{-1})$.
For a Fock input, $c_1\sim n$. For the magic input, the sector law instead gives
\[
 \boxed{\quad c_1=3+O(n^{-1}),\qquad c_2=(3/2+o(1))n.\quad}
\]
The first identity is~\eqref{eq:B1magic}; the second follows from
Proposition~\ref{prop:magic-fixed-coeffs}. The slowest mode is still present,
but its bounded coefficient makes its contribution vanish at the
transition. The second contributes approximately $(3/2)n(4/9)^t$, which
sets the scale $z=n(4/9)^t$ and the leading depth $\log n/\log(9/4)$.

\ModeMechanismFigure

Higher modes determine the full profile. For each fixed integer $m\ge0$,
with $c_0=1$, Proposition~\ref{prop:magic-fixed-coeffs} gives
\[
 \frac{c_{2m}}{n^m}\longrightarrow\frac{(3/2)^m}{m!},\qquad
 \frac{c_{2m+1}}{n^{m+1/2}}\longrightarrow0.
\]
These limits follow by comparing input fluctuations with a Gaussian
distribution and the Hahn polynomials with their Hermite limits.
In the transition window, the even spectral terms approach
$(3z/2)^m/m!$ and the odd terms vanish, giving the formal sum
$\sum_{m\ge0}(3z/2)^m/m!=e^{3z/2}$.
Fixed-degree limits alone do not justify this sum: the number of modes
grows with $n$. The uniform coefficient and tail estimates in
Section~\ref{app:uniform-summation} make the passage rigorous and yield
the error in~\eqref{eq:magicprofile}. They also control the sum of absolute
values used by the finite-size certificate in Section~\ref{sec:corollaries}.

\subsection{Lower bound: two-particle correlations}\label{sec:lower-overview}

The transition profile alone does not exclude an earlier crossing.
For this we track pair occupations. The magic input has one-particle
density matrix $I_n/2$, preserved by every passive circuit, so each mode
already has mean occupation $1/2$. Pair correlations can nevertheless
retain structure. For a fixed circuit, measure their departure from the
uniform half-filled law by
\[
 S_0=\sum_{i<j}\bigl(\langle n_i n_j\rangle-q_n\bigr)^2,
 \qquad q_n=\frac{n-2}{4(n-1)}.
\]
Cauchy--Schwarz, applied to quadratic occupation functions under that
uniform law, gives
\begin{equation}\label{eq:main-two-particle-bound}
 N\sum_y p_y^2\ge1+\frac{S_0}{\eta_n},\qquad
 \eta_n=\frac{n(n-2)}{16(n-1)(n-3)}.
\end{equation}
Thus persistent pair correlations force excess collision. This bound
holds for every circuit before ensemble averaging.

To evolve $S_0$, subtract $q_n I$ from the two-particle density matrix.
Its diagonal entries are the pair deviations above, which passive gates
mix with off-diagonal entries. Grouping squared entries according to
whether their row and column mode pairs share two, one, or zero indices
closes the averaged evolution on three states. Their decay rates are
$1,\lambda_1,\lambda_2$, the same first rates as in the full collision
chain. For the magic input this gives
\[
 \mathbb E S_0(t)=e_n+a_n\lambda_1^t+b_n\lambda_2^t,
 \qquad e_n,a_n,b_n\ge0,\qquad b_n/\eta_n\sim3n/2.
\]
Lemma~\ref{lem:magic-early-lower} gives the exact coefficients and proves
the three-state reduction and~\eqref{eq:main-two-particle-bound}.
Keeping the last positive contribution yields
\[
 \ratio_{n,t}\ge1+\frac{b_n}{\eta_n}\lambda_2^t
 =1+(3/2+o(1))n(4/9)^t
 \qquad\text{for }t=O(\log n).
\]
The exact lower bound holds at every depth and decreases with $t$.
Since $\ratio_{\rm Haar}(n)<4$, it excludes reaching any fixed multiple
of Haar before a bounded window around $\log n/\log(9/4)$.
Within that window, the uniform profile locates the first crossing with
the precision in~\eqref{eq:magicsharpthreshold}. This argument does not
assume monotonicity of the actual collision.

\section{Finite gates and fixed-output anticoncentration}\label{sec:corollaries}

The collision theorem has a finite-gate version and an exact certificate at
every system size. Classical relabeling then gives a probability bound for
each fixed output.

\subsection{An exact finite gate alphabet}\label{sec:finitealphabet}

Let
\begin{equation}\label{eq:single-particle-clifford}
 \mathcal F=\langle e^{\pi i/4}I,H,S\rangle\le U(2),\qquad
 H=\frac1{\sqrt2}\begin{pmatrix}1&1\\1&-1\end{pmatrix},\quad
 S=\operatorname{diag}(1,i).
\end{equation}
Here $\langle\cdot\rangle$ denotes the group generated by the displayed
matrices, $I$ is the two-dimensional identity, and $i^2=-1$.
This is the $192$-element single-particle Clifford group with eighth-root
scalar phases. Its fermionic action $1\oplus u\oplus\det u$ need not be a
stabilizer Clifford gate on qubits.

\begin{lemma}[Exact finite replacement]\label{lem:cliffordtwirl}
For every operator $X$ on two replicas of the two-mode Fock space,
\begin{equation}\label{eq:clifford-fock-twirl}
 \frac1{192}\sum_{f\in\mathcal F}F_{ij}(f)^{\otimes2}X
 F_{ij}(f)^{\dagger\otimes2}
 =\int_{U(2)}F_{ij}(u)^{\otimes2}X
 F_{ij}(u)^{\dagger\otimes2}\,du.
\end{equation}
Thus independent uniform replacements preserve the complete two-copy channel
and every collision ratio, including all statements of
Theorem~\ref{thm:magicprofile}.
\end{lemma}

\noindent\emph{Proof:} Section~\ref{app:finitegates}, by matching the local
commutants. Write $\mathsf G^{\rm Cl}_{n,t}$ for the resulting ensemble,
including the uniform mode permutation. This exact moment replacement
concerns collision and anticoncentration. The average-case hardness
reduction in Section~\ref{sec:hardness} uses continuous Haar gates.

\subsection{A finite-size certificate}

Using the exact coefficients and eigenvalues in
\eqref{eq:main-chain-reduction}, define
\begin{equation}\label{eq:magiccertificate}
 t_n^{\rm mag}=\min\left\{t\in\mathbb Z_{\ge0}:\quad
                          \sum_{j=1}^k|c_j|\lambda_j^t\le1\right\}.
\end{equation}

\begin{corollary}[Certified anticoncentration]\label{cor:certified-ac}
The certificate is computable exactly in polynomial time and satisfies
$t_n^{\rm mag}=\lceil\tau_n(2)+O(n^{-1/2})\rceil$.
For every $n=4B$ and $t\ge t_n^{\rm mag}$,
\begin{equation}\label{eq:magicdepth}
 \ratio_{n,t}\le2\ratio_{\rm Haar}(n),\qquad
 \Pr\!\left[p_y>\frac1{2N}\right]\ge\frac1{32}
\end{equation}
for each fixed half-filled output $y$ under either $\mathsf G_{n,t}$ or
$\mathsf G^{\rm Cl}_{n,t}$. The collision bound also holds for their cores.
\end{corollary}

\begin{proof}
Lemma~\ref{lem:magiccertificate} in Section~\ref{app:finite-certificate}
proves exact computability. The uniform
spectral estimate in Section~\ref{app:magicthreshold} applies to the absolute
sum as well as the signed sum and gives the stated rounding formula.
Since $0\le\lambda_j<1$, the certificate holds at every later depth.
For a fixed $y$, the probability and expectations are over circuit
instances. Output symmetry gives the nonnegative random variable $X=Np_y$
with $\mathbb EX=1$ and $\mathbb EX^2=\ratio_{n,t}<8$ by
Proposition~\ref{thm:magic}. The Paley--Zygmund inequality
\cite{Hangleiter2018Anticoncentration} gives
$\Pr[X>1/2]\ge1/(4\mathbb EX^2)\ge1/32$.
Proposition~\ref{prop:permutationremoval} transfers the collision bound to the
cores, and Lemma~\ref{lem:cliffordtwirl} supplies the finite-gate version.
\end{proof}

The certificate concerns a sufficient collision bound; it need not be the
first crossing of the signed spectral sum at a given finite size.

\section{Average-case hardness of output probabilities}\label{sec:hardness}

The collision theorem describes how broadly the output probabilities are
spread. We now show that estimating those probabilities is computationally
hard on average for the continuous Haar-gate ensemble at a larger
logarithmic depth. The reduction has two parts: place a hard computation
inside a typical matching schedule, then recover its probability from
queries close to the Haar gate law.

\subsection{The estimation problem and computational model}

An instance specifies a matching schedule $S=(M_1,\ldots,M_D)$, its
ordered two-mode gates $u=(u_{\ell,e})$, and an output $z\in\Y$.
Write $C(S,u)=L_D\cdots L_1$, let $g=nD/2$ be the total gate count, and set
\[
 p_z(S,u)=|\langle z|\Gamma(C(S,u))|\Psi_{\rm in}\rangle|^2.
\]
The gate convention extends to either order of the endpoints by
$F_{ji}(u)=F_{ij}(XuX)$. The choices $u=I$ and $u=X$ are called
\emph{switches}; they leave the modes in place or apply a fermionic swap.
An evaluator is successful if its reply is exact,
or differs from $p_z(S,u)$ by at most a specified additive tolerance.
Success is averaged over the circuit law and an independent uniform
$z\in\Y$. For a circuit ensemble $\nu$, write
$\textup{\textsc{Exact}}(\nu,\delta)$ for exact evaluation with success
probability at least $1-\delta$ under this law. A randomized evaluator is a fixed, stateless procedure with
fresh coins on each call; its coins are included in the success probability.

For the core $\Gp$, uniform-output estimation is equivalent to estimating
one fixed output of $\mathsf G_{n,t}$. Indeed,
Proposition~\ref{prop:permutationremoval} gives
\[
 p_{y_0}(CP)=p_{P^{-1}y_0}(P^{-1}CP).
\]
The conjugated circuit has the core law independently of $P$, and
$P^{-1}y_0$ is uniform. Conversely, given a core instance and uniform $z$,
choose $P$ uniformly subject to $Pz=y_0$ and conjugate the circuit by $P$.
These maps preserve the probability and the required instance laws.

The reductions below use the real-RAM model: Haar samples, circuit queries,
and arithmetic on gate entries are exact, as in continuous interpolation
arguments for random-circuit sampling
\cite{Bouland2019RCS,Movassagh2023hardness,Oszmaniec2022FermionSampling}.
We also allow exact comparisons and exact uniform choices from finite sets.
All stated schedule distributions admit efficient exact samplers in this
model. This convention concerns the reductions, not the physical gate
accuracy: a guarantee on a set of positive Haar measure need not hold on
any chosen discretization. The robust decoder itself uses finite rational
data and has polynomial bit complexity.

\subsection{Hardness for the random matching ensemble}

The embedding uses four windows to bring the interacting modes together,
followed by one window to route the occupied set to the requested output.
Define
\begin{align}
 \tau_n^{\rm meet}&=\left\lceil\tfrac58\log_2 n\right\rceil,
 &T_n^{\rm h}&=\lceil4\log_2 n\rceil,\notag\\
 t_{\rm hard}(n)&=4(\tau_n^{\rm meet}+1)+T_n^{\rm h}
 \le6.5\log_2 n+9,\label{eq:hardness-depth}\\
 \eta_n^{\rm emb}&=68n^{-1/4}+4n^{1/16}e^{-n^{1/4}/8}+n^{-3/10}.
 \label{eq:hardness-embedding-error}
\end{align}
Here $\eta_n^{\rm emb}$ bounds the probability that the embedding fails.
It is distinct from the covariance normalization $\eta_n$ in
\eqref{eq:main-two-particle-bound}.

\begin{theorem}[Average-case hardness in the shallow ensemble]
\label{thm:shallow-hardness}
Let $p(n)\ge4$ and $t(n)\ge t_{\rm hard}(n)$ be integer-valued,
polynomial-time computable, polynomially bounded functions, and put
$g=nt(n)/2$. There is an explicitly computable rational tolerance
$\epsilon_n>0$, given in~\eqref{hard:eq:robustprecision}, with
\[
 \log_2(1/\epsilon_n)=O\bigl(n+g\log(gp)+gp\bigr),
\]
such that the following holds. Suppose that, for all sufficiently large
$n$ divisible by four, an evaluator $\mathcal O$ estimates $p_z(S,u)$ to
additive error at most $\epsilon_n$ on a fraction at least
\[
 \frac34+\frac1{p(n)}+\eta_n^{\rm emb}
\]
of instances from $\Gp$ with uniform $z$. Then
$\mathsf P^{\#\mathsf P}\subseteq\mathsf{BPP}^{\mathcal O}$ for randomized
polynomial-time reductions in the real-RAM model. The conclusion also
holds for exact evaluation on the same fraction of instances.
For fixed $p$ and $t=O(\log n)$, the sufficient tolerance is
$\epsilon_n=2^{-O(n\log^2 n)}$.
At these depths, Corollary~\ref{cor:certified-ac} also gives
fixed-output anticoncentration for $\mathsf G_{n,t}$.
\end{theorem}
\noindent\emph{Proof:} Theorem~\ref{hard:thm:pure} proves the exact case;
Corollary~\ref{hard:cor:robusthard} gives the additive-error extension.

The theorem applies to any fixed success fraction strictly above $3/4$
for sufficiently large sizes. Its quantitative embedding is conservative:
$\eta_n^{\rm emb}<1/4$ only above approximately $2^{32.4}$ modes, and a
hard graph with $|E|$ edges is padded to $n\ge(2|E|)^{16}$. These bounds
establish an asymptotic reduction. They do not provide practical finite-size
hardness estimates. The depth $t_{\rm hard}$ also exceeds the sharp
collision threshold by a constant factor; hardness at the sharp threshold
is not established here.

\subsection{Why the reduction works}

\paragraph{A hard circuit in four layers.}
Each four-mode magic block is a Bell pair when one particle in two rails
encodes a qubit. A passive beam splitter followed by an occupation
postselection fuses two such logical vertices. Preparing graph edges,
fusing the copies of each vertex in two layers, and applying final
measurement rotations realizes a postselected universal computation in
four native layers. Theorem~\ref{hard:thm:worst} constructs a circuit
$C_Q$, an output $y_Q$, and a known integer $\kappa$ satisfying
\[
 p_{y_Q}(C_Q)=2^{-\kappa}|\langle0|Q|0\rangle|^2.
\]
The construction handles fermionic signs explicitly and uses a fixed finite
set of two-mode gates. This is a worst-case construction; its gate set is
not the finite moment-matching alphabet of Section~\ref{sec:finitealphabet}.

\paragraph{Embedding in typical schedules.}
The hard circuit need not have the sampled matching schedule. We fill
unused gates with swap or no-swap switches to move its interacting modes
together. The meeting lemma grows disjoint sets of positions reachable by
each mode, so the final matching of a window can realize the required
interactions. A separate set-routing lemma moves the occupied set to a
uniform target output. Its proof computes the collision of random switch
settings through the permutation representation on half-filled subsets.
A maximum-flow algorithm finds the switches whenever routing is possible.
Section~\ref{hard:sec:pure} combines the two steps to obtain
\eqref{eq:hardness-depth}--\eqref{eq:hardness-embedding-error}.

\paragraph{From Haar queries to the hard endpoint.}
Keep the sampled schedule and output fixed, and interpolate each gate
along a unitary Cayley path. The parameter value $0$ gives independent Haar
gates and the value $1$ gives the embedded hard circuit. Clearing the known
denominators turns its output probability into a polynomial of degree at
most $4g$. Queries near $0$ remain close in total variation to the promised
ensemble. Error-correcting interpolation recovers the value at $1$, even
when a fraction of replies is wrong. The proof uses only marginal error
bounds along a path, so it does not assume independent oracle errors.
Appendix~\ref{hard:sec:reduction} gives the exact reduction.

For approximate replies, a rational linear program replaces exact
interpolation. It suppresses the unknown incorrect nodes while fixing the
normalization at the hard endpoint. A tail bound on the Cayley denominators
and a Lagrange extrapolation bound yield the explicit tolerance
in~\eqref{hard:eq:robustprecision}. This decoder needs no additional
$\mathsf{NP}$ oracle. Section~\ref{hard:sec:robust} proves the bound and
shows how the resulting accuracy recovers an exact counting answer.

\subsection{A deterministic routing variant}

The embedding overhead can be separated from the random matching layers.
For $n=2^r$, a fixed prefix of five Bene\v{s} routing networks and four
interaction layers has depth $10\log_2 n-1$. With independent Haar gates
on every prefix edge, followed by $t$ random matching layers, this defines
the hybrid ensemble $\Gsh$ of Section~\ref{hard:sec:hybrid}.
Its schedules admit every hard instance and requested output, so the
embedding never fails. The same exact and additive-error hardness results
hold with success threshold $3/4+1/p(n)$ and total gate count
$g=n(10\log_2 n-1+t)/2$.

The random suffix gives anticoncentration once $t\ge T_n$, with $T_n$
from~\eqref{eq:matchingtime}. The prefix changes the input spectral weights,
so the sharp magic-input threshold does not transfer to this variant.
The uniform collision bound does transfer: the prefix is passive and
therefore preserves the Haar collision benchmark, while the spectral
envelope holds for every normalized half-filled input.

\paragraph{Scope.}
The continuous ensemble thus has both anticoncentration and average-case
hardness of exact or high-precision probability estimation. The finite
alphabet preserves the collision law, but the interpolation argument does
not establish its average-case hardness. The proved additive tolerance is
also finer than the accuracy needed for constant-total-variation sampling
hardness. Section~\ref{sec:conclusion} states these remaining questions.

\section{Numerical checks}\label{sec:numerics}

We first check two parts of the collision analysis: the approach of the exact
finite-size collision to its asymptotic profile, and agreement between the
reduced chain and independently evolved many-particle circuits. The latter
also tests fermionic signs and the finite gate replacement. These checks
support the analytical derivation; the asymptotic claims are proved in
the appendices. Checks of the hardness construction and reduction are
summarized below; the accompanying reproducibility files retain the
detailed check log.

\paragraph{Exact finite-size profile.}
For $n=32,64,128,256,512$, we evaluate the complete spectral sum in
\eqref{eq:main-chain-reduction} at integer depths using rational arithmetic.
Only the final plotted values are converted to floating point.
Figure~\ref{fig:magic-profile} compares the exact values with the limiting
profile as a function of centered depth
$s=t-\log n/\log(9/4)$. The certified depths are $6,7,8,9,9$, respectively.
The visible finite-size corrections show why the asymptotic curve alone
should not be used to certify a particular finite circuit size: the
absolute spectral sum in~\eqref{eq:magiccertificate} supplies that guarantee.

\begin{figure}[!ht]
\centering
\includegraphics[width=.82\textwidth]{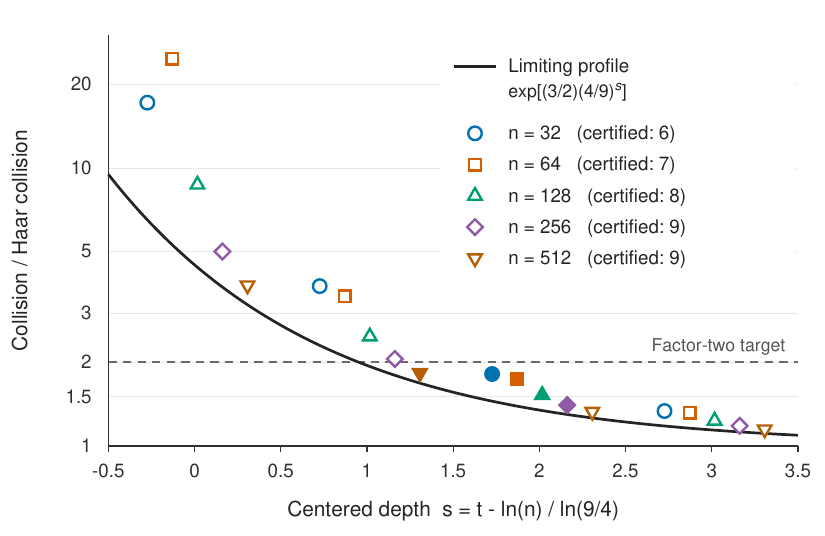}
\caption{\textbf{The magic-input collision profile.}
Markers show the complete exact spectral sum at integer depths for
$n=32,64,128,256,512$, against $s=t-\log n/\log(9/4)$.
The solid curve is the proved limit $\exp[\tfrac32(4/9)^s]$;
the dashed line marks twice the Haar collision. Filled markers indicate
certified depths. The vertical axis is logarithmic; finite-size points
are not interpolated. These exact calculations have no sampling error.}
\label{fig:magic-profile}
\end{figure}

\paragraph{Independent circuit simulations.}
For each $n=8,12,16$ and each gate law (Haar $U(2)$ or the uniform
$192$-element alphabet), we evolve $\CircuitSamples$ independently sampled
occupation-basis state vectors through eight matching layers. Each circuit
includes an independent uniform input permutation, and the evolution
includes fermionic reordering signs. At every depth we enumerate all $N$
output probabilities and compute $N\sum_y p_y^2$, so the sampling error
comes from the finite number of circuits, not from measurement shots.
The collision chain is used only for the comparison means.

Figure~\ref{fig:circuit-checks} reports the sample means and pointwise
$95\%$ normal confidence intervals, with standard errors estimated across
independent circuits. Across all positive-depth comparisons, the largest
absolute difference from the exact mean is $\CircuitMaxResidual$ estimated
standard errors. Depth-zero collision is deterministic and agrees exactly.
Depths share circuit prefixes, so these are correlated checks, not
independent tests or simultaneous confidence bands. Sizes and gate laws
use separate random streams.

\begin{figure}[!htbp]
\centering
\includegraphics[width=\textwidth]{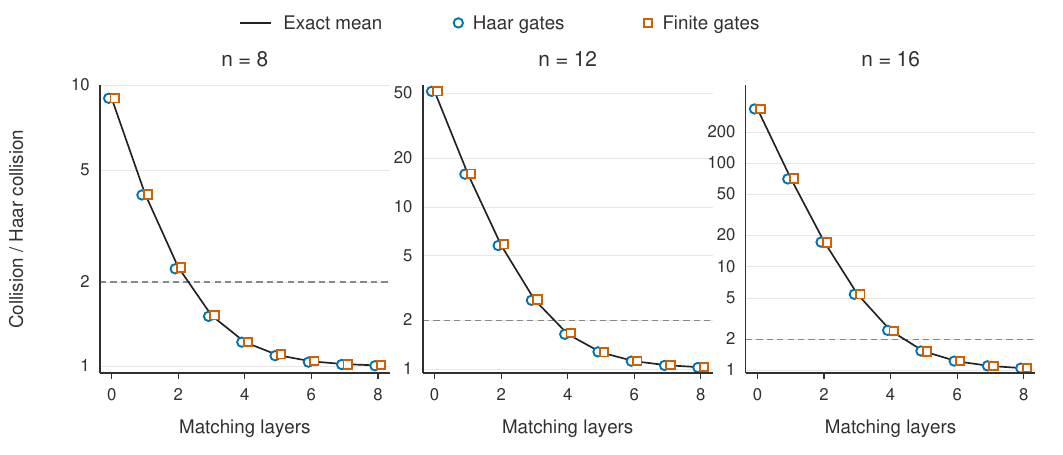}
\caption{\textbf{Direct circuit evolution checks the collision chain.}
For each size, circles (Haar gates) and squares (finite gates) show mean
collision divided by the exact Haar benchmark, using $\CircuitSamples$
circuits per gate law. Error bars are pointwise $95\%$ confidence intervals
and may be smaller than the markers. The solid line joins exact
integer-depth means as a guide; the dashed line is the factor-two target.
Agreement for both laws checks the collision prediction without using the
reduced chain to propagate the simulated state.}
\label{fig:circuit-checks}
\end{figure}

\paragraph{Fixed-output anticoncentration.}
For the fixed string $y$ with the first $k$ modes occupied, we also measure
the fraction of circuit instances satisfying $p_y>1/(2N)$ at the certified
depth. Table~\ref{tab:fixed-output-check} gives these fractions with $95\%$
Wilson binomial intervals. They lie well above the conservative bound
$1/32$ in Corollary~\ref{cor:certified-ac}. The event frequencies need not
be equal for the two gate laws: the finite replacement equates their
first and second probability moments, not their entire distributions.
In particular, the difference visible at $n=8$ does not contradict the
moment identity.

\begin{table}[!htbp]
\centering
\small
\begin{tabular}{cc cc}
\hline
Modes $n$ & Certified depth & Haar gates & Finite gates\\
\hline
\CircuitEventRows
\hline
\end{tabular}
\caption{Observed fractions with $p_y>1/(2N)$ at the certified depth,
with $95\%$ Wilson intervals in brackets. Each column and size uses
$\CircuitSamples$ independent circuit instances. The rigorous lower bound
is $1/32\simeq0.03125$ for both gate laws.}
\label{tab:fixed-output-check}
\end{table}

The exact calculations test the finite-size profile at substantially larger
$n$ than the direct simulations can reach. The simulations provide an
independent check of the circuit-to-chain reduction and the output event,
but their modest sizes do not establish asymptotic convergence rates or
sampling hardness. Larger exact calculations near the threshold and
noise-dependent circuit simulations would address distinct questions
beyond these checks.

\paragraph{Hardness construction and decoding.}
Exact occupation-basis checks verify the fusion map, including arbitrary
rail orders and all fermionic signs, in $16{,}848$ basis cases on two to six
logical slots. Independent circuit simulations check the graph-state
normalization and switch embeddings. Exhaustive averaging over small
matching schedules reproduces the set-routing collision formula, while
exact rational calculations test its spectral bounds through $n=800$.

For the interpolation step, exact rational checks exercise both successful
and inconsistent Berlekamp--Welch transcripts. The robust decoder is checked
through its feasibility witness and Lagrange bound, supplemented by complete
vertex enumeration for five small linear programs. These tests also check
the success-probability arithmetic and the spacing of the encoded counting
values. They support the finite-dimensional identities used by the
reduction; the proof of average-case hardness is the analytical argument in
Appendices~\ref{hard:sec:reduction} and~\ref{hard:sec:embeddings}.

\section{Discussion}\label{sec:conclusion}

\subsection{Conclusion}

Random matching layers acting on the paired magic input reach a constant
collision ratio in logarithmic native fermionic depth, with
$O(n\log n)$ two-mode gates. We determine the leading depth constant,
the full limiting transition profile, and the first depth reaching any
fixed factor above the passive-Haar benchmark, up to the stated rounding
error. A finite gate alphabet preserves the collision law exactly, and
classical output relabeling turns it into an anticoncentration guarantee
for every fixed half-filled output.

The mechanism is the interaction between the input and the circuit's
decay patterns. The canonical magic input keeps the slowest spectral
coefficient bounded, while the next one grows linearly with system size.
An exact chain on only $n/2+1$ occupation sectors resolves the resulting
transition; a separate three-state chain for two-particle correlations
excludes earlier crossings. This provides a way to analyze a physically
specified input and observable without requiring convergence of the
entire many-particle state ensemble to Haar randomness.

The same continuous random matching ensemble also has average-case
hardness of exact and high-precision output-probability estimation at
logarithmic depth. The proof combines a four-layer postselected
construction with typical-schedule embedding and unitary interpolation.
A rational linear-program decoder controls approximate and incorrect
replies without an additional $\mathsf{NP}$ oracle. A deterministic
routing prefix gives a variant with no embedding failures. These results
place the collision guarantee and the probability-estimation reduction in
a common shallow circuit model.

\subsection{Limitations and future work}

The depth statements assume native gates between arbitrary pairs of
fermionic modes and a supplied non-Gaussian input. Preparing the magic
blocks, routing interactions on restricted hardware, and encoding the
fermions into qubits are separate resource questions. The results also
assume ideal gates and measurements. The robustness theorem concerns errors
in estimated probabilities; it does not bound the effect of preparation
errors, particle loss, or physical gate noise.

The collision lower bound is optimal within the fresh random matching
ensemble. It does not prove optimality over all passive architectures, and
divergent collision alone does not rule out fixed-output anticoncentration.
The hardness depth exceeds the sharp collision threshold by a constant
factor, and the random-schedule embedding has substantial padding and
finite-size overheads. Reducing those overheads and bringing the hardness
depth closer to the collision threshold are natural next steps.

Hardness of sampling within constant total-variation distance remains
open. The additive tolerance proved here is finer than the uniform-output
scale $N^{-1}=2^{-n+O(\log n)}$ needed by the usual sampling-to-counting
arguments~\cite{AaronsonArkhipov2013,Hangleiter2023review}.
Establishing average-case hardness at that scale, or suitable relative-error
hardness, would address this gap. A further task is to transfer the
exact-real reductions to a specified finite-precision ensemble. The exact
moment identity for the $192$-element gate alphabet establishes its
anticoncentration, but does not supply this complexity transfer.

For other non-Gaussian inputs, the collision analysis suggests asking which
spectral weights grow with $n$ and whether their sum gives a comparable
transition profile. For other architectures, one can ask which structured
or geometrically local matching schedules support both rapid collision
relaxation and efficient embedding of hard instances. Exact finite-size
calculations and independent circuit simulations provide complementary
checks of these mechanisms.

\subsection{Acknowledgements}

I.K. was partially supported by PEPR EPiQ.

\emph{AI Disclosure:} Claude Opus 5.5 and ChatGPT 6 Astra were extensively used during the research process to explore and refine proof ideas, formalize the proofs, and to perform numerical verifications. The proofs were independently verified by the authors. The authors remain responsible for the claims, attribution, and the final content.

\clearpage
\phantomsection
\addcontentsline{toc}{section}{References}
{\raggedright
\bibliographystyle{unsrt}
\bibliography{paper/refs}

@book{AndrewsAskeyRoy1999,
  author    = {Andrews, George E. and Askey, Richard and Roy, Ranjan},
  title     = {Special Functions},
  series    = {Encyclopedia of Mathematics and its Applications},
  volume    = {71},
  publisher = {Cambridge University Press},
  year      = {1999},
  doi       = {10.1017/CBO9781107325937}
}

@article{Benes1964,
  author  = {Bene{\v s}, V{\'a}clav E.},
  title   = {Optimal Rearrangeable Multistage Connecting Networks},
  journal = {Bell System Technical Journal}, volume = {43}, number = {4},
  pages = {1641--1656}, year = {1964}, doi = {10.1002/j.1538-7305.1964.tb04103.x}
}

@book{BrouwerCohenNeumaier1989,
  author    = {Brouwer, Andries E. and Cohen, Arjeh M. and Neumaier, Arnold},
  title     = {Distance-Regular Graphs},
  series    = {Ergebnisse der Mathematik und ihrer Grenzgebiete (3)},
  volume    = {18},
  publisher = {Springer-Verlag},
  address   = {Berlin},
  year      = {1989},
  doi       = {10.1007/978-3-642-74341-2}
}

@misc{Braccia2025OptimalFLO,
  title  = {Optimal Haar random fermionic linear optics circuits},
  author = {Braccia, Paolo and Diaz, N. L. and Larocca, Martin and Cerezo, M. and Garc{\'i}a-Mart{\'i}n, Diego},
  year   = {2025},
  note   = {arXiv:2505.24212}
}

@article{Bremner2016averagecase,
  title   = {Average-case complexity versus approximate simulation of commuting quantum
             computations},
  author  = {Bremner, Michael J. and Montanaro, Ashley and Shepherd, Dan J.},
  journal = {Physical Review Letters}, volume = {117}, number = {8}, pages = {080501},
  year = {2016}, note = {arXiv:1504.07999}
}

@article{AaronsonArkhipov2013,
  author  = {Aaronson, Scott and Arkhipov, Alex},
  title   = {The Computational Complexity of Linear Optics},
  journal = {Theory of Computing},
  volume  = {9},
  number  = {4},
  pages   = {143--252},
  year    = {2013},
  doi     = {10.4086/toc.2013.v009a004},
  note    = {arXiv:1011.3245}
}

@misc{Kolarovszki2026Photonic,
  author = {Kolarovszki, Zolt{\'a}n and Kaposi, {\'A}goston and
            Zimbor{\'a}s, Zolt{\'a}n and Oszmaniec, Micha{\l}},
  title = {General framework for anticoncentration and linear cross-entropy
           benchmarking in photonic quantum advantage experiments},
  year = {2026},
  eprint = {2604.15258},
  archivePrefix = {arXiv},
  primaryClass = {quant-ph},
  doi = {10.48550/arXiv.2604.15258},
  note = {arXiv:2604.15258v1}
}

@misc{Mhiri2026Boson,
  author = {Mhiri, Hela and Thomas, Hugo and Monbroussou, L{\'e}o and
            Chabaud, Ulysse and Holmes, Zo{\"e} and Kashefi, Elham},
  title = {Boson sampling beyond the dilute regime: second moments and anti-concentration},
  year = {2026},
  eprint = {2604.14323},
  archivePrefix = {arXiv},
  primaryClass = {quant-ph},
  doi = {10.48550/arXiv.2604.14323},
  note = {arXiv:2604.14323v1}
}

@article{Dalzell2022LogDepth,
  title   = {Random quantum circuits anti-concentrate in log depth},
  author  = {Dalzell, Alexander M. and Hunter-Jones, Nicholas and Brand{\~a}o, Fernando G. S. L.},
  journal = {PRX Quantum},
  volume  = {3},
  pages   = {010333},
  year    = {2022},
  note    = {arXiv:2011.12277}
}

@book{GoodmanWallach2009,
  title     = {Symmetry, Representations, and Invariants},
  author    = {Goodman, Roe and Wallach, Nolan R.},
  publisher = {Springer},
  year      = {2009}
}

@article{Go2026ShallowBoson,
  title     = {On computational complexity and average-case hardness of shallow-depth boson sampling},
  author    = {Go, Byeongseon and Oh, Changhun and Jeong, Hyunseok},
  journal   = {Quantum},
  volume    = {10},
  pages     = {2026},
  year      = {2026},
  month     = mar,
  doi       = {10.22331/q-2026-03-13-2026},
  eprint    = {2405.01786},
  archivePrefix = {arXiv},
  primaryClass  = {quant-ph}
}

@article{Hangleiter2018Anticoncentration,
  title   = {Anticoncentration theorems for schemes showing a quantum speedup},
  author  = {Hangleiter, Dominik and Bermejo-Vega, Juan and Schwarz, Martin and Eisert, Jens},
  journal = {Quantum},
  volume  = {2},
  pages   = {65},
  year    = {2018}
}

@article{Hangleiter2023review,
  title   = {Computational advantage of quantum random sampling},
  author  = {Hangleiter, Dominik and Eisert, Jens},
  journal = {Reviews of Modern Physics}, volume = {95}, number = {3}, pages = {035001},
  year = {2023}, note = {arXiv:2206.04079}
}

@article{Hebenstreit2019magic,
  title   = {All pure fermionic non-Gaussian states are magic states for matchgate
             computations},
  author  = {Hebenstreit, Martin and Jozsa, Richard and Kraus, Barbara and
             Strelchuk, Sergii and Yoganathan, Mithuna},
  journal = {Physical Review Letters}, volume = {123}, number = {8}, pages = {080503},
  year = {2019}, note = {arXiv:1905.08584}
}

@article{Howe1989Duality,
  title   = {Remarks on classical invariant theory},
  author  = {Howe, Roger},
  journal = {Transactions of the American Mathematical Society},
  volume  = {313},
  pages   = {539--570},
  year    = {1989}
}

@article{Mackey1952,
  author  = {Mackey, George W.},
  title   = {Induced Representations of Locally Compact Groups {I}},
  journal = {Annals of Mathematics},
  volume  = {55},
  number  = {1},
  pages   = {101--139},
  year    = {1952},
  doi     = {10.2307/1969423}
}

@article{Kivlichan2018lineardepth,
  title   = {Quantum simulation of electronic structure with linear depth and connectivity},
  author  = {Kivlichan, Ian D. and McClean, Jarrod and Wiebe, Nathan and Gidney, Craig and
             Aspuru-Guzik, Al\'an and Chan, Garnet Kin-Lic and Babbush, Ryan},
  journal = {Physical Review Letters}, volume = {120}, number = {11}, pages = {110501},
  year = {2018}, note = {arXiv:1711.04789}
}

@misc{Krovi2022averagecase,
  title  = {Average-case hardness of estimating probabilities of random quantum circuits
            with a linear scaling in the error exponent},
  author = {Krovi, Hari}, year = {2022}, note = {arXiv:2206.05642v1}
}

@article{Movassagh2023hardness,
  title   = {The hardness of random quantum circuits},
  author  = {Movassagh, Ramis},
  journal = {Nature Physics}, volume = {19}, number = {11}, pages = {1719--1724},
  year = {2023}, doi = {10.1038/s41567-023-02131-2}
}

@misc{Oh2026MagicSimulation,
  title         = {Classical simulation of free-fermionic dynamics and quantum chemistry with magic input},
  author        = {Oh, Changhun and Oszmaniec, Micha{\l} and Reardon-Smith, Oliver and
                   Zimbor{\'a}s, Zolt{\'a}n},
  year          = {2026},
  eprint        = {2604.26813},
  archivePrefix = {arXiv},
  primaryClass  = {quant-ph},
  note          = {arXiv:2604.26813v2}
}

@article{Oszmaniec2022FermionSampling,
  title   = {Fermion Sampling: A Robust Quantum Computational Advantage Scheme Using Fermionic Linear Optics and Magic Input States},
  author  = {Oszmaniec, Micha{\l} and Dangniam, Ninnat and Morales, Mauro E. S. and Zimbor{\'a}s, Zolt{\'a}n},
  journal = {PRX Quantum},
  volume  = {3},
  pages   = {020328},
  year    = {2022},
  doi     = {10.1103/PRXQuantum.3.020328},
  eprint  = {2012.15825},
  archivePrefix = {arXiv},
  primaryClass  = {quant-ph},
  note    = {arXiv:2012.15825}
}

@article{Terhal2002Matchgates,
  title   = {Classical simulation of noninteracting-fermion quantum circuits},
  author  = {Terhal, Barbara M. and DiVincenzo, David P.},
  journal = {Physical Review A},
  volume  = {65},
  pages   = {032325},
  year    = {2002}
}

@article{Waksman1968,
  author  = {Waksman, Abraham}, title = {A Permutation Network},
  journal = {Journal of the ACM}, volume = {15}, number = {1}, pages = {159--163},
  year = {1968}, doi = {10.1145/321439.321449}
}

@misc{Grevink2025Glue,
  author = {Grevink, Lorenzo and Haferkamp, Jonas and Heinrich, Markus and
            Helsen, Jonas and Hinsche, Marcel and Schuster, Thomas and Zimbor\'{a}s, Zolt\'{a}n},
  title  = {Will it glue? {O}n short-depth designs beyond the unitary group},
  year   = {2025},
  eprint = {2506.23925},
  archivePrefix = {arXiv},
  primaryClass  = {quant-ph},
  doi = {10.48550/arXiv.2506.23925},
  note = {arXiv:2506.23925v3}
}

@misc{West2026StrongMatchgate,
  author = {West, Maxwell and Cerezo, M. and Larocca, Mart\'{i}n},
  title = {Strong matchgate designs in nearly optimal depth},
  year = {2026},
  eprint = {2609.26677},
  archivePrefix = {arXiv},
  primaryClass = {quant-ph},
  note = {\href{https://arxiv.org/abs/2609.26677}{arXiv:2609.26677}}
}

@misc{Nogo2025,
  author = {West, Maxwell and Garc\'{i}a-Mart\'{i}n, Diego and Diaz, N. L. and
            Cerezo, M. and Larocca, Martin},
  title  = {No-go theorems for sublinear-depth group designs},
  year   = {2025},
  eprint = {2506.16005},
  archivePrefix = {arXiv},
  primaryClass  = {quant-ph},
  note = {\href{https://arxiv.org/abs/2506.16005v2}{arXiv:2506.16005v2}, revised September 2026}
}

@misc{Braccia2026Commutant,
  title  = {The commutant of fermionic Gaussian unitaries},
  author = {Braccia, Paolo and Diaz, N. L. and Larocca, Martin and Cerezo, M.
            and Garc{\'i}a-Mart{\'i}n, Diego},
  year   = {2026},
  eprint = {2603.19210},
  archivePrefix = {arXiv},
  primaryClass  = {quant-ph},
  doi = {10.48550/arXiv.2603.19210},
  note = {arXiv:2603.19210}
}

@misc{Kerenidis2026Scalable,
  author        = {Kerenidis, Iordanis},
  title         = {Scalable Quantum Machine Learning: Trainability, Expressivity and
                   Efficiency},
  year          = {2026},
  month         = aug,
  eprint        = {2607.24014},
  archivePrefix = {arXiv},
  primaryClass  = {quant-ph},
  doi           = {10.48550/arXiv.2607.24014},
  note          = {arXiv:2607.24014v2}
}

@article{Landman2022QuantumMethods,
  author  = {Landman, Jonas and Mathur, Natansh and Li, Yun Yvonna and
             Strahm, Martin and Kazdaghli, Skander and Prakash, Anupam and
             Kerenidis, Iordanis},
  title   = {Quantum Methods for Neural Networks and Application to Medical
             Image Classification},
  journal = {Quantum},
  volume  = {6},
  pages   = {881},
  year    = {2022},
  doi     = {10.22331/q-2022-12-22-881},
  note    = {\href{https://arxiv.org/abs/2212.07389}{arXiv:2212.07389}}
}

@article{Coyle2020BornSupremacy,
  author  = {Coyle, Brian and Mills, Daniel and Danos, Vincent and Kashefi, Elham},
  title   = {The {Born} supremacy: quantum advantage and training of an
             {Ising Born} machine},
  journal = {npj Quantum Information},
  volume  = {6},
  pages   = {60},
  year    = {2020},
  doi     = {10.1038/s41534-020-00288-9},
  note    = {\href{https://arxiv.org/abs/1904.02214}{arXiv:1904.02214}}
}

@misc{Raj2026Generalization,
  author        = {Raj, Snehal and Mathur, Natansh and Perdomo-Ortiz, Alejandro},
  title         = {``{Train} classical, deploy quantum'' requires rethinking
                   generalization},
  year          = {2026},
  eprint        = {2608.31117},
  archivePrefix = {arXiv},
  primaryClass  = {quant-ph},
  doi           = {10.48550/arXiv.2608.31117},
  note          = {\href{https://arxiv.org/abs/2608.31117}{arXiv:2608.31117}}
}

@misc{DLMFHahn,
  title = {{NIST Digital Library of Mathematical Functions}, Sections 18.20--18.22},
  author = {{National Institute of Standards and Technology}},
  howpublished = {\href{https://dlmf.nist.gov/18.20}{dlmf.nist.gov/18.20}},
  note = {Hahn polynomials, duality, and recurrence relations; accessed September 2026}
}

@article{DiaconisShahshahani1987,
  author = {Diaconis, Persi and Shahshahani, Mehrdad},
  title = {Time to Reach Stationarity in the {Bernoulli--Laplace} Diffusion Model},
  journal = {SIAM Journal on Mathematical Analysis},
  volume = {18}, number = {1}, pages = {208--218}, year = {1987},
  doi = {10.1137/0518016}
}

@article{Ghosh2025RegularGraphs,
  author = {Ghosh, Soumik and Hangleiter, Dominik and Helsen, Jonas},
  title = {Random Regular Graph States Are Complex at Almost Any Depth},
  journal = {PRX Quantum}, volume = {6}, pages = {040344}, year = {2025},
  doi = {10.1103/52xz-3hpc},
  note = {\href{https://arxiv.org/abs/2412.07058v2}{arXiv:2412.07058v2}}
}

@misc{Shou2026GBS,
  author = {Shou, Laura and Ehrenberg, Adam and Wang, Yu-Xin and
            Iosue, Joseph T. and Gorshkov, Alexey V.},
  title = {Anticoncentration and entanglement in {Gaussian} boson sampling},
  year = {2026}, eprint = {2609.01241}, archivePrefix = {arXiv},
  howpublished = {\href{https://arxiv.org/abs/2609.01241}{arXiv:2609.01241}},
  note = {Preprint}
}

@misc{KoehlerLeung2026Permanent,
  author = {Koehler, Frederic and Leung, Pui Kuen},
  title = {Anticoncentration of the Permanent in {Ginibre} Ensembles},
  year = {2026}, eprint = {2607.20329}, archivePrefix = {arXiv},
  howpublished = {\href{https://arxiv.org/abs/2607.20329}{arXiv:2607.20329}},
  note = {Preprint}
}

@article{Heinrich2026Anticoncentration,
  author = {Heinrich, Markus and Haferkamp, Jonas and Roth, Ingo and Helsen, Jonas},
  title = {Anticoncentration Is (Almost) All You Need},
  journal = {Physical Review Letters}, volume = {137}, pages = {050601}, year = {2026},
  doi = {10.1103/z3mp-5gml},
  note = {\href{https://arxiv.org/abs/2510.23719v3}{arXiv:2510.23719v3}}
}

@article{Bouland2019RCS,
  author = {Bouland, Adam and Fefferman, Bill and Nirkhe, Chinmay and Vazirani, Umesh},
  title = {On the complexity and verification of quantum random circuit sampling},
  journal = {Nature Physics}, volume = {15}, pages = {159--163}, year = {2019},
  note = {arXiv:1803.04402}
}

@inproceedings{Kondo2021Robustness,
  author = {Kondo, Yasuhiro and Mori, Ryuhei and Movassagh, Ramis},
  title = {Quantum supremacy and hardness of estimating output probabilities of quantum circuits},
  booktitle = {2021 IEEE 62nd Annual Symposium on Foundations of Computer Science (FOCS)},
  year = {2021}, note = {arXiv:2102.01960}
}

@article{BrowneRudolph2005,
  author = {Browne, Daniel E. and Rudolph, Terry},
  title = {Resource-Efficient Linear Optical Quantum Computation},
  journal = {Physical Review Letters}, volume = {95}, pages = {010501}, year = {2005}
}

@article{Raussendorf2003MBQC,
  author = {Raussendorf, Robert and Browne, Daniel E. and Briegel, Hans J.},
  title = {Measurement-based quantum computation on cluster states},
  journal = {Physical Review A}, volume = {68}, pages = {022312}, year = {2003}
}

@article{Danos2007Calculus,
  author = {Danos, Vincent and Kashefi, Elham and Panangaden, Prakash},
  title = {The Measurement Calculus},
  journal = {Journal of the ACM}, volume = {54}, number = {2}, pages = {8}, year = {2007}
}

@inproceedings{Broadbent2009Blind,
  author = {Broadbent, Anne and Fitzsimons, Joseph and Kashefi, Elham},
  title = {Universal blind quantum computation},
  booktitle = {2009 50th Annual IEEE Symposium on Foundations of Computer Science},
  pages = {517--526}, year = {2009}, note = {arXiv:0807.4154}
}

@article{Dawson2005Polynomials,
  author = {Dawson, Christopher M. and Haselgrove, Henry L. and Hines, Andrew P. and
            Mortimer, Duncan and Nielsen, Michael A. and Osborne, Tobias J.},
  title = {Quantum computing and polynomial equations over the finite field {$\mathbb Z_2$}},
  journal = {Quantum Information \& Computation}, volume = {5}, pages = {102--112}, year = {2005}
}

@article{Aaronson2005Postselection,
  author = {Aaronson, Scott},
  title = {Quantum computing, postselection, and probabilistic polynomial-time},
  journal = {Proceedings of the Royal Society A}, volume = {461}, pages = {3473--3482}, year = {2005}
}

@misc{WelchBerlekamp1986,
  author = {Welch, Lloyd R. and Berlekamp, Elwyn R.},
  title = {Error correction for algebraic block codes},
  year = {1986}, note = {US Patent 4,633,470}
}

@misc{Quek2025Hamiltonian,
  author = {Quek, Yihui},
  title = {Quantum advantage from random geometrically-two-local {Hamiltonian} dynamics},
  year = {2025}, note = {arXiv:2510.06321v1}
}
}

\clearpage
\appendix
\addtocontents{toc}{\protect\setcounter{tocdepth}{1}}
\newgeometry{margin=0.75in}

\section{Haar benchmark}\label{app:proofs}

\noindent\textbf{Guide to the proofs.}
Appendices~\ref{app:proofs}--\ref{app:magic-analysis} prove the collision
results: the Haar scale, the exact chain and finite-gate replacement, and
the magic-input threshold. For an overview of the threshold proof, start
with Section~\ref{app:core-profile-proof}. Its uniform estimates and the
two-particle lower bound are proved in the surrounding subsections.

For hardness, Appendix~\ref{hard:sec:worst} constructs the depth-four
endpoint. Appendix~\ref{hard:sec:reduction} gives the exact reduction and
its additive-error extension for any efficiently sampleable schedule law
with a suitable embedding. Appendix~\ref{hard:sec:embeddings} supplies
the embeddings for random matching schedules and the deterministic
routing variant. Section~\ref{sec:numerics} summarizes the supporting checks.

Throughout, $n=2k$, $N=\binom nk$, and
$\mathcal F^{(k)}_n=\Lambda^k\mathbb C^n$ is the sector with exactly $k$
fermions. We write $\operatorname{Tr}$ for the operator trace and
$\|X\|_{\rm HS}^2=\operatorname{Tr}(X^\dagger X)$ for the squared
Hilbert--Schmidt norm. Constants in asymptotic bounds are independent of
$n$ unless their permitted dependencies are indicated by subscripts.

\begin{proof}[Proof of Theorem~\ref{thm:main}]
Use the graded identification of two copies with the Fock space on modes
$(i,c)$, $c\in\{1,2\}$. The mode action $\Gamma(V)^{\otimes2}$ commutes with
$E_{cc'}=\sum_i a^\dagger_{i,c}a_{i,c'}$, which generate the copy $SU(2)$
action $D(g)$. Ordering conventions change only fixed fermionic signs,
leaving probabilities and projector weights unchanged.

\emph{The singlet representation.} Skew Howe duality for $(GL_n,GL_2)$ on
$\Lambda(\CC^n\otimes\CC^2)$~\cite{Howe1989Duality},~\cite[Ch.~5]{GoodmanWallach2009}
gives the multiplicity-free decomposition into
$V^\lambda_{GL_n}\otimes W^{\lambda'}_{GL_2}$, where $V^\lambda$ and
$W^{\lambda'}$ are irreducible representations labelled by a partition
$\lambda$ and its transposed partition $\lambda'$, with $\lambda$ in an
$n\times2$ box. The notation $(2^k)$ means $k$ rows of length two.
At total particle number $2k$, the copy-singlet factor has
$\lambda'=(k,k)$ and is one dimensional. Thus its mode factor
$V_0=V^{(2^k)}_{GL_n}$ is a single irreducible $U(n)$ representation, contained in
the copy-occupation slice $(k,k)$. Write $P_0$ for its projector and $d_0$ for its
dimension.

\emph{Outcomes are singlets.} Write $\mathrm{supp}(y)=\{i:y_i=1\}$.
Up to sign
$|y\rangle\otimes|y\rangle=\prod_{i\in\mathrm{supp}(y)}a^\dagger_{i,1}a^\dagger_{i,2}|0\rangle$.
For the copy-spin generator $S_z=(E_{11}-E_{22})/2$, this vector has
$S_z=0$ and $E_{12}|yy\rangle=0$: a summand either
annihilates an empty mode or attempts to occupy an already occupied mode.
It is therefore a copy singlet.

\emph{Dimension.} Weyl's dimension formula for $\lambda=(2^k,0^k)$ telescopes to
\[
 d_0=\prod_{i=1}^k\prod_{j=k+1}^{2k}\frac{j-i+2}{j-i}
 =\prod_{m=1}^k\frac{(k+m)(k+m+1)}{m(m+1)}=(2k+1)\Cat_k^2.
\]
Here $\Cat_k=\binom{2k}k/(k+1)$ is the Catalan number; hence
$N^2/d_0=(k+1)^2/(2k+1)$.

\emph{Averaging.} Put $\varphi=P_0|\Psi\Psi\rangle$ and restrict
$\Gamma(V)^{\otimes2}$ to the representation $\rho_0(V)$ on $V_0$.
Because $P_0$ commutes with the circuit and $|yy\rangle\in V_0$,
$p_y(V)^2=|\langle yy|\rho_0(V)\varphi\rangle|^2$.
Schur averaging gives
$\EE_V[\rho_0(V)|\varphi\rangle\langle\varphi|\rho_0(V)^\dagger]
=(w_0/d_0)I_{V_0}$, where $w_0=\|\varphi\|^2$.
Hence $\EE_V[p_y^2]=w_0/d_0$ and
$\ratio=N^2w_0/d_0$, proving~\eqref{eq:main}.
For a Fock input, $|xx\rangle$ is itself a singlet, so $w_0=1$; this also
gives the Fock-input benchmark.
\end{proof}

\begin{proof}[Proof of Proposition~\ref{thm:magic}]
The generators $E_{cc'}$ are sums over modes and the blocks occupy disjoint mode sets, so
$D(g)$ factorizes over blocks and $w_0=\int_{SU(2)}\prod_b\langle v|D_b(g)|v\rangle\mathrm dg$
with $|v\rangle=|\Psi_4\rangle^{\otimes2}$ on one doubled block.
Put $u=|g_{11}|^2$. The two terms with doubled occupations contribute $1/2$ to
$\langle v|D(g)|v\rangle$. The two all-single terms contribute
$[u^2+(1-u)^2]/2$, including their cross terms. Therefore
$\langle v|D(g)|v\rangle=1/2+[u^2+(1-u)^2]/2=(3+x^2)/4$, with $x=2u-1$.
For Haar $g\in SU(2)$, $u$ is uniform on $[0,1]$, so $x$ is uniform on $[-1,1]$ and
$w_0(B)=\frac12\int_{-1}^1\bigl(\frac{3+x^2}4\bigr)^B\mathrm dx$; expanding the binomial and
integrating term by term gives the rational sum. Multiplication by $(k+1)^2/(2k+1)$ gives $\ratio_{\rm Haar}$.
Put $I_B=w_0(B)$. Since $x^2\le x$ on $[0,1]$,
\[
 I_B\le\int_0^1\left(\frac{3+x}{4}\right)^Bdx
 \le\frac4{B+1},\qquad
 \ratio_{\rm Haar}\le
 \frac{4(2B+1)^2}{(4B+1)(B+1)}<4.
\]
For the asymptotic benchmark and the first input coefficient, set
$x=1-z/B$ in the endpoint integral. The expansion
\[
 B\log[(3+x^2)/4]=-z/2+z^2/(8B)+O(B^{-2}z^3)
\]
gives $I_B=2/B+2/B^2+O(B^{-3})$. To justify termwise integration,
restrict first to $z\le B^{1/4}$ and use the Taylor remainder there;
the bound $((3+x^2)/4)^B\le e^{-B(1-x)/4}$ makes the remaining
integral exponentially small. Since $(2B+1)^2/(4B+1)$ differs from $B$
by a bounded amount,
$\ratio_{\rm Haar}=2+O(B^{-1})$.
\end{proof}

\paragraph{The collision observable in the singlet sector.}
Let $\Pi=\sum_y|yy\rangle\langle yy|$, $\varphi=P_0|\Psi\Psi\rangle$, and
$\rho_0=\Gamma(\,\cdot\,)^{\otimes2}|_{V_0}$. For
$\mathcal A_\nu=\mathbb E_V[\rho_0(V)^\dagger\Pi\rho_0(V)]$,
\begin{align}
 \sum_y p_y(V)^2&=\|\Pi\rho_0(V)\varphi\|^2,\label{eq:proj}\\
 Z(\nu,\Psi)&=\mathbb E_V\|\Pi\rho_0(V)\varphi\|^2,\qquad
 \ratio(\nu,\Psi)=N\langle\varphi|\mathcal A_\nu|\varphi\rangle.
 \label{eq:cartan}
\end{align}
Moreover, $\Tr\mathcal A_\nu=N$ and $\mathcal A_{\rm Haar}=(N/d_0)I$.
The vectors $|yy\rangle$ are orthonormal singlets and
$\langle yy|\rho_0(V)\varphi\rangle=\langle y|\Gamma(V)|\Psi\rangle^2$,
which gives both identities. Trace invariance gives $\Tr\mathcal A_\nu=\Tr\Pi=N$;
Schur averaging on $V_0$ gives the Haar value.

\section{Collision chain and finite-gate replacement}\label{app:matchingproof}

We derive the chain by identifying its occupation sectors, averaging one
matching through a generating function, and diagonalizing the resulting
polynomial action. We then connect it to the circuit collision and prove
that the finite gate alphabet preserves the same moments.

\subsection{The chain and its spectrum}\label{app:chain-spectrum}

\begin{lemma}[Exact random-matching chain]\label{lem:matchingchain}
Let $\mathsf M_n=U(1)^n\rtimes S_n$ be the group generated by independent
mode phases and mode permutations; $S_n$ is the permutation group and
$\rtimes$ denotes their semidirect product. There are irreducible sectors
$W_r$, $0\le r\le k$, in the restriction of $V_0$ to this group, with
\begin{equation}\label{eq:monomialrestriction}
 V_0\mathord\downarrow_{U(1)^n\rtimes S_n}=\bigoplus_{r=0}^kW_r,
 \qquad
 \delta_r=\dim W_r=\Cat_k\binom kr\binom{k+1}r.
\end{equation}
Let $Q_r$ be their projectors, let $\rho_0$ be the representation on $V_0$, and define the
one-layer twirl by
$\mathcal T(X)=\Ex_L[\rho_0(L)X\rho_0(L)^\dagger]$, with $L$ a random matching layer.
This twirl induces a reversible Markov matrix $A$ through
\begin{equation}\label{eq:matchingAdef}
 \mathcal T(Q_r/\delta_r)=\sum_{s=0}^kA_{sr}Q_s/\delta_s,
 \qquad
 \pi_r=\frac{\delta_r}{d_0}
 =\frac{\binom kr\binom{k+1}r}{\binom{2k+1}k}.
\end{equation}
For $q\ge0$, write
$x^{\underline q}:=x(x-1)\cdots(x-q+1)$ for the falling factorial and
$(x)_q:=x(x+1)\cdots(x+q-1)$ for the rising Pochhammer symbol, with both
equal to $1$ at $q=0$.  Put $m_j=\binom{2k+1}j-\binom{2k+1}{j-1}$, where
$\binom{2k+1}{-1}=0$.  The normalized Hahn functions
\begin{equation}\label{eq:matchingphi}
 \phi_j(r)=\sum_{q=0}^j(-1)^q\binom jq
 \frac{(2k+2-j)^{\underline q}r^{\underline q}}
 {k^{\underline q}(k+1)^{\underline q}}
\end{equation}
diagonalize $A$, with
\begin{equation}\label{eq:matchinglambda}
 \lambda_j=\sum_{q=0}^j(-1)^q\binom jq\left(\frac23\right)^q
 \frac{(2k+2-j)^{\underline q}k^{\underline q}}
 {(2k)^{\underline{2q}}}.
\end{equation}
Here restriction is denoted by $\downarrow$. With $\delta_{ij}$ equal to
one for $i=j$ and zero otherwise, the normalization and kernel identities are
\begin{gather}
 \sum_r\pi_r\phi_i(r)\phi_j(r)=\frac{\delta_{ij}}{m_j},
 \qquad |\phi_j(r)|\le1,                                      \label{eq:hahnnorm}\\
 \frac{(A^t)_{sr}}{\pi_s}
 =1+\sum_{j=1}^km_j\lambda_j^t\phi_j(s)\phi_j(r),             \label{eq:matchingkernel}\\
 0\le\lambda_j\le(5/6)^j.                                    \label{eq:matchingdecay}
\end{gather}
\end{lemma}

\begin{proof}[Proof of Lemma~\ref{lem:matchingchain}]

\emph{1. Why the state reduces to occupation sectors.}
The diagonal subgroup $U(1)^n$ is the torus. A torus weight in $V_0$ is a
word in $\{0,1,2\}^n$ of total weight $2k$.  If it contains $r$ zeros, it also contains
$r$ twos and $2(k-r)$ ones.  On a fixed weight the singly occupied modes are copy
doublets, and the global Howe-singlet space is the zero-spin space of these
$2(k-r)$ doublets.  Put $p=k-r$.  Under the ordinary tensor-factor action their
zero-spin space is the Specht module $S^{(p,p)}$, the irreducible symmetric-group
representation labelled by the partition $(p,p)$. Physical mode permutations act in the
fermionic exterior representation and carry the fermionic reordering (Koszul)
sign, so, writing $\operatorname{sgn}$ for the sign representation, in standard Specht
conventions, the stabilizer representation is instead the sign twist
\[
 S^{(p,p)}\otimes\operatorname{sgn}\cong S^{(2^p)},
 \qquad \dim S^{(2^p)}=\Cat_p.
\]
More explicitly, if $\chi_r$ is the torus character of the chosen weight
and $\operatorname{Ind}$ denotes induction of a representation from the
displayed subgroup, then
\[
 W_r\cong\operatorname{Ind}_{U(1)^n\rtimes(S_r\times S_r\times S_{2p})}^{\mathsf M_n}
 \bigl(\chi_r\otimes\mathbf 1\otimes\mathbf 1\otimes S^{(2^p)}\bigr).
\]
The Clebsch--Gordan calculation below uses the underlying invariant vector space
$\operatorname{Inv}(V_{1/2}^{\otimes2p})$, where $V_{1/2}$ is the
spin-half representation and $\operatorname{Inv}$ denotes its invariant
subspace; the sign twist changes its $S_{2p}$ label but
not its dimension, local spin projectors, or traces.
The little-group (Mackey) theorem~\cite{Mackey1952} therefore gives
\eqref{eq:monomialrestriction}; distinct $r$ have distinct torus orbits, so the
restriction is multiplicity free.  Cancelling factorials in the induced dimension gives
the second expression for $\delta_r$, and summing it gives
$d_0=(2k+1)\Cat_k^2$.

\emph{The one-layer chain.}
The monomial covariance of $\mathcal T$ preserves the span of the central projectors
$Q_r$.  Positivity and trace preservation make the coefficients in
\eqref{eq:matchingAdef} nonnegative with column sums one.  Since every Haar twirl is
self-adjoint in the Hilbert--Schmidt inner product,
$\delta_rA_{sr}=\delta_sA_{rs}$, so the stationary law is $\pi_r=\delta_r/d_0$.
Put $m=2k-2r$; a sector-$r$ weight
has $r$ double, $r$ empty, and $m$ singly occupied modes.  Let $R_0$ and $R_1$ denote
the input and output sector labels.

\emph{2. What one matching does to the sector label.}
Let $\chi=\chi_{1/2}$ and $\chi_1=\chi^2-1$ be the spin-half and spin-one
characters (traces of the corresponding representations) of $SU(2)$,
and let $D,E,U$ be formal variables marking double, empty, and single modes.
The local Clebsch--Gordan decomposition determines each edge weight.
Write $V_s^{\rm copy}$ for the spin-$s$ copy representation and
$\operatorname{Sym}^2$ for the symmetric square. Write $\mathsf U$ for the defining mode
representation on one physical edge; its single-copy Fock representation is
$\mathbf 1\oplus\mathsf U\oplus\det\mathsf U$.  In the two-particle subspace of two
copies, skew Howe duality gives
\[
 \operatorname{span}\{DE,SS_{\rm sing},ED\}
   \cong V_0^{\rm copy}\otimes\operatorname{Sym}^2\mathsf U,
 \qquad
 SS_{\rm trip}
   \cong V_1^{\rm copy}\otimes\det\mathsf U.
\]
Thus the two orientations of $DE$ and the copy-singlet direction of $SS$ form the three
weight lines of the irreducible mode representation $\operatorname{Sym}^2\mathsf U$.  Haar
twirling a fixed line gives $I_3/3$: two lines have one double output mode and the third
has none.  This gives $(1+2z)/3$.  The copy-triplet factor has mode representation
$\det\mathsf U$, so it cannot leave $SS$ and contributes its copy character $\chi_1$.  The
remaining particle-number sectors are even simpler: $DD$ and $EE$ are fixed, while
$DS$ and $ES$ always contain, respectively, one and zero double modes and carry one
copy doublet.  Hence the character-valued local weights are
\[
\begin{array}{c|cccccc}
\text{edge type}&DD&EE&DS&ES&DE&SS\\ \hline
\text{weight}&z^2&1&z\chi&\chi&(1+2z)/3&\chi_1+(1+2z)/3.
\end{array}
\]
At $z=1$ these weights are exactly the products of the input copy characters
$1,1,\chi,\chi,1,\chi^2$, as required by trace preservation.  The edge exponential is
\begin{align*}
 \mathcal E_z
 &=\frac{z^2D^2+E^2}{2}+z\chi DU+\chi EU
   +\frac{1+2z}{3}DE
   +\frac12\left(\chi^2+\frac{2(z-1)}3\right)U^2\\
 &=\frac{(zD+E+\chi U)^2}{2}+\frac{z-1}{3}(U^2-DE).
\end{align*}
To normalize the matching average, a fixed sector-$r$ torus weight has fiber
$\operatorname{Inv}(V_{1/2}^{\otimes m})$, of dimension $\Cat_{m/2}$, and there are
$(2k)!/(r!^2m!)$ such weights.  Because the matching law is permutation invariant, the
normalized trace on $Q_r$ may be evaluated at one fixed assignment of $r$ vertices of
type $D$, $r$ of type $E$, and $m$ of type $U$.  If $\mathrm{PM}_{2k}$ is the set of labeled
perfect matchings and $w_{\tau(e)}$ is the table entry for the type of edge $e$, then
\begin{equation}\label{eq:matchingfixedtrace}
 G_r(z)=\frac{1}{(2k-1)!!\,\Cat_{m/2}}
 \sum_{M\in\mathrm{PM}_{2k}}\int_{SU(2)}\prod_{e\in M}w_{\tau(e)}(z;g)\,dg .
\end{equation}
Indeed, the group integral is the trace of the global singlet projector, since
$\int_{SU(2)}\chi_{\mathcal V}(g)\,dg=\dim\operatorname{Inv}(\mathcal V)$.  The labeled
exponential formula turns the matching sum in~\eqref{eq:matchingfixedtrace} into
$r!^2m![D^rE^rU^m]\exp(\mathcal E_z)$; the factors $1/2$ on equal-type edges are their
automorphism factors.  Therefore, with $G_r(z)=\Ex[z^{R_1}\mid R_0=r]$,
\begin{equation}\label{eq:matchingpgf}
 G_r(z)=\frac{r!^2m!}{(2k-1)!!\,\Cat_{m/2}}
 [D^rE^rU^m]\int_{SU(2)}
 \exp\!\left\{\frac{(zD+E+\chi U)^2}{2}
 +\frac{z-1}{3}(U^2-DE)\right\}\,dg .
\end{equation}
At $z=1$, the coefficient is
\[
 [D^rE^rU^m]\exp\!\left(\frac{(D+E+\chi U)^2}{2}\right)
 =\frac{(2k)!\,\chi^m}{2^kk!\,r!^2m!}.
\]
Using $\int\chi^m=\Cat_{m/2}$ and $(2k-1)!!=(2k)!/(2^kk!)$ gives $G_r(1)=1$.

\emph{3. Why polynomials diagonalize the chain.}
The key is that one step preserves polynomial degree. To see the cancellation,
set $w=z-1$ and $C_r=r!^2m!/[(2k-1)!!\,\Cat_{k-r}]$.
Expanding only the correction exponential in~\eqref{eq:matchingpgf} gives
\[
 e^{w(U^2-DE)/3}
 =\sum_{u\ge0}\frac{w^u}{3^uu!}
   \sum_{q=0}^u(-1)^q\binom uqD^qE^qU^{2(u-q)}.
\]
For $a+b+c=2\ell$, the remaining Gaussian coefficient is
\begin{equation}\label{eq:matchinggaussiancoefficient}
 [D^aE^bU^c]e^{(\alpha D+E+\chi U)^2/2}
 =\frac{(2\ell)!}{2^\ell\ell!\,a!b!c!}\alpha^a\chi^c.
\end{equation}
Taking $\ell=k-u$, $a=b=r-q$, $c=m-2u+2q$, and $\alpha=1+w$ now gives
\begin{align}
 G_r(1+w)
 &=C_r\sum_{u\ge0}\frac{w^u}{3^u u!}
   \sum_{q=0}^u(-1)^q\binom uq(1+w)^{r-q}\notag\\
 &\quad\times
 \frac{(2k-2u)!\,\Cat_{k-r-u+q}}
 {2^{k-u}(k-u)!(r-q)!^2(m-2u+2q)!}.
 \label{eq:matchingpgfexpanded}
\end{align}
We used $\int_{SU(2)}\chi^{m-2u+2q}\,dg=\Cat_{k-r-u+q}$; terms with
negative factorial arguments are zero.

For functions $f$ on the sectors, write
$(A^*f)(r)=\sum_sA_{sr}f(s)=\Ex[f(R_1)\mid R_0=r]$.
For $0\le t\le k$, put
$M_t(r)=\Ex[R_1^{\underline t}\mid R_0=r]=t![w^t]G_r(1+w)$.
Using the Catalan factorial formula, the coefficient in
\eqref{eq:matchingpgfexpanded} simplifies to
\begin{align*}
 &C_r\frac{(2k-2u)!\,\Cat_{k-r-u+q}}
 {2^{k-u}(k-u)!(r-q)!^2(m-2u+2q)!}\\
 &\qquad=2^u\frac{k^{\underline u}}{(2k)^{\underline{2u}}}
 r^{\underline q}r^{\underline q}
 (k-r)^{\underline{u-q}}(k-r+1)^{\underline{u-q}}.
\end{align*}
Extract $[w^{t-u}](1+w)^{r-q}$ and use
$r^{\underline q}(r-q)^{\underline{t-u}}=r^{\underline{q+t-u}}$.
Renaming $u$ as $s$ gives
\begin{align}
 M_t(r)&=\sum_{s=0}^t\binom ts\left(\frac23\right)^s
 \frac{k^{\underline s}}{(2k)^{\underline{2s}}}H_{s,t-s}(r),
 \label{eq:matchingmoment}\\
 H_{s,h}(r)&=\sum_{q=0}^s(-1)^q\binom sq
 r^{\underline q}r^{\underline{q+h}}
 (k-r)^{\underline{s-q}}(k-r+1)^{\underline{s-q}}.
 \label{eq:matchingH}
\end{align}
For parameters $a_1,\ldots,a_p$, $b_1,\ldots,b_q$, the hypergeometric series is
\[
 {}_pF_q(a_1,\ldots,a_p;b_1,\ldots,b_q;z)
 =\sum_{h\ge0}\frac{(a_1)_h\cdots(a_p)_h}{(b_1)_h\cdots(b_q)_h}
 \frac{z^h}{h!}.
\]
All series used below terminate at a nonpositive integer upper parameter.
The cancellation that fixes the degree is the terminating Hahn identity
\begin{equation}\label{eq:matchinghahnidentity}
\begin{split}
 &\sum_{q=0}^s(-1)^q\binom sq
 x^{\underline q}(x+A)^{\underline q}
 (N-x)^{\underline{s-q}}(N-x+1)^{\underline{s-q}}\\
 &\quad=N^{\underline s}(N+1)^{\underline s}
 {}_3F_2\!\left(\begin{matrix}-s,s-2N-A-2,-x\\-N,-N-1\end{matrix};1\right).
\end{split}
\end{equation}
Apply the terminating transformation
\begin{equation}\label{eq:matching3f2transform}
 {}_3F_2\!\left(\begin{matrix}-s,a,b\\d,e\end{matrix};1\right)
 =\frac{(d-a)_s}{(d)_s}
 {}_3F_2\!\left(\begin{matrix}-s,a,e-b\\a-d-s+1,e\end{matrix};1\right),
\end{equation}
from~\cite[Ch.~2]{AndrewsAskeyRoy1999}.
Using $y^{\underline{s-q}}=y^{\underline s}/(y-s+1)_q$, the left side of
\eqref{eq:matchinghahnidentity} first becomes
\[
 (N-x)^{\underline s}(N-x+1)^{\underline s}
 {}_3F_2\!\left(
 \begin{matrix}-s,-x,-x-A\\N-x-s+1,N-x-s+2\end{matrix};1\right).
\]
Apply~\eqref{eq:matching3f2transform} twice, with parameter tuples
\begin{align*}
 (a,b,d,e)&=(-x,-x-A,N-x-s+1,N-x-s+2),\\
 (a,b,d,e)&=(-x,N-s+A+2,N-x-s+2,-N).
\end{align*}
The prefactors cancel to give~\eqref{eq:matchinghahnidentity}.
These steps prove the identity at generic parameters;
clearing denominators extends it to exceptional integers.  In our application
$N=k-(t-s)\ge s$, so the final terminating series is well defined.
Indeed~\eqref{eq:matchingH} is $r^{\underline h}$ times
\eqref{eq:matchinghahnidentity} with $x=r-h$, $N=k-h$ and $A=h$.  Therefore
$H_{s,h}$ has degree at most $s+h$.  Only the last term of the terminating series can
contribute degree $s$ in $x$, and its coefficient is
\[
 (s-2N-A-2)_s=(-1)^s(2N+A+2-s)^{\underline s}.
\]
Thus the leading coefficient of $H_{s,h}$ is
$(-1)^s(2k+2-s-h)^{\underline s}$.  Since $h=t-s$, taking the coefficient of $r^t$
in~\eqref{eq:matchingmoment} gives
\begin{align*}
 [r^t]M_t(r)
 &=\sum_{s=0}^t(-1)^s\binom ts\left(\frac23\right)^s
 \frac{(2k+2-t)^{\underline s}k^{\underline s}}
 {(2k)^{\underline{2s}}}
 =\lambda_t.
\end{align*}
Equation~\eqref{eq:matchingmoment} consequently shows that $A^*$ preserves every
polynomial flag and has diagonal entries $\lambda_t$ in the falling-factorial basis.
Reversibility makes
$A^*$ self-adjoint in $L^2(\pi)$.  The orthogonal complement of the degree-$(j-1)$
polynomial flag inside the degree-$j$ flag is therefore a one-dimensional eigenspace.  The Johnson scheme identifies its normalized polynomial.

\emph{Hahn basis and Johnson-scheme normalization.}
For the Johnson scheme on $\mathcal X=\binom{[2k+1]}k$, the set of
$k$-subsets of $\{1,\ldots,2k+1\}$, let $d(x,y)=k-|x\cap y|$.
A sphere of
radius $r$ has size $\binom kr\binom{k+1}r$, so its radial law is exactly $\pi$.
The degree-$j$ primitive idempotent $E_j$ has rank $m_j$ and kernel
\cite[Sec.~9.1]{BrouwerCohenNeumaier1989}
\[
 (E_j)_{xy}=\frac{m_j}{|\mathcal X|}\phi_j(d(x,y)),\qquad
 \phi_j(r)={}_3F_2\!\left(\begin{matrix}-j,j-2k-2,-r\\-k,-k-1\end{matrix};1\right).
\]
This is~\eqref{eq:matchingphi}, normalized by $\phi_j(0)=1$, and equals
$Q_j(r;-k-2,-k-1,k)$ in Hahn notation~\cite[Sec.~18.20]{DLMFHahn}.
Evaluating $E_iE_j=\delta_{ij}E_j$ at a diagonal entry gives
$\sum_r\pi_r\phi_i(r)\phi_j(r)=\delta_{ij}/m_j$.
Positivity and the constant diagonal $m_j/|\mathcal X|$ give
$|(E_j)_{xy}|\le m_j/|\mathcal X|$, hence $|\phi_j(r)|\le1$.
These checks establish~\eqref{eq:hahnnorm} with our normalization and identify
$\phi_j$ with the degree-$j$ eigenfunction above.
This is the same Johnson-scheme spectral basis
used in Bernoulli--Laplace diffusion~\cite{DiaconisShahshahani1987}, not the
same transition law: $A$ describes an entire random matching, rather than one exchange.

\emph{Kernel and collision coordinate.}
The $j=0$ term has $\phi_0=1$ and $\lambda_0=1$.  Expanding the reversible kernel in the
complete orthogonal basis above, whose squared norms are $1/m_j$, gives
\eqref{eq:matchingkernel}.  Finally, at $r=k$ there are no singly occupied
modes and $\Cat_0=1$, so $\delta_k=\binom{2k}k=N$ and $W_k$ is precisely the span of the
doubled outcomes $|yy\rangle$; hence $Q_k=\Pi$, as used in
\eqref{eq:matchingcollisioncoordinate}.

\emph{4. Why high degrees decay uniformly.}
We establish~\eqref{eq:matchingdecay}.  The one-layer
channel is an average of Haar twirls, hence its compression to the monomial sectors is
positive semidefinite and $\lambda_j\ge0$. Fix a degree $d$, put
$h=k-d$ and $a=h+\tfrac12$. Pfaff's transformation
\cite[Ch.~2]{AndrewsAskeyRoy1999} applied to
\eqref{eq:matchinglambda} gives
\begin{equation}\label{eq:matchingpfaff}
 \lambda_d=\left(\frac56\right)^d g_d(a),
 \qquad
 g_q(a)={}_2F_1\!\left(-q,a+2;1-a-q;-\frac15\right).
\end{equation}
Let $b_q=(a)_q/q!$ and $c_q=b_qg_q(a)$. Direct expansion gives
\begin{equation}\label{eq:matchingcgf}
 \sum_{q\ge0}c_qz^q=(1-z)^{-a}(1+z/5)^{-a-2}.
\end{equation}
Its logarithmic derivative yields
\begin{equation}\label{eq:matchinggrec}
 5(a+q)g_{q+1}
 =(4q+4a-2)g_q+\frac{q(q+2a+1)}{a+q-1}g_{q-1}.
\end{equation}
The coefficients are nonnegative. If $h\ge1$, then $a\ge3/2$, and the coefficient on the
left exceeds the sum of those on the right by
\begin{equation}\label{eq:matchingrecdeficit}
 \frac{(a+2)(a-1)}{a+q-1}\ge0.
\end{equation}
Since $g_0=1$ and $g_1=(4a-2)/(5a)\le1$, induction gives $0\le g_q\le1$.

For the remaining case $h=0$, $a=1/2$, one has $g_1=0$, $g_2=4/5$, and
\[
 g_{q+1}=\frac{8q}{5(2q+1)}g_q
 +\frac{4q(q+2)}{5(2q+1)(2q-1)}g_{q-1}.
\]
The sequence $L_q=1-3/(10q)$ is a strict supersolution for $q\ge2$, because
\begin{equation}\label{eq:matchingsupersolution}
 L_{q+1}-\frac{8q}{5(2q+1)}L_q
 -\frac{4q(q+2)}{5(2q+1)(2q-1)}L_{q-1}
 =\frac{22q^2-9q+59}{50(2q+1)(2q-1)(q-1)(q+1)}>0.
\end{equation}
The two base cases complete the induction. Thus $0\le g_d(a)\le1$ in every case, and
\eqref{eq:matchingpfaff} proves~\eqref{eq:matchingdecay} for every $k\ge1$
and $0\le d\le k$, for every system size.

\end{proof}

\subsection{Proof of classical relabeling}\label{app:relabeling-proof}
\begin{proof}[Proof of Proposition~\ref{prop:permutationremoval}]
For a fixed permutation $P$, put $L'_i=P^{-1}L_iP$ and
$C'=L'_t\cdots L'_1$. Relabeling preserves the matching law; reversing an edge
conjugates its gate by $X$, preserving both Haar measure and the uniform law
on $\mathcal F$. Thus $C'$ has the core law, while
\[
 L_t\cdots L_1P=PC'.
\]
Second quantization maps basis occupations to their permuted occupations up to
sign. Measuring $C'$ and reporting $Pz$ therefore reproduces each instance's
conditional probabilities. Each $L'_i$ is still one matching layer.
Since permuting a probability vector preserves its collision, averaging over
$P$ proves~\eqref{eq:permutationcollision}, including $t=0$.
\end{proof}

\subsection{From the circuit to the collision coordinate}\label{app:collision-coordinate}
Write $\rho=\rho_0$ and define the monomial twirl
\[
 \mathcal M(X)=\Ex_{D,P}\!\left[\rho(DP)X\rho(DP)^\dagger\right],
\]
where $D$ has independent uniform diagonal phases and $P$ is uniform in $S_n$.

Let $\varphi=P_0|\Psi\Psi\rangle$, so $\|\varphi\|^2=w_0$; for a normalized input,
Theorem~\ref{thm:main} and $\ratio_{\rm Haar}\ge1$ imply $w_0>0$.  Since the restriction in
\eqref{eq:monomialrestriction} is multiplicity free, Schur orthogonality gives the explicit
twirl identity
\begin{equation}\label{eq:monomialtwirlexplicit}
 \mathcal M(|\varphi\rangle\langle\varphi|)
 =\sum_{r=0}^k\frac{\langle\varphi|Q_r|\varphi\rangle}{\delta_r}Q_r
 =w_0\sum_{r=0}^k b_r\frac{Q_r}{\delta_r},
 \qquad
 b_r:=\frac{\|Q_r\varphi\|^2}{w_0}.
\end{equation}
Because the projectors sum to $I_{V_0}$, $b_r\ge0$ and
\begin{equation}\label{eq:matchingbnormalization}
 \sum_{r=0}^k b_r=\frac{\|\varphi\|^2}{w_0}=1.
\end{equation}

For $m\in\mathsf M_n$, conjugation by its permutation part sends a uniform perfect
matching to a uniform perfect matching, while conjugation by its diagonal part preserves
the Haar law of every edge gate.  Consequently
\begin{equation}\label{eq:matchingcovariance}
 \mathcal T\!\left(\rho(m)X\rho(m)^\dagger\right)
 =\rho(m)\mathcal T(X)\rho(m)^\dagger.
\end{equation}
Thus $\mathcal T$ commutes with $\mathcal M$ and preserves the span of the $Q_r$.
The collision is unchanged by inserting this virtual twirl: since
$\mathcal M(\Pi)=\Pi$ and $\mathcal M$ is self-adjoint,
\begin{equation}\label{eq:virtualmonomialtwirl}
 \Tr[\Pi\mathcal T^t(X)]
 =\Tr[\Pi\mathcal M\mathcal T^t(X)]
 =\Tr[\Pi\mathcal T^t\mathcal M(X)].
\end{equation}
Thus the monomial twirl may be inserted when computing the core's collision, even
though its full output-state channel need not equal $\mathcal T^t\mathcal M$.
Applying
\eqref{eq:matchingAdef} repeatedly to~\eqref{eq:monomialtwirlexplicit} therefore yields
\begin{equation}\label{eq:matchingstateevolution}
 \mathcal T^t\mathcal M(|\varphi\rangle\langle\varphi|)
 =w_0\sum_{s=0}^k(A^tb)_s\frac{Q_s}{\delta_s}.
\end{equation}

Finally, $Q_k=\Pi$ and $\delta_k=N$ imply
$\Tr(\Pi Q_s)=N\mathbf1_{s=k}$. Equation~\eqref{eq:cartan} and
Proposition~\ref{prop:permutationremoval} therefore give
\begin{equation}\label{eq:matchingtracebridge}
 \ratio(\mathsf C_{n,t},\Psi)=\ratio(\mathsf G_{n,t},\Psi)
 =Nw_0(A^tb)_k.
\end{equation}
On the other hand, $\ratio_{\rm Haar}=N^2w_0/d_0$ and
$\pi_k=\delta_k/d_0=N/d_0$.  Dividing~\eqref{eq:matchingtracebridge} by the Haar value gives
\begin{equation}\label{eq:matchingcollisioncoordinate}
 \frac{\ratio(\mathsf G_{n,t},\Psi)}{\ratio_{\rm Haar}(n,\Psi)}
 =\frac{(A^tb)_k}{\pi_k}.
\end{equation}

\subsection{Exact finite-gate replacement}\label{app:finitegates}

\begin{proof}[Proof of Lemma~\ref{lem:cliffordtwirl}]
Write $U$ for the defining representation of $U(2)$.  On one copy of the local
two-mode Fock space, $F_{ij}\cong 1\oplus U\oplus\det$.  To match the twirls, it suffices to match their commutants. The two-copy
representation decomposes as
\begin{equation}\label{eq:hardapp-local-clifford-decomposition}
 R\cong 1\oplus2U\oplus3\det\oplus\operatorname{Sym}^2U
 \oplus2(\det\otimes U)\oplus\det^2.
\end{equation}
The group $\mathcal F$ in~\eqref{eq:single-particle-clifford} is
$C_8\mathbin{\cdot}2O$, where $C_8=\{e^{\pi i m/4}I:0\le m<8\}$
is the scalar phase group and $2O\subset SU(2)$ is the binary octahedral group of
order $48$ and $C_8\cap2O=\{\pm I\}$.  Thus $|\mathcal F|=192$.  Its scalar $C_8$
acts on the six summand types in~\eqref{eq:hardapp-local-clifford-decomposition} with
charges $0,1,2,2,3,4$.  All charges are distinct modulo eight except the two
charge-two types.  The defining
two-dimensional representation remains irreducible on $2O$, while
$\operatorname{Sym}^2U$ descends to the irreducible three-dimensional rotation
representation of the octahedral group and is inequivalent to the one-dimensional
$\det$ representation in charge two.  Hence all six $U(2)$ types remain irreducible and
pairwise inequivalent after restriction to $\mathcal F$, with the same multiplicities
$1,2,3,1,2,1$.

It follows that the commutants of $R(U(2))$ and $R(\mathcal F)$ coincide; both have
dimension
\[
 1^2+2^2+3^2+1^2+2^2+1^2=20.
\]
Haar averaging over a compact group and uniform averaging over a finite group are the
Hilbert--Schmidt orthogonal projections onto their respective commutants, which proves
\eqref{eq:clifford-fock-twirl}. 

The group contains $X=HS^2H$, so relabeling either endpoint of a gate
preserves its law. The identity tensorizes over disjoint edges and composes
over layers. Thus the complete two-copy channel, including the classical
mode symmetrization, agrees with the Haar-gated channel.
\end{proof}

\section{Magic-input profile and sharp threshold}\label{app:magic-analysis}

The exact chain separates the input coefficients from the eigenvalues.
Section~\ref{app:sharpthreshold} computes the former and sharpens the
latter near the transition. Section~\ref{app:magicthreshold} sums the
spectrum uniformly to obtain the profile; Section~\ref{app:magiclower}
rules out earlier threshold crossings.

\subsection{Input coefficients and spectral estimates}
\label{app:sharpthreshold}

\subsubsection{The exact spectral expansion}

Combining \eqref{eq:matchingkernel} with \eqref{eq:matchingcollisioncoordinate} and
$\sum_rb_r=1$ gives
\begin{equation}\label{eq:sharpexpansion}
 \frac{\ratio(\mathsf G_{n,t},\Psi)}{\ratio_{\rm Haar}(n,\Psi)}-1
 =\sum_{j=1}^{k}c_j\lambda_j^{\,t},
 \qquad
 c_j:=m_j\,\phi_j(k)\,B_j,
 \qquad
 B_j:=\sum_{r}b_r\phi_j(r).
\end{equation}
All spectral quantities are rational for the magic input. In particular,
\begin{equation}\label{eq:lam12}
 \lambda_1=1-\frac{2k+1}{3(2k-1)}\longrightarrow\frac23,
 \qquad
 \lambda_2\longrightarrow\frac49.
\end{equation}
At the boundary coordinate, cancelling the common $-k$ parameter in the
${}_3F_2$ expression for $\phi_j$ and applying Chu--Vandermonde gives the exact identity
\begin{equation}\label{eq:phiboundary}
 \phi_j(k)
 ={}_2F_1\!\left(\begin{matrix}-j,j-2k-2\\-k-1\end{matrix};1\right)
 =(-1)^j\frac{k+1-j}{k+1}.
\end{equation}

\subsubsection{The exact initial sector law of the magic input}

\begin{proposition}[Exact sector law]\label{prop:sectorlaw}
Let $n=4B$ and $\Psi=|\Psi_4\rangle^{\otimes B}$.  Then $b_r=0$ for odd $r$, and
\begin{equation}\label{eq:sectorlaw}
 b_{2m}\,w_0(B)=\binom Bm 2^{-m}\,
 \frac12\int_{-1}^{1}\!\left(\frac{1+x^2}{4}\right)^{\!B-m}\!dx
 =\binom Bm 2^{-m}4^{-(B-m)}
   \sum_{\ell=0}^{B-m}\frac{\binom{B-m}\ell}{2\ell+1},
 \qquad 0\le m\le B.
\end{equation}
\end{proposition}

\begin{proof}
Write $P_0=\int_{SU(2)}D(g)\,dg$ and $Q_r=\Pi_rP_0=P_0\Pi_r$, where $\Pi_r$ projects
onto configurations with $r$ doubly occupied modes.  Then
$b_rw_0=\langle\Psi\Psi|\Pi_rP_0|\Psi\Psi\rangle$.  Since each four-mode block has even
degree across the two copies, distinct blocks commute and
\begin{equation}\label{eq:blockfactor}
 \langle\Psi\Psi|z^R D(g)|\Psi\Psi\rangle=\mathcal F(z,g)^B,
 \qquad
 \mathcal F(z,g)=\sum_{r_{\rm b}}z^{r_{\rm b}}
 \langle\psi\psi|\Pi_{r_{\rm b}}D(g)|\psi\psi\rangle,
\end{equation}
where $|\psi\rangle=|\Psi_4\rangle$ and $R$ counts doubly occupied modes.

The four terms of $|\psi\psi\rangle$ have mode-occupation profiles $(0,0,2,2)$,
$(2,2,0,0)$ and twice $(1,1,1,1)$, so
$\mathcal F(z,g)=\mathcal F_0(g)+z^2\mathcal F_2(g)$.  The two profiles contributing to
$\mathcal F_2$ do not mix, and each contributes $(\det g)^2=1$ with squared input
coefficient $1/4$.  Hence $\mathcal F_2=1/2$.  The direct block calculation in the proof
of Proposition~\ref{thm:magic} gives
$\mathcal F(1,g)=(3+x^2)/4$, with $x=2|g_{11}|^2-1$.  Therefore, pointwise,
\begin{equation}\label{eq:F02direct}
 \mathcal F_0(g)=\frac{1+x^2}{4},\qquad \mathcal F_2(g)=\frac12.
\end{equation}
For Haar $g\in SU(2)$, $x$ is uniform on $[-1,1]$.  Extracting the coefficient of
$z^{2m}$ in $\bigl(\mathcal F_0+z^2/2\bigr)^B$ and integrating proves
\eqref{eq:sectorlaw}.  Summing over $m$ returns $w_0(B)$.
\end{proof}

Let $I_B=w_0(B)$ and let $R$ be the random sector index with
$\Pr(R=r)=b_r$; $\mathbb E_b$ denotes expectation under this law.
Proposition~\ref{prop:sectorlaw} equivalently gives the probability
generating function
\begin{equation}\label{eq:magicsectorpgf}
 \sum_{r=0}^{k}b_rz^r
 =\frac1{2I_B}\int_{-1}^{1}
 \left(\frac{1+x^2}{4}+\frac{z^2}{2}\right)^Bdx.
\end{equation}
Differentiating at $z=1$ yields
$\Ex_bR=B I_{B-1}/I_B$.  Since
$\phi_1(r)=1-(2k+1)r/[k(k+1)]$, the endpoint expansion for $I_B$ proved with
Proposition~\ref{thm:magic} gives
\begin{equation}\label{eq:B1magic}
 B_1=-\frac3n+O(n^{-2}),
 \qquad
 c_1=m_1\phi_1(k)B_1=3+O(n^{-1}).
\end{equation}
Thus the magic input removes the $\Theta(n)$ amplification of the first mode; it does not
remove the mode itself.

\subsubsection{Quantitative low-degree eigenvalues}

\begin{lemma}[Quantitative spectral bound]\label{lem:quantmatching}
For $1\le j\le k$, with $a=k-j+1/2$,
\begin{equation}\label{eq:quant-gap-target}
 0\le\lambda_j\le(2/3)^j\exp\!\left(\frac{5j(j-1)}{16a}\right).
\end{equation}
For $j\le k/2$ the exponent is at most $5j(j-1)/(8k)$.
Moreover, uniformly for $1\le j\le k^{1/3}$ as $k\to\infty$,
\begin{equation}\label{eq:matching-eigenvalue-comparison}
 \left|\log\frac{\lambda_j}{(2/3)^j}\right|\le\frac{Cj^2}{n},
 \qquad n=2k,
\end{equation}
for an absolute $C$.
\end{lemma}

\begin{proof}
Normalize the Pfaff recurrence~\eqref{eq:matchinggrec} by
$u_q=(5/4)^qg_q(a)$. Then $\lambda_j=(2/3)^ju_j$,
$u_0=1$, $u_1=1-1/(2a)\in[0,1]$, and, for $q\ge1$,
\[
 u_{q+1}=\left(1-\frac1{2(a+q)}\right)u_q
 +\frac{5q(q+2a+1)}{16(a+q)(a+q-1)}u_{q-1}.
\]
The first coefficient lies in $[0,1]$ and the second in $[0,5q/(8a)]$, since
\[
 2(a+q)(a+q-1)-a(q+2a+1)=(q-1)(3a+2q)\ge0.
\]
The nondecreasing majorant $M_0=M_1=1$,
$M_{q+1}=M_q(1+5q/(8a))$ therefore bounds $u_q$ by induction. Hence
\[
 u_j\le\prod_{q=1}^{j-1}\left(1+\frac{5q}{8a}\right)
 \le\exp\!\left(\frac{5j(j-1)}{16a}\right).
\]
For $j\le k/2$ use $a\ge k/2$.

For the logarithmic comparison restrict to $j\le k^{1/3}$ and
large $k$, so $a\ge k/2$. Dropping the nonnegative second term of the
recurrence gives
\[
 u_j\ge\prod_{q=0}^{j-1}\left(1-\frac1{2(a+q)}\right)
 \ge e^{-Cj/k}.
\]
Combine this with the upper bound $u_j\le e^{Cj^2/k}$ to obtain
\eqref{eq:matching-eigenvalue-comparison}.
\end{proof}

\subsection{The collision profile and depth certificate}
\label{app:magicthreshold}

Throughout, $n=4B$, $k=2B$, and $c_j$ are the coefficients in
\eqref{eq:sharpexpansion}. We first assemble the main proof using the estimates
proved below.

\subsubsection{Proof of the collision profile and depth threshold}\label{app:core-profile-proof}
\begin{proof}[Proof of Theorem~\ref{thm:magicprofile}]
Put $z=n(4/9)^t$. Lemma~\ref{lem:magic-spectral-sums}, proved below,
shows that both $1+\sum_jc_j\lambda_j^t$ and $1+\sum_j|c_j|\lambda_j^t$
equal $e^{3z/2}(1+O_Z(n^{-1/2}))$, uniformly for $0<z\le Z$.
The first sum is the Haar-relative collision and the second bounds it above,
so the collision tends to Haar when $z\to0$.
For fixed $q>1$, solving $3z/2=\log q$ gives $t=\tau_n(q)$.
Lemma~\ref{lem:magic-early-lower} excludes all depths before a bounded window
about this value. Inside that window the slope of $3n(4/9)^t/2$ is bounded
away from zero. The uniform error therefore brackets each first crossing
within $O_q(n^{-1/2})$ of $\tau_n(q)$, giving the integer rounding in
\eqref{eq:magicsharpthreshold}. The same lower bound gives divergence for
$z\to\infty$; no monotonicity of the collision is assumed.

For the certificate, the absolute sum is nonincreasing and bounds the
signed sum above. The same upper bracket therefore applies at $q=2$, while
the lower bound for the actual collision excludes earlier certificates.
This proves the rounding assertion used in Corollary~\ref{cor:certified-ac}.
\end{proof}

\subsubsection{Uniform coefficient bounds}\label{app:coefficient-control}
We bound the polynomial roots and the input fluctuations separately;
Lemma~\ref{lem:magic-coeff-majorant} combines them into a summable coefficient bound.

Write $\mu=k(k+1)/(2k+1)$ and
\[
 \ell_j=\frac{(2k+2-j)^{\underline j}}
 {k^{\underline j}(k+1)^{\underline j}},\qquad
 P_j(r)=\frac{(-1)^j}{\ell_j}\phi_j(r).
\]
Thus $P_j$ is monic of degree $j$, and $\mu$ is the stationary mean
overlap. The random variable $R$ below has the magic-input law $b$,
not the stationary law $\pi$.

\begin{lemma}[Hahn root localization]\label{lem:magic-hahn-roots}
For $1\le j\le k/2$, every zero $\zeta$ of $P_j$ satisfies
\[
 |\zeta-\mu|\le\sqrt{kj}.
\]
\end{lemma}
\begin{proof}
Use the standard Hahn recurrence~\cite[Eqs.~18.22.2--18.22.3]{DLMFHahn}
with $(\alpha,\beta,N)=(-k-2,-k-1,k)$:
\[
 r\phi_q(r)=-A_q\phi_{q+1}(r)+(A_q+C_q)\phi_q(r)-C_q\phi_{q-1}(r),
\]
where
\[
 A_q=\frac{(2k+2-q)(k-q)}{2(2k+1-2q)},\qquad
 C_q=\frac{q(k+2-q)}{2(2k+3-2q)},\qquad C_0=0.
\]
Although these parameters lie outside the usual positive-parameter range,
the recurrence is a rational identity and extends after clearing denominators.
All denominators are nonzero for $0\le q<j\le k/2$; hence the specialization is valid.
The coefficient of the leading term also gives
$\ell_{q+1}=\ell_q/A_q$. Consequently
\[
 rP_q=P_{q+1}+\alpha_qP_q+\beta_qP_{q-1},\qquad
 \alpha_q=A_q+C_q,\quad \beta_q=A_{q-1}C_q.
\]
Put $u=k+1-q$. Direct simplification gives
\begin{equation}\label{eq:magic-jacobi-coefficients}
 \alpha_q=\frac k2+\frac14-\frac{2k+3}{4(4u^2-1)},\qquad
 \beta_q=\frac{qu(u+1)(k+u+2)}{4(2u+1)^2}
 \le\frac{(k+2)q}{8}\le\frac{kq}{4}.
\end{equation}
The first inequality follows because
$(k+2)(2u+1)^2-2u(u+1)(k+u+2)
=2qu^2+(k+2)(2u+1)\ge0$.
For $q\le k/2$, the diagonal formula gives
$|\alpha_q-\mu|<1/4$, and $\alpha_0=\mu$.
The $j$-by-$j$ real symmetric tridiagonal matrix with diagonal
$\alpha_0,\ldots,\alpha_{j-1}$ and adjacent entries
$\sqrt{\beta_1},\ldots,\sqrt{\beta_{j-1}}$ has characteristic polynomial $P_j$,
by expansion of its leading principal minors. Its eigenvalues therefore are
the zeros of $P_j$. The row-sum bound after subtracting $\mu I$ gives
\[
 |\zeta-\mu|\le\frac14+\sqrt{k(j-1)}\le\sqrt{kj}\qquad(j\ge2),
\]
since $\sqrt{kj}-\sqrt{k(j-1)}\ge\sqrt{k}/(2\sqrt j)\ge1/\sqrt2$.
For $j=1$, the sole root is $\mu$.
\end{proof}

\begin{lemma}[Magic-sector moments]\label{lem:magic-sector-moments}
For every integer $1\le j\le B$,
\[
 \left(\mathbb E_b|R-\mu|^j\right)^{1/j}\le3\sqrt{Bj}.
\]
\end{lemma}
\begin{proof}
For a random variable $W$, use $\|W\|_j=(\mathbb E|W|^j)^{1/j}$.
Put $h(x)=(3+x^2)/4$ and $I_B=\int_0^1h(x)^Bdx$. The exact sector
generating function~\eqref{eq:magicsectorpgf} gives the following probability-mixture representation:
sample $X$ on $[0,1]$ with density $h(x)^B/I_B$, and, conditional on $X=x$,
sample $Y\sim\operatorname{Bin}(B,p_x)$, the binomial law with $B$ trials
and success probability $p_x=2/(3+x^2)$; then $R=2Y$.
Write
\[
 D=\mathbb E[R\mid X]-B=\frac{B(1-X^2)}{3+X^2},\qquad
 S=R-\mathbb E[R\mid X].
\]
Weighted arithmetic--geometric mean gives $h(x)\ge x^{1/2}$, hence
$I_B\ge2/(B+2)$. Also
\[
 h(x)^B\le e^{-B(1-x)/4},\qquad D\le(2B/3)(1-X).
\]
For every integer $j\ge1$, integration gives
\begin{align*}
 \mathbb E D^j
 &\le\frac{B+2}{2}\left(\frac{2B}{3}\right)^j
       \int_0^1(1-x)^je^{-B(1-x)/4}\,dx\\
 &\le(2+4/B)(8/3)^j j!\le6(8/3)^j j!.
\end{align*}
For $2\le j\le B$, arithmetic--geometric mean gives
$(j!)^{1/j}\le(j+1)/2\le3j/4$. Since $B\ge j\ge2$, the sharper
prefactor above obeys $(2+4/B)^{1/j}\le2$; hence
$\|D\|_j\le2(8/3)(3j/4)=4j$, including $j=2$.
Together with $D\le B/3$, this yields
$\|D\|_j\le\min(B/3,4j)\le2\sqrt{Bj/3}$.
Conditional Hoeffding's inequality, uniformly in $x$, gives
\[
 \Pr(|S|\ge u\mid X=x)\le2e^{-u^2/(2B)}.
\]
Integration implies $\mathbb E e^{S^2/(4B)}\le3$; maximizing
$u^je^{-u^2/(4B)}$ then gives
$\|S\|_j\le\sqrt{6/e}\sqrt{Bj}<\tfrac32\sqrt{Bj}$ for $j\ge2$.
Since $0<\mu-B<1/4$ and $\sqrt{Bj}\ge2$, Minkowski gives
\[
 \|R-\mu\|_j\le(3/2+2/\sqrt3+1/8)\sqrt{Bj}<3\sqrt{Bj}.
\]
For $j=1$, conditional variance gives $\mathbb E|S|\le\sqrt B$.
The drift obeys $\mathbb ED\le\min(B/3,16/3+32/(3B))\le3\sqrt B/2$:
use the first bound for $B\le16$ and the second for $B\ge16$.
Adding $|\mu-B|<1/4$ proves the same conclusion.
\end{proof}

\begin{lemma}[A summable bound for all low-degree magic coefficients]
\label{lem:magic-coeff-majorant}
The absolute constant $C=15$ suffices for
\begin{equation}\label{eq:magic-coeff-majorant}
 |c_j|\le\frac{(C\sqrt n)^j}{\sqrt{j!}},
 \qquad 1\le j\le B.
\end{equation}
\end{lemma}
\begin{proof}
Cancellation of falling factorials gives the exact identity
\begin{equation}\label{eq:magic-exact-normalization}
 m_j\ell_j=\frac1{j!}\prod_{h=0}^{j-1}
 \left(4-\frac2{k+1-h}\right)\le\frac{4^j}{j!}.
\end{equation}
Indeed, their numerator before cancellation is
$(2k+1)^{\underline{2j}}$; pair its adjacent factors.
Since $|\phi_j(k)|\le1$, factorization over the real roots
and Lemmas~\ref{lem:magic-hahn-roots}--\ref{lem:magic-sector-moments} give
\begin{align*}
 |c_j|&\le m_j\ell_j\mathbb E|P_j(R)|\\
 &\le\frac{4^j}{j!}\left(3\sqrt{Bj}+\sqrt{kj}\right)^j
 \le\frac{(15\sqrt n)^j}{\sqrt{j!}}.
\end{align*}
Here $B=k/2=n/4$, $j!\ge(j/e)^j$, and
$4(3/2+1/\sqrt2)\sqrt e<15$.
\end{proof}

\subsubsection{A quantitative Gaussian limit}\label{app:gaussian-limit}
We compare the input fluctuations and the Hahn polynomials directly with
their Gaussian limits. Only the leading generating function is needed.

\begin{proposition}[Gaussian coefficient comparison]\label{prop:magic-fixed-coeffs}
Define $a_j$ by $\sum_{j\ge0}a_jx^j=e^{3x^2/2}$.
There is an absolute $K$ such that, uniformly for
$1\le j\le B^{1/3}$ as $B\to\infty$,
\begin{equation}\label{eq:magic-fixed-coeffs}
 \left|\frac{c_j}{n^{j/2}}-a_j\right|
 \le\frac{K^j}{\sqrt n\sqrt{j!}}.
\end{equation}
In particular, $a_{2m}=(3/2)^m/m!$ and $a_{2m+1}=0$.
\end{proposition}
\begin{proof}
Put $\epsilon=B^{-1/2}$, $V=(R-B-1/4)/\sqrt B$, and
$Q_j(y)=B^{-j/2}P_j(B+1/4+\sqrt B\,y)$.
Let $\operatorname{He}_j$ denote the probabilists' Hermite polynomial and
set $H_j(y)=2^{-j}\operatorname{He}_j(2y)$, with exponential generating
function $e^{ys-s^2/8}$, and let $Z_0$ be a standard normal variable.

\emph{Input moments.}
For the drift $D$ in Lemma~\ref{lem:magic-sector-moments},
$\mathbb ED^r\le6(8/3)^r r!$. Thus $D$ has uniformly bounded
exponential moments in a fixed neighborhood of zero. Its conditional
moment generating function is exactly
\[
 \mathbb E[e^{s(R-B)/\sqrt B}\mid D]
 =\bigl[\cosh(s\epsilon)+D\epsilon^2\sinh(s\epsilon)\bigr]^B.
\]
Factor out $\cosh(s\epsilon)^B$ and expand for complex
$|s|\le B^{1/6}$. The two logarithms are
$s^2/2+O(\epsilon^2|s|^4)$ and
$D\epsilon s+O(D\epsilon^3|s|^3+D^2\epsilon^4|s|^2)$.
Here $D\le B/3$ and $\epsilon|s|\le B^{-1/3}$, so the errors are
integrable under a uniformly bounded exponential tilt of $D$.
Multiplying by $e^{-s\epsilon/4}$ gives
\[
 \left|\mathbb E e^{sV}-e^{s^2/2}\right|
 \le C\epsilon(1+|s|^4)e^{C|s|^2}.
\]
Cauchy's formula on $|s|=\sqrt j$ consequently gives, for $0\le r\le j$,
\[
 |\mathbb EV^r-\mathbb EZ_0^r|
 \le\epsilon K^j r!j^{-r/2}.
\]
By Lemma~\ref{lem:magic-hahn-roots}, all roots of $Q_j$ have modulus at
most $C\sqrt j$. Its coefficient of $y^r$ is therefore bounded by
$\binom jr(C\sqrt j)^{j-r}$. Summing the moment errors yields
\begin{equation}\label{eq:magic-polynomial-moment-error}
 \frac{2^j}{j!}|\mathbb E Q_j(V)-\mathbb E Q_j(Z_0)|
 \le\frac{\epsilon K^j}{\sqrt{j!}}.
\end{equation}
Indeed the sum is bounded by $\epsilon K^j j^{-j/2}$, using
$\sum_{h=0}^j(Cj)^h/h!\le e^{Cj}$, and $j!\le j^j$.

\emph{Polynomial comparison.}
The scaled Jacobi matrix $J_j$ with characteristic polynomial $Q_j$ has
entries, uniformly for $q<j\le B^{1/3}$,
\[
 \frac{\alpha_q-B-1/4}{\sqrt B}=O(B^{-3/2}),\qquad
 \frac{\beta_q}{B}=\frac q4+O\!\left(\frac{q(q+1)}B\right),
\]
by~\eqref{eq:magic-jacobi-coefficients}. The Hermite Jacobi matrix
$J_j^{\rm H}$ has zero diagonal and adjacent entries $\sqrt q/2$.
In operator norm,
$\|J_j-J_j^{\rm H}\|\le Cj^{3/2}/B$ and both norms are at most $C\sqrt j$.
Telescoping the columns of the two characteristic determinants gives
\[
 |Q_j(y)-H_j(y)|
 \le\frac{Cj^{5/2}}B(|y|+C\sqrt j)^{j-1}.
\]
Gaussian moment bounds and $j!\ge(j/e)^j$ imply
\[
 \frac{2^j}{j!}\mathbb E|Q_j(Z_0)-H_j(Z_0)|
 \le\frac{K^j}{B\sqrt{j!}}.
\]
Polynomial factors in $j$ are absorbed into $K^j$.

\emph{Normalization.}
The exact boundary identity~\eqref{eq:phiboundary} and the leading
coefficient normalization give
\[
 \frac{c_j}{n^{j/2}}=\frac{2^j}{j!}\theta_{k,j}\mathbb E Q_j(V),\qquad
 \theta_{k,j}=\left(1-\frac j{k+1}\right)
 \prod_{h=0}^{j-1}\left(1-\frac1{2(k+1-h)}\right)=1+O(j/B).
\]
Finally,
\[
 \sum_{j\ge0}\frac{2^j}{j!}\mathbb E H_j(Z_0)x^j
 =\mathbb E e^{2xZ_0-x^2/2}=e^{3x^2/2}.
\]
The Gaussian moments also bound $|a_j|$ by $K^j/\sqrt{j!}$.
Combining these estimates with~\eqref{eq:magic-polynomial-moment-error}
and $\epsilon=2/\sqrt n$ proves the proposition.
\end{proof}

\subsubsection{Uniform summation of the spectrum}\label{app:uniform-summation}
We first control the high degrees, then sum the quantitative Gaussian limit.

\begin{lemma}[Uniform spectral sums]\label{lem:magic-spectral-sums}
For every fixed $Z>0$, uniformly over integer depths with
$0<z=n(4/9)^t\le Z$,
\begin{equation}\label{eq:magic-spectral-sums}
 1+\sum_{j=1}^k c_j\lambda_j^t=e^{3z/2}(1+O_Z(n^{-1/2})),\qquad
 1+\sum_{j=1}^k |c_j|\lambda_j^t=e^{3z/2}(1+O_Z(n^{-1/2})).
\end{equation}
\end{lemma}
\begin{proof}
Let $t_0=\lceil\log(n/Z)/\log(9/4)\rceil$,
$L=\lfloor B^{1/3}\rfloor$, and $J=\lfloor n/(\log n)^2\rfloor$.
For large $n$, $L\le J\le B$, and all stated depths satisfy $t\ge t_0$.
At $t_0$, Lemma~\ref{lem:quantmatching} and the coefficient majorant give
\[
 |c_j|\lambda_j^{t_0}
 \le\frac{[15\sqrt Z\exp(5t_0J/(4n))]^j}{\sqrt{j!}}
 \qquad(j\le J).
\]
The bracket stays bounded. Its tail from $j>L$ decreases faster than every
power of $n$. For $j>J$, use $|c_j|\le m_j\le\binom{n+1}j$ and
$\lambda_j\le(5/6)^j$ to obtain
\[
 \sum_{j>J}|c_j|\lambda_j^{t_0}
 \le\frac{r_n^{J+1}}{1-r_n},\qquad
 r_n=\frac{e(n+1)}J(5/6)^{t_0}
 =O_Z\!\left((\log n)^2n^{-\log(6/5)/\log(9/4)}\right)\to0.
\]
This tail also decreases faster than every power of $n$.
Both estimates hold for every $t\ge t_0$, since $0\le\lambda_j<1$.

For $j\le L$ put $v_j=(2/3)^j$. Equation~\eqref{eq:matching-eigenvalue-comparison}
gives
\[
 |\lambda_j^t-v_j^t|
 \le\frac{Ct j^2}{n}v_j^t e^{Ct j^2/n}.
\]
To make this estimate uniform at later depths, write $t=t_0+s$.
For large $n$, $(v_je^{Cj^2/n})^s\le e^{-cjs}$, and
$\sup_{s\ge0}(t_0+s)e^{-cjs}=O(1+t_0)$.
Multiplying by the coefficient majorant and summing therefore bounds the
power-replacement error by $O_Z(\log n/n)$: use
$n^{j/2}v_j^{t_0}\le Z^{j/2}$ and the bounded per-degree factor
$e^{Ct_0j/n}$ for $j\le L$.

Proposition~\ref{prop:magic-fixed-coeffs} now gives
\[
 1+\sum_{j=1}^k c_j\lambda_j^t
 =\sum_{j\ge0}a_jz^{j/2}+O_Z(n^{-1/2})
 =e^{3z/2}+O_Z(n^{-1/2}).
\]
The coefficient error sums because
$\sum_j(K\sqrt Z)^j/\sqrt{j!}<\infty$; the omitted Gaussian tail
is smaller than every power of $n$, and $\log n/n=O(n^{-1/2})$.
Since every $a_j\ge0$,
$\bigl||u|-a_j\bigr|\le|u-a_j|$ gives the same comparison for absolute
coefficients. Finally $e^{3z/2}\ge1$, so both additive errors imply the
relative estimates stated above.
\end{proof}

\subsubsection{Exact finite-size certification}\label{app:finite-certificate}
The asymptotic profile locates the transition. Exact rational arithmetic
certifies the bound at each size, including sizes outside the asymptotic regime.

\begin{lemma}[Exact magic-input depth certificate]\label{lem:magiccertificate}
The certificate~\eqref{eq:magiccertificate} exists and is computable exactly
in time polynomial in $n$. At this depth the canonical input has Haar-relative collision at most two
under both $\mathsf G$ and its core, including their finite-alphabet versions.
The certificate does not assert that $t_n^{\rm mag}$ is the least depth at
which the actual collision crosses this factor-two threshold.
\end{lemma}
\begin{proof}
Since $0\le\lambda_j<1$ for $j\ge1$, the absolute spectral sum in
\eqref{eq:magiccertificate} is nonincreasing and tends to zero.
For an explicit finite scan limit, set
\begin{equation}\label{eq:matchingtime}
 T_n=\left\lceil
 \frac{\log((n+1)/\log2)}{\log(6/5)}
 \right\rceil.
\end{equation}
Since $|c_j|\le m_j$, the global envelope~\eqref{eq:matchingdecay} gives
\[
 \sum_{j=1}^k|c_j|\lambda_j^{T_n}
 \le\sum_{j=1}^k\binom{n+1}{j}(5/6)^{jT_n}
 \le\exp\!\bigl((n+1)(5/6)^{T_n}\bigr)-1\le1.
\]
Thus $0\le t_n^{\rm mag}\le T_n$.
The collision guarantee follows from
\eqref{eq:sharpexpansion}, Proposition~\ref{prop:permutationremoval}, and
Lemma~\ref{lem:cliffordtwirl}.

For exact computation, the unnormalized masses $b_{2m}I_B$ in
\eqref{eq:sectorlaw} have common denominator $D_B=4^B(2B+1)!$.
Writing $I_B=W_B/D_B$, the denominators of $b_r$, $\phi_j(r)$, and $c_j$
divide, respectively,
\[
 W_B,\qquad k!(k+1)!,\qquad (k+1)W_Bk!(k+1)!.
\]
These integers have $O(n\log n)$ bits. The eigenvalues have common denominator
$3^k(2k)!$, since each summand's denominator in~\eqref{eq:matchinglambda}
divides $3^q(2k)^{\underline{2q}}$. Raising it to
$t\le T_n=O(\log n)$ still gives polynomial bit length.
Polynomially many exact rational operations therefore evaluate the absolute
sum at each $t=0,\ldots,T_n$; the first successful comparison is the certificate.
\end{proof}

\subsection{The two-particle lower bound}
\label{app:magiclower}

Two-particle correlations retain a positive contribution until the transition
window. We isolate that contribution without assigning signs to the full
collision expansion.

\begin{lemma}[Two-particle lower bound on the collision ratio]\label{lem:magic-early-lower}
Let $n=4B\ge4$, let $\Psi_{\rm in}=|\Psi_4\rangle^{\otimes B}$, and put
\[
 \eta_n=\frac{n(n-2)}{16(n-1)(n-3)},\qquad
 e_n=\frac{7n-8}{16(n^2-1)},\qquad
 a_n=\frac{3n(n-4)}{16(n^2-1)},\qquad
 b_n=\frac{3n^3-7n^2+6n+16}{32(n^2-1)}.
\]
With $\lambda_1,\lambda_2$ as in Lemma~\ref{lem:matchingchain}, every integer
depth $t\ge0$ obeys
\[
 \ratio(\mathsf G_{n,t},\Psi_{\rm in})
 =\ratio(\mathsf C_{n,t},\Psi_{\rm in})
 \ge1+\frac{e_n+a_n\lambda_1^t+b_n\lambda_2^t}{\eta_n}.
\]
The coefficients $a_n,b_n$ are nonnegative and $b_n/\eta_n\sim3n/2$.
Consequently, if $n(4/9)^{t(n)}\to\infty$, then the collision ratio diverges.
For every fixed $q>1$, the least factor-$q$ depth satisfies
\[
 t^\star(q)\ge\frac{\log n}{\log(9/4)}-C_q,
\]
where $C_q$ is a constant depending only on $q$.
\end{lemma}

\begin{proof}
Throughout, $k=n/2$ and $\eta=\eta_n$.
Let $\Gamma_2$ be the two-particle reduced density matrix, regarded as an operator on
$\Lambda^2\mathbb C^n$, with the convention
\[
 (\Gamma_2)_{ij,\ell m}
 =\langle a_\ell^\dagger a_m^\dagger a_j a_i\rangle,
 \qquad i<j,\quad\ell<m.
\]
In particular $(\Gamma_2)_{ij,ij}=\langle n_i n_j\rangle$, and
$\Gamma_2$ transforms by conjugation with $\Lambda^2 U$ under a passive
transformation $U$.
The canonical magic input has one-particle reduced density matrix $I_n/2$;
this identity is preserved by every passive circuit. Put
\[
 q=\frac{k(k-1)}{n(n-1)}=\frac{n-2}{4(n-1)},\qquad
 T=\Gamma_2-qI_{\Lambda^2\mathbb C^n}.
\]
The one-index contraction of $T$ vanishes, since the contraction of $\Gamma_2$
is $(k-1)I_n/2$ and that of $qI$ is $q(n-1)I_n$.
Define
\[
 S_0(T)=\sum_{i<j}|T_{ij,ij}|^2,
\]
and let $S_1(T)$ and $S_2(T)$ be the sums of squared matrix entries whose row
and column index pairs intersect in exactly one and zero indices, respectively.
Thus $S_0+S_1+S_2=\|T\|_{\rm HS}^2$.

\paragraph{A pointwise collision lower bound.}
For real coefficients $a_{ij}=a_{ji}$, $a_{ii}=0$, satisfying
$\sum_{j\ne i}a_{ij}=0$ for every $i$, set
$f_a(y)=\sum_{i<j}a_{ij}y_i y_j$. Under the uniform measure $u$ on the
weight-$k$ slice, $\mathbb E_u f_a=0$ and
\[
 \mathbb E_u f_a^2=\eta\sum_{i<j}a_{ij}^2,\qquad
 \eta=\mu_2-2\mu_3+\mu_4
 =\frac{n(n-2)}{16(n-1)(n-3)},\qquad
 \mu_\ell=\frac{k^{\underline\ell}}{n^{\underline\ell}}.
\]
Indeed, the sum of $a_ea_f$ over unordered distinct pairs of edges sharing a
vertex is $-\sum_ea_e^2$, while the corresponding sum over disjoint edge pairs
is $\tfrac12\sum_ea_e^2$. These identities follow by squaring the row-sum
constraints and the resulting identity $\sum_ea_e=0$.
For the output law $p$ of any fixed passive circuit, put
$a_{ij}=T_{ij,ij}=\mathbb E_p[y_i y_j]-q$.
These coefficients have zero row sums. Writing $h(y)=Np(y)-1$ gives
\[
 \langle h,f_a\rangle_{L^2(u)}=S_0(T),\qquad
 \|h\|_{L^2(u)}^2=N\sum_y p(y)^2-1.
\]
Cauchy--Schwarz, with the zero case understood separately, therefore yields
\begin{equation}\label{eq:magic-2rdm-parseval}
 N\sum_y p(y)^2\ge1+\frac{S_0(T)}\eta.
\end{equation}

\paragraph{The exact three-sector matching chain.}
The contraction-free operator space containing $T$ is the irreducible $U(n)$
module of highest weight $(1,1,0^{n-4},-1,-1)$. This follows from the same
skew-Howe decomposition used in Appendix~\ref{app:proofs}: under the
particle-hole identification displayed below, the three summands of
$\operatorname{End}(\Lambda^2\mathbb C^n)$ are the scalar, adjoint, and
contraction-free modules. The contraction maps the scalar and adjoint
summands isomorphically onto $\operatorname{End}(\mathbb C^n)$ and has the
third summand as its kernel. Explicitly, if $L(A)$ is the induced one-body
operator on $\Lambda^2\mathbb C^n$, then its contraction is
$(n-2)A+(\operatorname{Tr}A)I$, which verifies the first two assertions.
Its restriction to the monomial
group has exactly three inequivalent irreducible sectors, with dimensions
\[
 d_0'=\frac{n(n-3)}2,\qquad
 d_1'=n(n-1)(n-3),\qquad
 d_2'=\frac{n(n-1)(n-2)(n-3)}4.
\]
Here the zero torus weight consists of the pair-indexed diagonal arrays with
zero row sums, namely the Specht module $S^{(n-2,2)}$.
A weight $e_i-e_j$ has $n-2$ coordinates indexed by the remaining shared
index and one zero-sum constraint; its stabilizer acts by the irreducible
standard module of dimension $n-3$ (up to an immaterial sign twist).
A weight $e_i+e_j-e_\ell-e_m$, with four distinct indices, has a
one-dimensional fiber. Little-group induction gives the dimensions displayed
above and proves multiplicity-freeness.

Regard $T$ as a vector $|T\rangle$ in operator space with the
Hilbert--Schmidt inner product. Consequently, the averaged matching twirl of the rank-one operator
$|T\rangle\langle T|$ induces a reversible three-state chain on these squared
norms. A virtual monomial twirl can be inserted because each $S_r$ is a
monomial-invariant quadratic observable and the layer law is monomial covariant.
The chain therefore applies to arbitrary initial $T$ in this module; no physical
input permutation is required. Its stationary probabilities are
\[
 \widehat\pi_0=\frac2{n(n+1)},\qquad
 \widehat\pi_1=\frac{4(n-1)}{n(n+1)},\qquad
 \widehat\pi_2=\frac{(n-1)(n-2)}{n(n+1)}.
\]
In the column-input convention the transition matrix is
\begin{equation}\label{eq:magic-2rdm-chain}
 \widehat A=\frac1{9(n-1)(n-3)}
 \begin{pmatrix}
 (4n-11)(n-1)&2n-7&2\\
 (4n-14)(n-1)&6n^2-22n+14&4(3n-10)\\
 (n-2)(n-1)&(n-2)(3n-10)&9n^2-48n+65
 \end{pmatrix}.
\end{equation}

Here is a local derivation of the matrix, independent of an eigenvalue fit.
The particle-hole identification
\[
 \Lambda^2\mathbb C^n\otimes(\Lambda^2\mathbb C^n)^*
 \cong (\det)^{-1}\otimes
       \Lambda^2\mathbb C^n\otimes\Lambda^{n-2}\mathbb C^n
\]
identifies the contraction-free module with the copy-spin
$S=n/2-2$ summand. The fixed copy occupations select one weight line in
this spin representation. A torus weight with $r$ double and $r$ empty
modes has $n-2r$ copy doublets, and the multiplicity of spin $S$ in those
doublets is
\[
 f_0=\binom n2-n=\frac{n(n-3)}2,\qquad f_1=n-3,\qquad f_2=1;
\]
there are no such weights with $r>2$.
The local mode-$U(2)$ twirl acts exactly as in the collision chain: the two
$DE$ orientations and the $SS$ copy singlet form the three weight lines of
$\operatorname{Sym}^2\mathbb C^2$, which the twirl makes equiprobable;
the $SS$ copy triplet cannot leave $SS$. The sectors with a different local
particle number preserve the number of double modes on that edge.

At $r=0$ all $n/2$ matching edges are $SS$. If $J$ counts local singlet
edges in the normalized global spin-$S$ multiplicity state, then fixing $q$
such edges leaves the spin-$S$ multiplicity space on $n-2q$ doublets. Hence
\[
 \mathbb E J=\frac n2\frac{f_1}{f_0}=1,\qquad
 \mathbb E[J(J-1)]
 =\frac n2\left(\frac n2-1\right)\frac{f_2}{f_0}
 =\frac{n-2}{2(n-3)}.
\]
This trace identity may equivalently be evaluated on the entire spin-$S$
isotypic space: both numerator and denominator then acquire the same factor
$2S+1$. The local singlet projectors commute with the global spin action,
so the result also holds on the selected copy-weight line.
Conditionally on these commuting local projectors, each singlet independently
produces one double mode with probability $2/3$. Thus
\[
 \mathbb E[R_1\mid R_0=0]=\frac23,\qquad
 \mathbb E[R_1(R_1-1)\mid R_0=0]
 =\frac{2(n-2)}{9(n-3)}.
\]
Since $R_1\le2$, these determine the first column of
\eqref{eq:magic-2rdm-chain}.

At $r=2$, the remaining $n-4=2S$ doublets have maximal total spin and are
fully symmetric, so every $SS$ edge is a triplet. Only an initially $DE$
edge can reduce the sector label, and it does so with probability $1/3$.
If $J$ is now the number of $DE$ edges in the uniform matching, elementary
matching counts give
\[
 \mathbb E J=\frac4{n-1},\qquad
 \mathbb E[J(J-1)]=\frac4{(n-1)(n-3)}.
\]
Thus $2-R_1$, conditionally on $J$, is $\operatorname{Bin}(J,1/3)$, giving
the third column. Detailed balance with $d_r'$ then determines the two
off-diagonal entries of the middle column, and its column sum determines
the remaining entry. This proves the matrix explicitly.

Its characteristic polynomial factors as
\[
 \det(zI-\widehat A)=(z-1)(z-\lambda_1)(z-\lambda_2),
\]
where
\[
 \lambda_1=\frac{2(n-2)}{3(n-1)},\qquad
 \lambda_2=\frac{4n^2-19n+27}{9(n-1)(n-3)}.
\]
These are precisely the first two nontrivial eigenvalues of the full
collision chain. Both lie strictly between zero and one for $n\ge4$,
and they are distinct for these integer sizes.

\paragraph{Initial norms and the lower bound at every depth.}
For a canonical magic block, the two supported pair occupations have
probability $1/2$ and the other four pair occupations have probability zero.
Pairs in different blocks have probability $1/4$. The only off-diagonal
two-particle entries are the two complementary-pair entries per block,
each of magnitude $1/2$. Therefore
\[
 S_0(0)=\frac{n(3n-4)}{32(n-1)},\qquad
 S_1(0)=0,\qquad S_2(0)=\frac n8,
\]
and
\[
 H:=\|T\|_{\rm HS}^2=\frac{n(7n-8)}{32(n-1)}.
\]
Writing $S_0(t)$ for the averaged value after $t$ independent matching layers,
\eqref{eq:magic-2rdm-chain} and its characteristic polynomial give
\begin{equation}\label{eq:magic-2rdm-evolution}
 S_0(t)=e_n+a_n\lambda_1^t+b_n\lambda_2^t,
\end{equation}
with $e_n=H\widehat\pi_0$ and the coefficients stated in the lemma.
The remaining two coefficients follow by using the stated $S_0(0)$ and
\[
 S_0(1)=\frac{n(12n^2-49n+52)}{288(n-1)(n-3)}
\]
to solve the two linear equations after subtracting $e_n$.
Here $a_n\ge0$ and $b_n>0$ for all $n\ge4$, and $b_n\sim3n/32$.
Averaging \eqref{eq:magic-2rdm-parseval} therefore proves the all-depth estimate
\begin{equation}\label{eq:magic-early-depth-lower}
 \ratio_{n,t}\ge1+\frac{b_n}{\eta}\lambda_2^t,
 \qquad t\ge0.
\end{equation}
In particular $b_n/\eta=(3/2+o(1))n$ and
$\lambda_2=(4/9)(1+O(n^{-1}))$.
For any nonnegative depth sequence satisfying $n(4/9)^{t(n)}\to\infty$,
one has $t(n)\le\log n/\log(9/4)$ for all sufficiently large $n$. Thus
$\lambda_2^{t(n)}=(4/9)^{t(n)}(1+o(1))$ uniformly for these depths,
and \eqref{eq:magic-early-depth-lower} proves the asserted divergence.

Fix $q>1$ and choose $C=C(q)$ large enough. At
$t_0=\lfloor\log n/\log(9/4)-C\rfloor$, the lower bound is at least
$1+(3/2+o(1))(9/4)^C>q\ratio_{\rm Haar}(n)$ for large $n$.
Because this bound is nonincreasing in $t$, it excludes every earlier crossing.
No monotonicity of the collision itself is needed.
\end{proof}

\clearpage
\section{Worst-case hardness at native depth four}\label{hard:sec:worst}

We work with the paired input and ordered two-mode gates of
Section~\ref{sec:model}. The proof has three stages: express each input
block as a logical Bell pair, fuse these pairs into a graph state, and
postselect a measurement pattern for a universal computation.

\subsection{The slot picture}
Fix $w$ \emph{slots}. Slot $q$ carries an ordered pair of distinct modes $(r_q(0),r_q(1))$, its two rails, and all $2w$ rails are distinct.
For $x\in\{0,1\}^w$ define the logical basis state
\begin{equation}
\ket{x}_L:=a^\dagger_{r_1(x_1)}\,a^\dagger_{r_2(x_2)}\cdots a^\dagger_{r_w(x_w)}\ket{\rm vac},
\end{equation}
with creation operators ordered by \emph{slot index}, not by mode label.
Each $\ket{x}_L$ equals $\pm$ an occupation-basis state, with a sign fixed by $x$.
Probabilities of occupation strings can therefore be read off in either basis.

\begin{lemma}[Slot calculus]\label{hard:lem:slot}
\begin{enumerate}[label=(\roman*),itemsep=1pt]
\item For $u\in U(2)$, the gate $F_{r_q(0)r_q(1)}(u)$ on the two rails of slot $q$ acts on $\{\ket{x}_L\}$ as $u$ on qubit $q$, with no signs.
\item Give block $b$ the slots $X_b=(4b+1,4b+3)$ and $Y_b=(4b+2,4b+4)$, listed block by block. Then $\ket{\Psi_{\rm in}}=\bigotimes_b(\ket{00}_L+\ket{11}_L)/\sqrt2$ on the pairs $(X_b,Y_b)$, a product of logical Bell pairs.
\end{enumerate}
\end{lemma}
\begin{proof}
A passive transformation satisfies $\Gamma(U)a_i^\dagger\Gamma(U)^\dagger=\sum_jU_{ji}a_j^\dagger$ and $\Gamma(U)\ket{\rm vac}=\ket{\rm vac}$.
Hence it replaces each creation operator in the product \emph{in place} by a linear combination.
In (i), the operator in slot $q$ becomes a combination of that slot's two rails, with coefficients given by $u$.
For (ii), $a^\dagger_{4b+1}a^\dagger_{4b+2}=a^\dagger_{r_{X}(0)}a^\dagger_{r_{Y}(0)}$ and $a^\dagger_{4b+3}a^\dagger_{4b+4}=a^\dagger_{r_{X}(1)}a^\dagger_{r_{Y}(1)}$.
Different blocks contribute even monomials, which commute.
\end{proof}

\subsection{A fermionic type-I fusion}
\begin{lemma}[Fusion]\label{hard:lem:fusion}
Let $\ket{\phi}=\sum_x\phi(x)\ket{x}_L$ have every slot singly occupied. Take slots $A$ and $B$ with slot indices $\alpha<\beta$, and set $m=r_A(0)$, $m'=r_B(1)$.
Apply $F_{mm'}(H)$ and postselect the occupations $(n_m,n_{m'})=(1,0)$, implemented as $(1-n_{m'})\,a_m$.
The unnormalized result is $\ket{\phi'}=(-1)^\beta\,Z_V K_0\ket\phi$, where
\[
K_0=\tfrac{1}{\sqrt2}\big(\ket0_V\bra{00}_{AB}+\ket1_V\bra{11}_{AB}\big).
\]
The new slot $V$ sits at slot index $\alpha$ with rails $(r_B(0),r_A(1))$, and slot $B$ is removed.
In particular the sign depends only on $(\alpha,\beta)$, never on the other qubits.
\end{lemma}
\begin{proof}
It suffices to check the four logical basis states. The postselection keeps exactly one particle in the mixed modes $(m,m')$, so it kills $01$ (both occupied) and $10$ (both empty). The two surviving cases are
\[
\begin{array}{c|c|c}
(x_A,x_B)&\text{remaining rail}&\text{amplitude}\\\hline
(0,0)&r_B(0)&(-1)^{\alpha-1}(-1)^{\beta-\alpha-1}H_{11}=(-1)^\beta/\sqrt2\\
(1,1)&r_A(1)&(-1)^{\beta-1}H_{12}=-(-1)^\beta/\sqrt2.
\end{array}
\]
In the first row, annihilation passes $\alpha-1$ occupied slots, and moving the remaining rail to slot $\alpha$ passes another $\beta-\alpha-1$. In the second row, annihilation passes $\beta-1$ slots and the remaining rail is already in place. These signs depend only on the slot indices. The two amplitudes give $(-1)^\beta Z_VK_0$.
\end{proof}
For graph states $\ket{G}=2^{-|V|/2}\sum_x(-1)^{\sum_{uv\in E}x_ux_v}\ket x$, applying $K_0$ to non-adjacent vertices $a,b$ gives $\tfrac12\ket{G'}$.
Here $G'$ merges $a,b$ into one vertex with neighborhood $N(a)\triangle N(b)$; this is the type-I fusion of~\cite{BrowneRudolph2005}.
Since $K_0(Z\otimes I)=ZK_0$, byproducts $Z$ on fused vertices propagate to the merged vertex.

\subsection{Postselected measurement-based computation}
Write $\ket{+_\phi}=(\ket0+e^{i\phi}\ket1)/\sqrt2$ and $R_\phi=H\,{\rm diag}(1,e^{-i\phi})$, so that $\bra0R_\phi=\bra{+_\phi}$.
\begin{fact}[\cite{Raussendorf2003MBQC,Danos2007Calculus,Broadbent2009Blind}]\label{hard:fact:mbqc}
Let $Q$ be a Clifford$+T$ circuit on $w$ qubits with $s$ gates. One can compute in time $\poly(w,s)$ the following data.
\begin{itemize}[leftmargin=*,itemsep=1pt]
\item A subgraph $G=(V,E)$ of the square lattice, so every degree is at most $4$, with $w$ designated output vertices. There are no isolated vertices: a wire of length one is padded with two $\phi=0$ measurements, which implement $H^2=I$.
\item For every other vertex, either an angle $\phi_v\in\frac\pi4\mathbb{Z}$ or a $Z$-deletion.
\end{itemize}
For a known local Clifford output frame $K$ and an irrelevant global phase $e^{i\chi}$, these satisfy
\[\textstyle
\big(\bigotimes_{v}\bra{m_v}\big)\ket{G}=e^{i\chi}2^{-(|V|-w)/2}\,KQ\ket{+^{w}}.
\]
Here the product runs over non-output vertices, with $\bra{m_v}=\bra{+_{\phi_v}}$ or $\bra{0}$.
The brickwork patterns of~\cite{Broadbent2009Blind} have this form and have flow.
The normalization follows from the one-bit teleportation identity
\[
 (\bra{+_\phi}\otimes I)\,CZ(\ket\psi\otimes\ket+)
 =2^{-1/2}H\operatorname{diag}(1,e^{-i\phi})\ket\psi
 =2^{-1/2}R_\phi\ket\psi,
\]
valid for every input qubit $\ket\psi$, with $CZ$ the controlled-$Z$ gate.
We use the all-zero measurement branch: each measured vertex contributes
$2^{-1/2}$, and the outcome-dependent byproduct corrections are trivial.
A $Z$-deletion with outcome zero also contributes $2^{-1/2}$ and deletes
its vertex without a byproduct. No generalized-flow result is needed here.
Here $K=\bigotimes_vK_v$ is a tensor product of single-qubit Cliffords on the outputs; an example is Hadamards fixed by wire-length parity. This known frame is absorbed into the final postselection, which becomes $\bra0K_v^{-1}$ on output $v$. Up to phase this is one of $\bra0$, $\bra1$ or $\bra{+_\phi}$ with $\phi\in\frac\pi2\mathbb Z$.
\end{fact}
\begin{fact}[\cite{Dawson2005Polynomials,Aaronson2005Postselection}]\label{hard:fact:sharpP}
Computing $|\bra{0^w}Q\ket{0^w}|^2$ exactly for Clifford$+T$ circuits $Q$ is $\#\mathsf P$-hard.
Explicitly, let $f$ be a Boolean formula on $m$ variables with $s$ satisfying assignments, and put $f'(x,b)=b\wedge f(x)$. Compute $f(x)$ into an ancilla with Toffoli gates, apply $CZ$ between $b$ and that ancilla, and uncompute; this realizes the phase oracle $(-1)^{f'}$. Sandwiching it between Hadamards on all $m+1$ input qubits gives a Clifford$+T$ circuit $Q_f$ with $\bra0Q_f\ket0=2^{-m-1}\sum_{x,b}(-1)^{b f(x)}=1-s/2^m\ge0$. So $|\bra0Q_f\ket0|^2$ determines $s$. More precisely, the rational values $(1-j/2^m)^2$, $0\le j\le2^m$, are strictly decreasing. Binary search using exact comparisons recovers $s$ in $O(m)$ comparisons, without enumerating all counts or assuming a square-root primitive.
\end{fact}

\begin{theorem}[Worst case at native depth four]\label{hard:thm:worst}
For every Clifford$+T$ circuit $Q_0$ one can compute in polynomial time the following objects.
\begin{itemize}[leftmargin=*,itemsep=1pt]
\item A number of blocks $B$ and a passive circuit $C_{Q_0}$ on $n=4B$ modes. It consists of four \emph{native layers} $P_1,\dots,P_4$, each a set of gates on disjoint mode pairs, with all gates in $\{R_\phi:\phi\in\frac\pi4\mathbb Z\}$ (note $H=R_0$). Modes outside the pairs of a layer are idle.
\item A string $y_{Q_0}\in\Y$.
\item An integer $\kappa\ge0$.
\end{itemize}
These satisfy
\[
p_{y_{Q_0}}(C_{Q_0})=\big|\bra{y_{Q_0}}\Gamma(C_{Q_0})\ket{\Psi_{\rm in}}\big|^2=2^{-\kappa}\,\big|\bra{0^w}Q_0\ket{0^w}\big|^2 .
\]
Each layer has at most $2B$ pairs.
Hence computing $p_y$ exactly for depth-four passive circuits on the magic input is $\#\mathsf P$-hard.
\end{theorem}
\begin{proof}
The construction prepares graph edges, fuses copies of each vertex, and then measures the resulting graph state.
Apply Fact~\ref{hard:fact:mbqc} to $Q=Q_0H^{\otimes w}$. Use one magic block per edge of its graph $G=(V,E)$, so $B=|E|$. The block's two logical slots represent the endpoints of that edge.
\begin{enumerate}[leftmargin=*,itemsep=3pt]
\item[\emph{Layer 1.}] Apply $H$ to the second slot of each block. This turns its Bell pair into a two-vertex graph state, producing the disjoint union of all graph edges.
\item[\emph{Layers 2--3.}] Fuse the $d_v\le4$ copies of each vertex $v$ along a binary tree of depth at most two. All fusions in a layer use disjoint modes. Because $G$ is simple, the copies have disjoint neighborhoods and each fusion contributes a factor $1/2$. After
\[
F=\sum_v(d_v-1)=2|E|-|V|
\]
fusions, the state is $2^{-F}(\bigotimes_v Z_v^{b_v})\ket G$, up to a global sign, with known bits $b_v$.
\item[\emph{Layer 4.}] Let $\bra{m_v}$ be the desired measurement bra from Fact~\ref{hard:fact:mbqc}; on an output take $\bra{m_v}=\bra0K_v^{-1}$. Implement $\bra{m_v}Z^{b_v}$ to cancel the fusion byproduct. A computational-basis bra needs no gate: select the corresponding occupied rail. For $\bra{m_v}=\bra{+_{\phi_v}}$, apply $R_{\phi_v+\pi b_v}$ and select rail 0. Thus every gate belongs to the stated finite set.
\end{enumerate}
Define $y_{Q_0}$ by these final rail choices and by $(1,0)$ on each fusion's measured modes. Its weight is $F+|V|=2|E|=n/2$.
Fusion modes are never used again, so their projections commute with all later gates. Moreover $\bra y=\pm\bra{y\setminus m}a_m$ when $y_m=1$, with a sign fixed by $y$. Hence the final occupation measurement implements the annihilation-form postselections of Lemma~\ref{hard:lem:fusion}, up to an irrelevant overall sign.
The resulting amplitude, up to phase, is
\[
2^{-F}2^{-(|V|-w)/2}\bra{0^w}Q\ket{+^w}
=2^{-F-(|V|-w)/2}\bra{0^w}Q_0\ket{0^w}.
\]
Squaring gives $\kappa=2F+|V|-w=4|E|-|V|-w$. The graph and all gates are computable in polynomial time, and Fact~\ref{hard:fact:sharpP} gives the hardness claim.
\end{proof}
\begin{remark}[Padding]\label{hard:rem:pad}
If a larger $n=4B'$ is prescribed, add $B'-B$ untouched blocks and give each the pattern $1100$ in $y_{Q_0}$. Each contributes $|\langle1100|\Psi_4\rangle|^2=1/2$, so $\kappa$ increases by $B'-B$.
In particular, the padded exponent is $\kappa=3|E|+n/4-|V|-w\le n$, since $|E|\le n/4$.
\end{remark}

\subsection{Relabelling by switches}
Worst cases are placed inside a schedule by filling it with switches. The next lemma records that this only relabels modes.
\begin{lemma}[Relabelling]\label{hard:lem:relabel}
Let $C$ be a passive circuit whose gates are switches, except in four layers $\ell_1<\dots<\ell_4$.
Let $\pi_\ell$ be the permutation of positions implemented by the switches before layer $\ell$, and let $\pi$ be the permutation implemented by all switches.
Suppose layer $\ell_r$ applies $F_{\pi_{\ell_r}(i)\,\pi_{\ell_r}(j)}(u)$ for every gate $F_{ij}(u)$ of $P_r$, and $I$ elsewhere.
Then $p_z(C)=p_{\pi^{-1}(z)}(C_{Q_0})$ for every $z$. In particular $p_z(C)=p_{y_{Q_0}}(C_{Q_0})$ whenever $\pi$ maps the ones of $y_{Q_0}$ onto the ones of $z$.
\end{lemma}
\begin{proof}
Work with single-particle matrices, on which $\Gamma$ is a homomorphism.
A permutation matrix $P_\sigma$ satisfies $P_\sigma F_{ij}(u)P_\sigma^{-1}=F_{\sigma(i)\sigma(j)}(u)$. So each native layer of $C$ equals $P_{\pi_{\ell_r}}P_rP_{\pi_{\ell_r}}^{-1}$, where $P_r$ also denotes the layer's matrix, and the switches between native layers multiply to $P_{\pi_{\ell_{r+1}}}P_{\pi_{\ell_r}}^{-1}$.
The product telescopes to $P_\pi\,C_{Q_0}$.
Finally, $\Gamma(P_\pi)$ maps each occupation state $\ket{x}$ to $\pm\ket{\pi(x)}$.
\end{proof}

\section{Exact and robust worst-to-average reductions}\label{hard:sec:reduction}

A \emph{gate-random ensemble} is $\nu=\mu_S\otimes\Haar(U(2))^{\otimes g}$. Here $\mu_S$ is a law on schedules of a prescribed depth $D\ge1$ that admits a polynomial-time exact sampler in the adopted real-RAM model; the gates are i.i.d.\ Haar, and $g=nD/2$.
Both $\Gp$ (all $M_\ell$ i.i.d.\ uniform) and $\Gsh$ (Section~\ref{hard:sec:hybrid}) satisfy this sampling requirement. A uniform permutation, paired consecutively, gives a uniform perfect matching; the hybrid prefix is fixed. Likewise, a uniform output in $\Y$ is sampled by taking the first $n/2$ elements of a uniform permutation.

\paragraph{How the reduction works.}
First draw a schedule and an output from the target law, then program a hard instance into that schedule with the embedding procedure $\mathcal W$.
The interpolation below changes only the gates, so the schedule and output retain their correct joint distribution. Queries are made near the Haar endpoint, where the oracle's average-case guarantee transfers by total variation, and the recovered polynomial is evaluated at the hard endpoint.
The task is therefore to control three separate quantities: the probability that embedding fails, the distributional error of each query, and the degree needed for decoding.

\subsection{Cayley paths and polynomial degree}
Let $\cay(A)=(I-iA)(I+iA)^{-1}$ for Hermitian $A$; its inverse is $\cay^{-1}(U)=-i(U+I)^{-1}(I-U)$, defined when $-1\notin{\rm spec}(U)$.
Given a worst-case gate $w_e$ and a Haar sample $h_e$, set $A_e=\cay^{-1}(w_e^{-1}h_e)$, which exists almost surely.
Define
\begin{equation}
u_e(\theta)=w_e\,\cay\big((1-\theta)A_e\big),\qquad \theta\in[0,1].
\end{equation}
Then $u_e(0)=h_e$, which is Haar distributed, and $u_e(1)=w_e$. The path is unitary throughout.

\begin{lemma}[Total variation]\label{hard:lem:tv}
For $h\sim\Haar(U(2))$ and $a=1-\theta\in(0,1]$, the law of $\cay(a\,\cay^{-1}(h))$ has density with respect to Haar in $[a^4,a^{-4}]$.
Consequently $\TV\le1-a^4\le4\theta$ for every $\theta\in[0,1]$.
By left invariance the same bound holds for $u_e(\theta)$, and for $g$ independent gates the distance is at most $4g\theta$.
\end{lemma}
\begin{proof}
The map preserves eigenvectors and sends an eigenphase $\varphi$ to $\psi=2\arctan(a\tan(\varphi/2))$. Put $s_i=\tan(\psi_i/2)$. The inverse phase Jacobian is
\[
J_i=\frac1a\frac{1+s_i^2}{1+s_i^2/a^2}\in[a,a^{-1}].
\]
Haar's eigenphase density is proportional to $|e^{i\varphi_1}-e^{i\varphi_2}|^2$. Using
$|e^{i\psi_1}-e^{i\psi_2}|^2=4(s_1-s_2)^2/[(1+s_1^2)(1+s_2^2)]$,
the ratio of the preimage and image densities is $J_1J_2$. Multiplying by the change-of-variables Jacobian gives $\rho=(J_1J_2)^2\in[a^4,a^{-4}]$. Hence
\[
\TV=\int(1-\rho)_+\,d\Haar\le1-a^4\le4\theta.
\]
Left multiplication preserves Haar measure, and total variation of product laws is at most the sum of the factor distances. At $\theta=1$ the bound follows from $\TV\le1$.
\end{proof}

\begin{lemma}[Degree]\label{hard:lem:degree}
Fix $S$, $z$, the worst-case gates $w$ and the samples $A_e$. Put $q_e(\theta)=\det(I+i(1-\theta)A_e)$ and $D(\theta)=\prod_e|q_e(\theta)|^2$.
Then $R(\theta):=p_z(S,u(\theta))\,D(\theta)$ agrees on $\mathbb R$ with a polynomial of degree at most $4g$.
Moreover $D(\theta)\ge1$ on $[0,1]$, $D(1)=1$, and $R(1)=p_z(S,w)$.
\end{lemma}
\begin{proof}
Use $\cay(X)=(I-iX)\,{\rm adj}(I+iX)/\det(I+iX)$ and $\det(I-isA)=\overline{\det(I+isA)}$ for real $s$.
Then $q_e(\theta)\,F(u_e(\theta))=q_e\oplus w_e(I-i(1-\theta)A_e){\rm adj}(I+i(1-\theta)A_e)\oplus\det(w_e)\,\bar q_e$.
Here $\bar q_e$ denotes the polynomial with conjugated coefficients, and every entry has degree $\le2$ in $\theta$.
Hence $\Gamma(C(\theta))=\prod_eF_e(u_e(\theta))$, multiplied by $\prod_eq_e(\theta)$, has polynomial matrix elements of degree $\le2g$.
So $\bra z\Gamma(C(\theta))\ket{\Psi_{\rm in}}=P(\theta)/\prod_eq_e(\theta)$ with $\deg P\le2g$.
For real $\theta$ this gives $p_z=(P\bar P)(\theta)/D(\theta)$ with $\deg(P\bar P)\le4g$.
Also $|q_e(\theta)|^2=\det(I+(1-\theta)^2A_e^2)\ge1$, with equality at $\theta=1$.
\end{proof}

\subsection{Exact decoding and reduction}

\begin{fact}[Berlekamp--Welch with a total decoder~\cite{WelchBerlekamp1986}]\label{hard:fact:bw}
Let $L>d\ge0$ be integers, let $\theta_1,\ldots,\theta_L$ be distinct real numbers, and set
\[
b=\left\lfloor\frac{L-d-1}{2}\right\rfloor.
\]
There is a deterministic decoder that uses $\poly(L)$ arithmetic operations on every real transcript $(y_1,\ldots,y_L)$.
It returns either a polynomial of degree at most $d$ agreeing with at least $L-b$ replies, or the failure symbol $\bot$.
If the replies disagree with a degree-at-most-$d$ polynomial $R$ at at most $b$ nodes, the decoder returns that $R$, uniquely.
\end{fact}
\begin{proof}
Introduce a monic error-locator polynomial $E$ of degree $b$ and a polynomial $Q$ of degree at most $d+b$.
Solve the linear system
\[
Q(\theta_i)=y_iE(\theta_i),\qquad 1\le i\le L,
\]
for their coefficients; the leading coefficient of $E$ is fixed to one.
If no solution exists, return $\bot$. Otherwise take a solution, divide $Q$ by $E$, and verify that the remainder is zero, that the quotient has degree at most $d$, and that it agrees with at least $L-b$ replies. Return the quotient if all checks pass, and $\bot$ otherwise.
Gaussian elimination, polynomial division and these checks have polynomial arithmetic cost on every transcript.

For correctness, suppose the bad-node set is $B$, with $f=|B|\le b$.
Then
\[
E_0(\theta)=\theta^{b-f}\prod_{i\in B}(\theta-\theta_i),\qquad Q_0=E_0R
\]
is a feasible monic pair. For any solution $(E,Q)$, the polynomial $Q-ER$ has degree at most $d+b$ and vanishes at all $L-f\ge L-b>d+b$ good nodes. Hence $Q=ER$, and the verification returns $R$.
Finally, two degree-at-most-$d$ polynomials each agreeing with at least $L-b$ replies agree with each other at at least $L-2b>d$ distinct nodes, so they coincide.
Outside the correction radius no correctness claim is made, but the decoder still terminates in polynomial time.
\end{proof}

The exact and robust reductions use the same query grid. For $g\ge1$
and an integer margin parameter $p\ge4$, set
\begin{equation}\label{hard:eq:decodingparameters}
d=4g,\quad L=(4p+1)d,\quad \Delta=(32gp)^{-1},\quad
\theta_i=i\Delta/L,\quad b=2pd-1.
\end{equation}
Here $1\le i\le L$, and $b=\lfloor(L-d-1)/2\rfloor$ is the
exact decoder's correction radius.

\begin{theorem}[Worst-to-average link]\label{hard:thm:reduction}
Let $\nu$ be a gate-random ensemble with $g\ge1$ gates, let $Q_0$ be a Clifford$+T$ circuit, let $p$ be a positive integer, and let $0<\xi<1$.
Suppose a polynomial-time procedure $\mathcal W$, on input $(S,z)$, either fails or outputs gates $w\in U(2)^g$ and a nonnegative integer $\kappa$ with $p_z(S,w)=2^{-\kappa}|\bra0Q_0\ket0|^2$.
The integer $\kappa$ has a polynomial-length binary description; it is at most $n$ for the constructions below.
Assume that $\mathcal W$ fails with probability at most $\eta$ over $S\sim\mu_S$, uniform $z$, and its private coins if randomized.
If an oracle $\mathcal O$ solves $\textup{\textsc{Exact}}(\nu,\delta)$ with $\delta,\eta\ge0$ and
\[
\delta+\eta\le\frac14-\frac1p,
\]
then a randomized real-RAM algorithm outputs $|\bra0Q_0\ket0|^2$ with probability at least $1-\xi$.
It makes $O(gp^3\log(2/\xi))$ oracle queries and $\poly(g,p,\log(2/\xi))$ further arithmetic operations, in addition to schedule sampling, the runs of $\mathcal W$, and $O(\log(2+\kappa))$ operations per successful embedding to form $2^\kappa$ by repeated squaring.
If $\mathcal O$ is randomized, it is a fixed, stateless randomized procedure, and its success probability includes fresh internal coins for each query.
\end{theorem}
\begin{proof}
\emph{Algorithm.} The promise implies $p\ge4$, so use the parameters
in~\eqref{hard:eq:decodingparameters}. One run has three steps.
\begin{enumerate}[itemsep=2pt]
\item Sample $S,z$ from the target law and run $\mathcal W$. On failure, mark the run and use dummy gates $w_e=I$.
\item Draw independent Haar masks and form the Cayley paths. Query all $L$ points, even on marked runs, and multiply each reply by $D(\theta_i)$. A singular Cayley inverse returns $\bot$; this event has probability zero.
\item On a marked run return $\bot$. Otherwise apply Fact~\ref{hard:fact:bw}; return $2^\kappa R(1)$ if decoding succeeds, and $\bot$ otherwise.
\end{enumerate}

\emph{Why a run succeeds.} Conditional on $S,z$ and the embedding coins, the programmed gates are fixed before the Haar masks are drawn. Lemma~\ref{hard:lem:tv} therefore puts each query within $4g\Delta=1/(8p)$ of the target law. Averaging over $S,z$ and the embedding coins preserves this bound; fresh oracle coins can be included in both laws. Dummy queries ensure that this argument never conditions on embedding success.

Let $F$ count incorrect replies. The replies may be correlated, but
$\mathbb EF\le L(\delta+1/(8p))$.
On a successful embedding with $F\le b$, Lemma~\ref{hard:lem:degree} and Fact~\ref{hard:fact:bw} recover the desired value exactly. Since $b+1=(L-d)/2$, Markov's inequality gives
\begin{align}
\Pr[\text{run fails}]&\le\eta+\Pr[F>b]
\le\eta+\left(2+\frac1{2p}\right)\left(\delta+\frac1{8p}\right)\notag\\
&\le\frac12-\frac{13}{8p}-\frac7{16p^2}
-\eta\left(1+\frac1{2p}\right)
\le\frac12-\frac{13}{8p}.\label{hard:eq:runfailure}
\end{align}

\emph{Amplification and cost.} Repeat independently $J=\lceil2p^2\ln(2/\xi)\rceil$ times and return the most frequent output, counting $\bot$ as an output and breaking ties deterministically. Hoeffding's inequality bounds the probability that at most half the runs are correct by
$\exp[-2J(13/(8p))^2]\le\xi$.
There are $JL=O(gp^3\log(2/\xi))$ queries. The total decoder has polynomial cost on every transcript; forming $2^\kappa$ by repeated squaring gives the stated remaining cost.
\end{proof}
\begin{remark}
The degree $4g$ counts every gate, including those that are the identity in the worst case, because every gate is randomized.
With $g=O(n\log n)$ the interpolation degree is $O(n\log n)$.
\end{remark}

\subsection{Robustness to additive approximation}\label{hard:sec:robust}
The exact reduction extends to sufficiently small additive errors, including arbitrary incorrect answers on the allowed fraction of instances. We give a self-contained endpoint-decoding argument using a rational linear program. It is related to the robust interpolation methods of~\cite{Kondo2021Robustness,Quek2025Hamiltonian}, but requires no additional $\mathsf{NP}$ oracle. All circuit queries retain the real-RAM conventions above.

\subsubsection{A robust decoder at the hard endpoint}
\begin{lemma}[Endpoint decoding with outliers]\label{hard:lem:robustdecode}
Let $0<\Delta<1/2$ be rational, let $d,b,L$ be nonnegative integers with $L\ge d+2b+1$, and put $r=d+b\ge1$ and $x_i=i\Delta/L$ for $1\le i\le L$. Set $\lambda=(\Delta/(1-\Delta))^b$.
Suppose a real polynomial $R$ of degree at most $d$ obeys $|y_i-R(x_i)|\le\sigma$ at all but at most $b$ nodes, where $y_i$ and $\sigma>0$ are rational.
A deterministic algorithm, polynomial in the rational input length, returns $\widehat R_1$ such that
\[
|\widehat R_1-R(1)|\le2\sigma\lambda\frac{(2L/\Delta)^r}{r!}.
\]
If $L/r\le2$, this is at most $2\sigma K$, where $K=\lambda(12/\Delta)^r$. The algorithm terminates on every transcript, returning a default value if its linear program is infeasible.
\end{lemma}
\begin{proof}
\emph{Decoder.} The idea is to suppress the bad nodes with an unknown polynomial $E$, while fixing $E(1)=1$ so that the desired endpoint is unchanged. Solve the rational linear program
\begin{equation}\label{hard:eq:robustlp}
E(1)=1,\qquad |E(x_i)|\le\lambda,\qquad
|Q(x_i)-y_iE(x_i)|\le\sigma\lambda,
\end{equation}
with $\deg E\le b$ and $\deg Q\le r$. Return $Q(1)$, or a default if infeasible. Rational linear programming has polynomial bit complexity on every transcript.

\emph{Feasibility.} If $B$ is the bad-node set and $f=|B|\le b$, a witness is
\[
E_0(x)=x^{b-f}\prod_{j\in B}\frac{x-x_j}{1-x_j},\qquad Q_0=E_0R.
\]
On $[0,\Delta]$, $|E_0(x)|\le\Delta^{b-f}(\Delta/(1-\Delta))^f\le\lambda$, and $E_0(1)=1$. The residual vanishes at bad nodes and is at most $\sigma\lambda$ at good ones. This real witness suffices: a feasible rational linear program has a polynomial-length rational solution.

\emph{Accuracy.} For any feasible solution, $H=Q-ER$ has degree at most $r$ and $|H(x_i)|\le2\sigma\lambda$ at the good nodes. There are at least $L-b\ge r+1$ of them. Choose $r+1$ in increasing order, $z_0<\cdots<z_r$. Their spacing is at least $h=\Delta/L$, so the Lagrange weights at 1 obey
\[
|\ell_j(1)|\le\frac{h^{-r}}{j!(r-j)!},\qquad
\sum_{j=0}^r|\ell_j(1)|\le\frac{(2L/\Delta)^r}{r!}.
\]
Here the numerator factors are at most one. Interpolation bounds $|H(1)|=|Q(1)-R(1)|$ as claimed. If $L/r\le2$, use $r!\ge(r/\mathrm e)^r$ and $4\mathrm e<12$ to get $2\sigma K$. The decoder never needs to identify the good nodes.
\end{proof}

\subsubsection{The denominator and probability of success}
\begin{lemma}[Denominator control]\label{hard:lem:robustden}
Fix an integer $p\ge4$ and let $T=32gp$ and $M=(1+T^2)^{2g}$. For the Cayley path of Appendix~\ref{hard:sec:reduction},
\[
\Pr[D(0)>M]\le\frac1{8p},\qquad
1\le D(\theta)\le D(0)\quad(0\le\theta\le1).
\]
The probability bound holds conditional on each schedule, output and programmed hard circuit, before the Haar masks are drawn.
\end{lemma}
\begin{proof}
The absolute Cayley eigenvalue associated with an eigenphase $\varphi$ is $|\tan(\varphi/2)|$. Haar $U(2)$ has a uniform one-eigenphase marginal, so the expected number of eigenvalues with absolute Cayley value above $T$ is
\[
\frac4\pi\arctan(1/T)\le\frac4{\pi T}.
\]
A union bound over gates gives $\Pr[\max_e\|A_e\|>T]\le4g/(\pi T)\le1/(8p)$; independence of the two eigenphases is not needed. Off this event each factor $\det(I+(1-\theta)^2A_e^2)$ is at most $(1+T^2)^2$, proving the bound on $D$. Monotonicity in $1-\theta$ proves the remaining inequalities. Left invariance makes the argument valid for every fixed programmed gate.
\end{proof}

\begin{theorem}[Robust worst-to-average reduction]\label{hard:thm:robust}
Use the hypotheses of Theorem~\ref{hard:thm:reduction}, except that the oracle returns an estimate within additive error $0<\epsilon\le1$ on a fraction $1-\delta$ of the target instances, with fresh coins if randomized. Let $\kappa\le K_0$ be a known bound on the successful embedding exponents, and assume $\delta+\eta\le1/4-1/p$ for an integer $p\ge4$. Use $d,L,b,\Delta$ from~\eqref{hard:eq:decodingparameters}, and put
\begin{align*}
r&=d+b,& M&=(1+\Delta^{-2})^{2g},&
\lambda&=\left(\frac{\Delta}{1-\Delta}\right)^b,& K&=\lambda(12/\Delta)^r.
\end{align*}
For $0<\xi<1$, a randomized reduction in the stated real-RAM oracle model outputs an estimate $\widehat a$ of $a=|\bra0Q_0\ket0|^2$ with
\begin{equation}\label{hard:eq:robusterror}
\Pr[|\widehat a-a|\le4\,2^{K_0}MK\epsilon]\ge1-\xi.
\end{equation}
It uses $O(gp^3\log(2/\xi))$ oracle queries. Further work is polynomial in the embedding size, $g$, $p$, $K_0$, $\log(1/\epsilon)$ and $\log(2/\xi)$, for polynomial-length rational precision parameters. No additional $\mathsf{NP}$ oracle is used.
\end{theorem}
\begin{proof}
\emph{Algorithm.} Use the sampling and queries of Theorem~\ref{hard:thm:reduction}, including dummy queries on failed embeddings, and write $x_i=\theta_i$. Replace its exact decoder as follows. Clip each oracle reply to $[0,1]$, multiply by $D(x_i)$, clip to $[0,M]$, and round with error at most $M\epsilon$ to a rational grid of mesh at most $M\epsilon$. Call the result $y_i$.
After all queries, return zero if embedding failed or $D(0)>M$. Otherwise use Lemma~\ref{hard:lem:robustdecode} with $\sigma=2M\epsilon$ and return $2^\kappa\widehat R_1$, clipped to $[0,1]$.

\emph{Accuracy of a successful run.} When $D(0)\le M$, we have $R(x_i)\in[0,M]$. Clipping cannot increase error, so every accurate oracle reply gives
\[
|y_i-R(x_i)|\le D(x_i)\epsilon+M\epsilon\le2M\epsilon.
\]
If at most $b$ replies are inaccurate, Lemma~\ref{hard:lem:robustdecode} applies: $L\ge d+2b+1$ and $L/r<2$. Its bound is $4MK\epsilon$, or at most $4\,2^{K_0}MK\epsilon$ after scaling.

\emph{Success probability and cost.} Queries are drawn before testing the denominator, so the unconditional query-law and Markov arguments of Theorem~\ref{hard:thm:reduction} apply unchanged, now counting errors larger than $\epsilon$. The only extra failure event is $D(0)>M$. By~\eqref{hard:eq:runfailure} and Lemma~\ref{hard:lem:robustden},
\[
\Pr[\text{run fails}]\le\frac12-\frac{13}{8p}+\frac1{8p}
=\frac12-\frac3{2p}.
\]
Take the median of $J=\lceil2p^2\ln(2/\xi)\rceil$ independent runs. Hoeffding's inequality bounds the chance that at most half are accurate by $\exp[-2J(3/(2p))^2]\le\xi$.
Rounding uses binary search with exact comparisons and yields rational data of polynomial bit length. Thus the decoder has polynomial bit cost, while the queries retain the exact-real model. The query count is the same as in Theorem~\ref{hard:thm:reduction}.
\end{proof}

\subsubsection{Hardness at an explicit additive precision}
\begin{corollary}[High-precision average-case hardness]\label{hard:cor:robusthard}
For either ensemble in Theorem~\ref{hard:thm:pure} or Proposition~\ref{hard:thm:hybrid}, use its total gate count $g$, margin parameter $p=p(n)$, and the corresponding $M,K$ above. Set
\begin{equation}\label{hard:eq:robustprecision}
\epsilon_n=\frac{2^{-3n-5}}{MK}
=\frac{2^{-3n-5}(32gp-1)^b}{[1+(32gp)^2]^{2g}(384gp)^r}.
\end{equation}
The conclusions of both results remain valid if exact replies are replaced by estimates to additive error $\epsilon_n$, on the same promised fractions of instances. These remain real-RAM oracle reductions.
For a fixed success margin above $3/4$ and logarithmic total depth, the sufficient precision has scale $\epsilon_n=2^{-O(n\log^2n)}$.
\end{corollary}
\begin{proof}
For the counting circuit of Fact~\ref{hard:fact:sharpP}, $a_s=(1-s/2^m)^2$, and distinct adjacent possible values have separation at least $2^{-2m}$. Enlarge the polynomial padding in Theorem~\ref{hard:thm:pure} and Proposition~\ref{hard:thm:hybrid} to include $n\ge m$. Remark~\ref{hard:rem:pad} gives $\kappa\le n$, so Theorem~\ref{hard:thm:robust} and~\eqref{hard:eq:robustprecision} yield
\[
|\widehat a-a_s|\le4\,2^n MK\epsilon_n
=2^{-2n-3}\le2^{-2m-3}
\]
with the amplified success probability. This identifies the nearest allowed $a_s$ uniquely; binary search and comparison of the two neighboring candidates recover $s$ in polynomial time. Thus the count is recovered exactly from approximate replies.

Finally,
\[
K=(384gp)^d\left(\frac{384gp}{32gp-1}\right)^b
\le(384gp)^d13^b,
\]
so $\log_2(1/\epsilon_n)=O(n+g\log(gp)+gp)$. A fixed margin permits fixed $p$ for all sufficiently large sizes, including the embedding-failure allowance. With $g=O(n\log n)$ this gives the claimed scale. For an inverse-polynomial margin, the dependence on $p(n)$ in~\eqref{hard:eq:robustprecision} must be retained.
\end{proof}

\begin{remark}[What this robustness means]
The theorem tolerates numerical error in evaluating probabilities, as well as incorrect oracle replies. It is not a theorem of robustness to laboratory noise or constant-TV sampling error. Here $N^{-1}=2^{-n+O(\log n)}$, whereas the proved tolerance at logarithmic depth is much finer. Moreover, the circuit-query part remains an exact-real reduction; only the decoder has been given a finite rational implementation.
\end{remark}

\section{Embedding hard circuits into schedules}\label{hard:sec:embeddings}

We now supply the embedding procedure required by
Appendix~\ref{hard:sec:reduction}. Random matching schedules need two
routing estimates; a fixed routing prefix gives a simpler variant with
no embedding failures.

\subsection{Random matching schedules}\label{hard:sec:pure}

Let $T_n$ be the scan limit in~\eqref{eq:matchingtime}. The procedure $\mathcal W$ of Theorem~\ref{hard:thm:reduction} must realize $C_{Q_0}$ inside a typical schedule of $\Gp$, whose layers are i.i.d.\ uniform perfect matchings, and must do so for a uniformly random output string.
By Lemma~\ref{hard:lem:relabel} it suffices to fill the schedule with switches and solve two routing tasks.
\begin{itemize}[leftmargin=*,itemsep=1pt]
\item \emph{Meeting.} Before native layer $r$, the two modes of every pair of $P_r$ must sit on a common edge of the next layer. Only $m\le 2|E|$ pairs are involved, and we take $n$ polynomially larger than $m$.
\item \emph{Hiding.} At the end, the ones of $y_{Q_0}$ must be moved onto the ones of $z$. All $n/2$ occupied modes move, but only their \emph{set} matters.
\end{itemize}
Under a switch setting on a matching $M$, a mode at position $p$ moves to $M(p)$ if its edge carries $X$ and stays otherwise.

\subsubsection{Meeting windows}
\begin{lemma}[Meeting window]\label{hard:lem:meet}
Let $n$ be even and $m,\tau\ge1$. Let $(a_q,b_q)_{q\le m}$ be pairs of positions, all $2m$ distinct, and let $M_1,\dots,M_{\tau+1}$ be i.i.d.\ uniform perfect matchings of $[n]$.
A deterministic polynomial-time procedure either fails or outputs switch settings on $M_1,\dots,M_\tau$ with the following property: afterwards, for every $q$, the modes that started at $a_q$ and $b_q$ occupy the two endpoints of an edge of $M_{\tau+1}$.
The procedure fails with probability at most
\[
\varepsilon_{\rm meet}:=\frac{8m^2 2^{\tau}}{n-1}+m\exp\Big(-\frac{4^{\tau-1}}{2(n-1)}\Big).
\]
\end{lemma}
\begin{proof}
Call the $2m$ modes \emph{tokens} and write $s_u$ for the start of token $u$.
We grow \emph{frontiers} $A_u(\ell)\subseteq[n]$, the positions token $u$ may occupy after layer $\ell$, from $A_u(0)=\{s_u\}$.
Given the frontiers after layer $\ell-1$, call $p\in A_u(\ell-1)$ \emph{safe} if $M_\ell(p)\notin A_v(\ell-1)$ for every token $v\neq u$, and set $A_u(\ell)=\{p,M_\ell(p): p\ \text{safe}\}$.

\emph{Disjointness and paths.} The frontiers of distinct tokens stay disjoint. Suppose a position lay in $A_u(\ell)\cap A_v(\ell)$. It is an endpoint of an $M_\ell$-edge with an endpoint $p\in A_u(\ell-1)$ and an endpoint $r\in A_v(\ell-1)$. Disjointness at $\ell-1$ gives $p\neq r$, so $r=M_\ell(p)$, and then $p$ is not safe.
Every $x\in A_u(\ell)$ has a \emph{parent} $x'\in\{x,M_\ell(x)\}$ that is safe in $A_u(\ell-1)$.
Following parents from any $x_u\in A_u(\tau)$ gives a path $s_u=p_u(0),\dots,p_u(\tau)=x_u$ with $p_u(\ell)\in\{p_u(\ell-1),M_\ell(p_u(\ell-1))\}$ and $p_u(\ell-1)$ safe.
Put $X$ on the edge $\{p_u(\ell-1),M_\ell(p_u(\ell-1))\}$ if $p_u(\ell)\ne p_u(\ell-1)$, put $I$ on it otherwise, and put $I$ on all other edges.
Safety makes these edges distinct for distinct tokens and keeps other tokens off them. So every token follows its path, for any choice of endpoints.

\emph{Growth.} Let $W=\bigcup_vA_v(\ell-1)$ and let $X_u(\ell)$ count positions in $A_u(\ell-1)$ whose partner lies in $W$. All other positions double safely, so
\[
|A_u(\ell)|\ge2|A_u(\ell-1)|-2X_u(\ell),\qquad
\mathbb E X_u(\ell)\le\frac{2m\,4^{\ell-1}}{n-1}.
\]
The expectation bound follows because each partner is uniform among $n-1$ positions, while $|A_u(\ell-1)|\le2^{\ell-1}$ and $|W|\le2m2^{\ell-1}$. Iterating the deficit recurrence gives
\[
\mathbb E\big[2^\tau-|A_u(\tau)|\big]
\le\sum_{\ell=1}^\tau2^{\tau-\ell+1}\frac{2m\,4^{\ell-1}}{n-1}
\le\frac{2m\,4^\tau}{n-1}.
\]
Markov's inequality at deficit $2^{\tau-1}$, followed by a union bound over $2m$ tokens, shows that all frontiers have size $>2^{\tau-1}$ except with probability $8m^22^\tau/(n-1)$.

\emph{Meeting.} The frontiers depend only on $M_1,\dots,M_\tau$, so $M_{\tau+1}$ is independent of them.
Fix $q$ and put $A=A_{a_q}(\tau)$ and $B=A_{b_q}(\tau)$, which are disjoint.
Reveal the $M_{\tau+1}$-partners of the elements of $A$ one at a time, skipping elements that are already matched. At least $|A|/2$ partners get revealed.
While no partner has landed in $B$, all of $B$ is unmatched, so each new partner lands in $B$ with probability at least $|B|/(n-1)$.
Hence no edge joins $A$ and $B$ with probability at most $(1-|B|/(n-1))^{|A|/2}\le\exp(-4^{\tau-1}/(2(n-1)))$ when $|A|,|B|>2^{\tau-1}$.
The procedure picks, for each $q$, an edge $\{x_q,y_q\}$ of $M_{\tau+1}$ with $x_q\in A$ and $y_q\in B$, and uses the paths to $x_q$ and $y_q$.
A union bound over $q$ completes the proof.
\end{proof}

\subsubsection{Set routing}
Put $d_j=\binom nj-\binom n{j-1}$ (with $\binom n{-1}=0$), so that $\sum_{j=0}^kd_j=\binom nk$.
\begin{lemma}[Set routing]\label{hard:lem:setroute}
Let $n=2k\ge400$ and $T\ge20$. Let $Y\subseteq[n]$ with $|Y|=k$, and let $M_1,\dots,M_T$ be i.i.d.\ uniform perfect matchings.
Call $Z$ \emph{reachable} if some switch setting on $M_1,\dots,M_T$ maps the set $Y$ onto $Z$.
For $Z$ uniform among $k$-subsets and independent of the matchings,
\[
\Pr[Z\text{ is not reachable}]\le\tfrac12\Big(e^{n(3/4)^T}-1+2^{-n/2}\Big)^{1/2}.
\]
For reachable $Z$, a switch setting is found in polynomial time as an integral maximum flow.
\end{lemma}
\begin{proof}
\emph{Algorithm.} Use the time-expanded network with nodes $(p,\ell)$ for $p\in[n]$ and $0\le\ell\le T$, node capacities $1$, and arcs $(p,\ell-1)\to(p,\ell)$ and $(p,\ell-1)\to(M_\ell(p),\ell)$. A source feeds $\{(y,0):y\in Y\}$ and $\{(z,T):z\in Z\}$ drain to a sink.
An integral flow of value $k$ is a family of $k$ token trajectories with at most one token per node.
On an edge $\{p,p'\}$ of $M_\ell$ carrying exactly one token, put $X$ if the token crosses and $I$ otherwise; on the other edges put $I$. The occupied set then evolves as the flow does, so $Y$ is mapped onto $Z$.
If both endpoints are occupied, crossing both trajectories and leaving both in place give the same occupied set. Thus these edges need no swap: the task routes a set, not individually labelled particles.
Conversely, a switch setting mapping $Y$ onto $Z$ yields such a flow.

\emph{Reduction to a collision bound.} Fix the matchings and let $\mu$ be the law of the image of $Y$ when every switch is $X$ independently with probability $1/2$.
Every set in the support of $\mu$ is reachable. So for uniform $Z$,
\[
\Pr[Z\notin{\rm supp}\,\mu]\le\TV(\mu,{\rm Unif})\le\tfrac12\big(N\|\mu\|_2^2-1\big)^{1/2}
\]
by Cauchy--Schwarz.
By Jensen's inequality it suffices to prove
\begin{equation}\label{hard:eq:switchcoll}
\mathbb E_{M}\big[N\|\mu\|_2^2\big]=\sum_{j=0}^k d_j\beta_j^T,\qquad \sum_{j=1}^k d_j\beta_j^T\le e^{n(3/4)^T}-1+2^{-n/2},
\end{equation}
with $\beta_0=1$ and $\beta_j\in[0,1]$ defined below.

\emph{Spectral identity.} The permutation action on $V=\mathbb C^{\binom{[n]}k}$ decomposes as $V=\bigoplus_{j=0}^kV_j$, where $V_j\cong S^{(n-j,j)}$ has dimension $d_j$. Write $E_j$ for the corresponding orthogonal projectors.
For a matching $M$, averaging all its switch settings gives the orthogonal projector $P_M$ onto invariants of $H_M\cong\mathbb Z_2^k$. Each $P_M$ preserves every $V_j$.
Since the matching law is invariant under relabeling, Schur's lemma gives
\[
\left.\mathbb E_M P_M\right|_{V_j}=\beta_j I,
\qquad \beta_j=\frac{\operatorname{tr}(E_jP_M)}{d_j}\in[0,1].
\]
The trace is independent of $M$, and $\beta_0=1$. Because $P_M^2=P_M$, any fixed $v\in V_j$ satisfies
\[
\mathbb E_M\|P_Mv\|^2=\mathbb E_M\langle v,P_Mv\rangle=\beta_j\|v\|^2.
\]
Thus each independent layer multiplies the expected squared norm in $V_j$ by $\beta_j$. Starting at $\delta_Y$ and using transitivity, which gives $\|E_j\delta_Y\|^2=d_j/N$, we obtain
\[
\mathbb E\|\mu\|^2=\sum_{j=0}^k\beta_j^T\|E_j\delta_Y\|^2
=\frac1N\sum_{j=0}^kd_j\beta_j^T.
\]

\emph{The eigenvalues.} The trace ${\rm tr}(E_jP_M)$ is the dimension of the $H_M$-invariant subspace of $V_j$.
The permutation module on $j$-subsets is $\bigoplus_{i\le j}V_i$. Its $H_M$-invariant subspace has dimension equal to the number of $H_M$-orbits on $j$-subsets.
An orbit is determined by recording, for each edge of $M$, whether the subset meets it in $0$, $1$ or $2$ points. So the number of orbits is $N_j:=[x^j](1+x+x^2)^k$, and
\[
\beta_j=\frac{N_j-N_{j-1}}{d_j}.
\]
For a uniformly random $j$-subset $S$, let $s$ be the number of edges of $M$ that $S$ meets exactly once, and $D=(j-s)/2$ the number it contains.
The orbit of $S$ has size $2^s$, so $N_j/\binom nj=\mathbb E\,2^{-s}=2^{-j}\,\mathbb E\,4^D$.
Expand $4^D=\prod_e(1+3\cdot\mathbf 1[e\subseteq S])$ and use $\Pr[\text{$a$ given edges}\subseteq S]=\frac{j(j-1)\cdots(j-2a+1)}{n(n-1)\cdots(n-2a+1)}\le(j/n)^{2a}$. This gives
\[
\beta_j\le r_j\,2^{-j}\Big(1+\frac{3j^2}{4k^2}\Big)^k,\qquad r_j:=\frac{\binom nj}{d_j}=\frac{2k-j+1}{2k-2j+1}.
\]
For $j\le k/2$ we have $r_j\le e^{j/(k+1)}$ and $(1+3j^2/4k^2)^k\le e^{3j/8}$, so $\beta_j\le(e^{3/8+1/(k+1)}/2)^j\le(3/4)^j$ for $k\ge32$.
For $k/2\le j\le k$ we have $r_j\le k+1$. With $\alpha=j/k$, put $f(\alpha)=2^{-\alpha}(1+3\alpha^2/4)$. On $[1/2,1]$,
\[
(\log f)''(\alpha)=\frac{(3/2)(1-3\alpha^2/4)}{(1+3\alpha^2/4)^2}>0.
\]
Hence $f$ is at most its larger endpoint value: $f(1/2)=19/(16\sqrt2)<f(1)=7/8$.
So $\beta_j\le(k+1)(7/8)^k\le(9/10)^k$ for $k\ge200$.
Therefore
\[
\sum_{1\le j\le k/2}d_j\beta_j^T\le\sum_{j\ge1}\binom nj(3/4)^{jT}\le e^{n(3/4)^T}-1,
\qquad
\sum_{j>k/2}d_j\beta_j^T\le\binom{n}{k}(9/10)^{kT}\le(4\cdot0.9^{T})^k\le2^{-k}
\]
for $T\ge20$.
\end{proof}

\subsubsection{The embedding and the theorem}
Use the window lengths and embedding-failure bound from
\eqref{eq:hardness-depth}--\eqref{eq:hardness-embedding-error}.

\begin{lemma}[Worst cases inside $\Gp$]\label{hard:lem:embedpure}
Let $n=4B\ge400$ and $t\ge t_{\rm hard}(n)$.
Let $Q_0$ be a Clifford$+T$ circuit whose construction in Theorem~\ref{hard:thm:worst} uses $|E|$ blocks with $2|E|\le n^{1/16}$, padded as in Remark~\ref{hard:rem:pad}.
Then $\Gp$ admits a procedure $\mathcal W$ as in Theorem~\ref{hard:thm:reduction} with failure probability $\eta\le\eta_n^{\rm emb}$.
\end{lemma}
\begin{proof}
Split the $t$ layers of $S$ into four \emph{meeting windows} of $\tau_n^{\rm meet}+1$ layers, then a \emph{hiding window} of $T^{\rm h}_n$ layers, then the rest. Track the current position of every mode.
\begin{itemize}[leftmargin=*,itemsep=1pt]
\item In window $r\le4$, apply Lemma~\ref{hard:lem:meet} to the pairs of $P_r$, of which there are at most $2|E|\le n^{1/16}$; if $P_r$ is empty, set the whole window to $I$. Use the switches it returns on the first $\tau_n^{\rm meet}$ layers. In the last layer, place $F_{\pi(i)\pi(j)}(u)$ on the edge joining the current positions $\pi(i),\pi(j)$ for each gate $F_{ij}(u)$ of $P_r$, and $I$ elsewhere. With the ordered-pair convention this gate is $u$ or $XuX$.
\item In the hiding window, apply Lemma~\ref{hard:lem:setroute} with $Y$ the current positions of the ones of $y_{Q_0}$ and $Z$ the ones of $z$.
\item Set all remaining gates to $I$.
\end{itemize}
If every step succeeds, Lemma~\ref{hard:lem:relabel} gives $p_z(S,w)=2^{-\kappa}|\bra0Q_0\ket0|^2$, with $\kappa$ from Theorem~\ref{hard:thm:worst} and Remark~\ref{hard:rem:pad}. All gates lie in $\{I,X\}\cup\{R_\phi,XR_\phi X\}$.
Each window uses fresh matchings, independent of the positions it starts from and of $z$.
With $m\le n^{1/16}$ and $n^{5/8}\le2^{\tau_n^{\rm meet}}<2n^{5/8}$, Lemma~\ref{hard:lem:meet} bounds each meeting window's failure probability by $16n^{3/4}/(n-1)+n^{1/16}e^{-n^{1/4}/8}\le17n^{-1/4}+n^{1/16}e^{-n^{1/4}/8}$.
Since $T^{\rm h}_n\ge20$ and $n(3/4)^{T^{\rm h}_n}\le n^{1+4\log_2(3/4)}\le n^{-0.66}$, Lemma~\ref{hard:lem:setroute} bounds the hiding window's by $\frac12(3n^{-0.66})^{1/2}\le n^{-3/10}$.
A union bound gives $\eta\le\eta_n^{\rm emb}$.
\end{proof}

\begin{theorem}[Average-case hardness for the random-matching ensemble]\label{hard:thm:pure}
Let $p(n)\ge4$ be an integer-valued, polynomial-time computable, polynomially bounded margin parameter. Let $t(n)\ge t_{\rm hard}(n)$ be an integer-valued, polynomial-time computable, polynomially bounded depth on the admissible sizes.
Suppose that, for all sufficiently large $n$ divisible by $4$, an oracle $\mathcal O$ computes $p_z$ exactly on a fraction at least $\frac34+\frac1{p(n)}+\eta_n^{\rm emb}$ of instances of $\Gp$ at depth $t(n)$, with $z$ uniform.
Then $\mathsf P^{\#\mathsf P}\subseteq\mathsf{BPP}^{\mathcal O}$, for reductions in the real-RAM model.
At every such depth the ensemble also anticoncentrates: $t_{\rm hard}(n)\ge T_n\ge t^{\rm mag}_n$, so Corollary~\ref{cor:certified-ac} applies.
\end{theorem}
\begin{proof}
Given a Clifford$+T$ circuit $Q_0$ of size $s$, Theorem~\ref{hard:thm:worst} uses $|E|=\poly(s)$ blocks.
Let $n$ be the least multiple of $4$ with $n\ge\max(n_1,(2|E|)^{16})$, where $n_1\ge400$ is a constant beyond which the hypothesis on $\mathcal O$ holds. Then $n=\poly(s)$.
The eventual cutoff $n_1$ is fixed for the given oracle and can be built into the reduction; it need not be inferred from oracle answers. Computing $t(n)$ and $p(n)$ and sampling the schedule take polynomial time by the hypotheses.
Apply Theorem~\ref{hard:thm:reduction} to $Q_0$ with $\mathcal W$ from Lemma~\ref{hard:lem:embedpure}, $g=nt(n)/2$ and $\delta=\frac14-\frac1{p(n)}-\eta_n^{\rm emb}$. It computes $|\bra0Q_0\ket0|^2$ with probability $1-\xi$ using $\poly(s)$ oracle queries and time. Fact~\ref{hard:fact:sharpP} then recovers the encoded count by binary search. To simulate a polynomial-time computation with polynomially many adaptive $\#\mathsf P$ queries, amplify each call to inverse-polynomial failure and use a union bound, conditional on the preceding answers being correct. This gives $\mathsf P^{\#\mathsf P}\subseteq\mathsf{BPP}^{\mathcal O}$ in the stated real-RAM model.
For the last claim, the scan limit $T_n$ in~\eqref{eq:matchingtime} satisfies $t^{\rm mag}_n\le T_n\le3.81\log_2(n+1)+4\le6.5\log_2n+4\le t_{\rm hard}(n)$ for $n\ge4$.
\end{proof}

\begin{remark}
The margin above $3/4$ must exceed the specific inverse polynomial $\eta_n^{\rm emb}=\Theta(n^{-1/4})$, so an arbitrary $1/\poly$ margin does not suffice. The hybrid ensemble of Section~\ref{hard:sec:hybrid} avoids this, since there $\eta=0$.
With the present constants $\eta_n^{\rm emb}<1/4$ only for $n$ above about $2^{32.4}$, so the theorem is asymptotic.
The exponent $16$ in $n\ge(2|E|)^{16}$ and the constants in $t_{\rm hard}$ come from the Markov bound in Lemma~\ref{hard:lem:meet} and the $(3/4)^j$ envelope in Lemma~\ref{hard:lem:setroute}.
\end{remark}

\subsection{A deterministic routing prefix}\label{hard:sec:hybrid}

A fixed routing prefix removes the embedding failures. Let $n=2^r$,
$r\ge2$, and label the modes $0,\ldots,n-1$. The Bene\v{s} schedule
$\mathcal B_r$ pairs $i$ with $i\oplus2^{b_s}$ at stage $s$, where
$(b_s)=(r-1,r-2,\ldots,1,0,1,\ldots,r-1)$. Its $2r-1$ switch stages
realize every permutation, with settings computable in $O(n\log n)$
time~\cite{Benes1964,Waksman1968}. Let
$\mathcal A=\{(2i,2i+1):0\le i<n/2\}$.

The \emph{hybrid ensemble} $\Gsh$ has prefix
\[
 S_{\rm pre}=(\mathcal B_r,\mathcal A,\mathcal B_r,\mathcal A,
 \mathcal B_r,\mathcal A,\mathcal B_r,\mathcal A,\mathcal B_r),
\]
followed by $t$ independent uniform matching layers. Every gate, including
those in the prefix, is independently Haar distributed in $U(2)$, and
$z$ is uniform on $\Y$. The total depth is $10\log_2n-1+t$ and the gate
count is $g=n(10\log_2n-1+t)/2$.

\begin{samepage}
\begin{proposition}[Deterministic routing variant]\label{hard:thm:hybrid}
\label{hard:lem:ac}\label{hard:lem:embed}
The hybrid ensemble has the following properties.
\begin{enumerate}[label=(\roman*),itemsep=2pt]
\item Every hard circuit from Theorem~\ref{hard:thm:worst} using at most
$n/4$ blocks, padded as in Remark~\ref{hard:rem:pad}, embeds in every
schedule and at every requested output $z$. The embedding is computable
in polynomial time, uses gates in
$\{I,X\}\cup\{R_\phi,XR_\phi X:\phi\in\frac\pi4\mathbb Z\}$, and preserves
$p_z(S,w)=2^{-\kappa}|\bra0Q_0\ket0|^2$. Thus $\eta=0$.
\item For every $t\ge0$,
\[
 \mathbb E\Big[N\sum_zp_z^2\Big]
 \le(1+\varepsilon_t)\mathcal R_{\rm Haar}(n),\qquad
 \varepsilon_t=\sum_{j=1}^km_j\lambda_j^t
 \le e^{(n+1)(5/6)^t}-1.
\]
For $t\ge T_n$ from~\eqref{eq:matchingtime}, the collision is less than
$8$ and $\Pr_{(S,u),z}[p_z\ge1/(2N)]\ge1/32$.
\item Let $p(n)\ge4$ and $t(n)\ge0$ be integer-valued, polynomial-time
computable, polynomially bounded functions on powers of two $n\ge4$.
If an oracle computes $p_z$ exactly on a fraction at least
$3/4+1/p(n)$ of instances for all sufficiently large such $n$, then
$\mathsf P^{\#\mathsf P}\subseteq\mathsf{BPP}^{\mathcal O}$ by randomized
polynomial-time reductions in the real-RAM model. The additive-error
extension is Corollary~\ref{hard:cor:robusthard}.
\end{enumerate}
\end{proposition}
\end{samepage}

\begin{proof}
\emph{Embedding.} Before each of the four native layers, use one
Bene\v{s} block to place every interacting pair on an $\mathcal A$-edge.
Apply its gate on that edge, conjugating by $X$ if the orientation is
reversed, and use identities on idle edges. The last Bene\v{s} block
maps the occupied set of $y_{Q_0}$ to that of $z$; set every suffix gate
to $I$. Rearrangeability and Lemma~\ref{hard:lem:relabel} prove (i).

\emph{Collision.} Condition on the prefix and set
$\Phi=\Gamma(C_{\rm pre})\ket{\Psi_{\rm in}}$. Haar invariance gives
$\mathcal R_{\rm Haar}(n,\Phi)=\mathcal R_{\rm Haar}(n,\Psi_{\rm in})$.
The collision expansion~\eqref{eq:main-chain-reduction} holds for this
input, and $|c_j(\Phi)|\le m_j$ because its sector weights form a
probability vector and $|\phi_j|\le1$. This proves the first bound in (ii).
The binomial estimate used in~\eqref{eq:matchingtime} bounds
$\varepsilon_t$ and gives $\varepsilon_t\le1$ for $t\ge T_n$.
Use $\mathcal R_{\rm Haar}<4$ and Paley--Zygmund on $Np_z$, whose mean
over uniform $z$ is one, to obtain the stated event probability.

\emph{Hardness.} Pad a size-$s$ counting circuit to the smallest power
of two $n\ge\max(n_1,4|E(Q_0)|)$, where the fixed cutoff $n_1\ge4$
is beyond the oracle's eventual guarantee. Then $n=\poly(s)$.
Apply Theorem~\ref{hard:thm:reduction} with the embedding in (i),
$\eta=0$, and $\delta=1/4-1/p(n)$. Schedule sampling and all parameter
computations are efficient. Recover counts and amplify adaptive calls
as in Theorem~\ref{hard:thm:pure}, proving (iii).
\end{proof}

The prefix changes the input coefficients $c_j$, so the sharp magic-input
threshold does not transfer directly. The uniform bound above guarantees
anticoncentration once the suffix reaches $T_n\simeq5.5\ln n$.

\end{document}